\documentclass{article}
\usepackage[utf8]{inputenc}
\DeclareUnicodeCharacter{2265}{\ensuremath{\geq}}
\DeclareUnicodeCharacter{2264}{\ensuremath{\leq}}

\usepackage[style=authoryear-comp,backend=bibtex, natbib=true]{biblatex}

\usepackage{algpseudocode}
\usepackage{graphicx}
\usepackage{amsmath}
\usepackage{algorithm}
\usepackage{amsthm}
\usepackage{amsfonts}
\usepackage{setspace}
\usepackage{subcaption}
\usepackage{multirow}
\usepackage{multicol}
\usepackage{bm}
\usepackage{rotating}
\usepackage[shortlabels]{enumitem}
\usepackage{xcolor}
\usepackage{url}

\usepackage{pgf,tikz}
\usetikzlibrary{arrows,shapes.arrows,shapes.geometric,shapes.multipart,decorations.pathmorphing,positioning,shapes.swigs}

\newtheorem{theorem}{Theorem}

\newtheorem{proposition}{Proposition}
\newtheorem{lemma}{Lemma}

\newtheorem{corollary}{Corollary}
\newtheorem{example}{Example}

\newtheorem{assumption}{Assumption}

\newtheorem*{example*}{\examplenumber}
\providecommand{\examplenumber}{}
\makeatletter
\newenvironment{newexample}[1]
{%
	\renewcommand{\examplenumber}{Example #1$'$}%
	\begin{example*}%
		\protected@edef\@currentlabel{#1$'$}%
	}
	{%
	\end{example*}
	}
	
\newtheorem*{assumption*}{\assumptionnumber}
\providecommand{\assumptionnumber}{}
\makeatletter
\newenvironment{newassumption}[1]
{%
	\renewcommand{\assumptionnumber}{Assumption #1$'$}%
	\begin{assumption*}%
		\protected@edef\@currentlabel{#1$'$}%
	}
	{%
	\end{assumption*}
	}

\newtheorem*{theorem*}{\theoremnumber}
\providecommand{\theoremnumber}{}
\makeatletter
\newenvironment{newt}[1]
{%
	\renewcommand{\theoremnumber}{Theorem #1$'$}%
	\begin{theorem*}%
		\protected@edef\@currentlabel{#1$'$}%
	}
	{%
	\end{theorem*}
	}

\newcommand{\expit}{\text{expit}}

\newcommand{\indep}{\perp \!\!\! \perp}
\newcommand{\nindep}{\not\!\perp \!\!\! \perp}

\newcommand{\blind}{1}

\begin{document}

\def\spacingset#1{\renewcommand{\baselinestretch}%
{#1}\small\normalsize} \spacingset{1}

%%%%%%%%%%%%%%%%%%%%%%%%%%%%%%%%%%%%%%%%%%%%%%%%%%%%%%%%%%%%%%%%%%%%%%%%%%%%%%

\if1\blind
{
  \title{\bf Double Negative Control Inference in Test-Negative Design Studies of Vaccine Effectiveness}
  \author{Kendrick Qijun Li\thanks{
    The authors gratefully acknowledge NIH grants R01AI27271, R01CA222147, R01AG065276, R01GM139926. }\hspace{.2cm}\\
    Department of Biostatistics, University of Michigan;\\
    Xu Shi \\
    Department of Biostatistics, University of Michigan;\\
    Wang Miao\\
    Department of Probability and Statistics, Peking University;\\
    Eric Tchetgen Tchetgen\\
    Department of Statistics and Data Science, The Wharton School,\\ University of Pennsylvania}
  \maketitle
} \fi

\if0\blind
{
  \bigskip
  \bigskip
  \bigskip
  \begin{center}
    {\LARGE\bf Double Negative Control Inference in Test-Negative Design Studies of Vaccine Effectiveness}
   %%Estimating vaccine effectiveness in a test negative design with negative control variables to reduce unmeasured confounding Confounding and selection biases adjustment in test negative design studies of vaccine effectivenesswith double negative controls(Bias correction for test-negative design using double negative controls)}
\end{center}
  \medskip
} \fi

\bigskip
\begin{abstract}
The test-negative design (TND) has become a standard approach to evaluate vaccine effectiveness against the risk of acquiring infectious diseases in real-world settings, such as Influenza, Rotavirus, Dengue fever, and more recently COVID-19. In a TND study, individuals who experience symptoms and seek care are recruited and tested for the infectious disease which defines cases and controls. Despite TND's potential to reduce unobserved differences in healthcare seeking behavior (HSB) between vaccinated and unvaccinated subjects, it remains subject to various potential biases. First, residual confounding bias may remain due to unobserved HSB, occupation as healthcare worker, or previous infection history. Second, because selection into the TND sample is a common consequence of infection and HSB, collider stratification bias may exist when conditioning the analysis on testing, which further induces confounding by latent HSB. In this paper, we present a novel approach to identify and estimate vaccine effectiveness in the target population by carefully leveraging a pair of negative control exposure and outcome variables to account for potential hidden bias in TND studies. We illustrate our proposed method with extensive simulation and an application to study COVID-19 vaccine effectiveness using data from the University of Michigan Health System.

\end{abstract}

\vfill

\newpage
\spacingset{1.45} % DON'T change the spacing!
\newpage
\tableofcontents

\clearpage
\section{Introduction}
\label{sec:intro}

\subsection{Text-negative design studies of vaccine effectiveness}
The test-negative design (TND) has become a standard approach to evaluate real-world vaccine effectiveness (VE) against the risk of acquiring infections diseases~\citep{jackson2017influenza,flannery2019influenza,rolfes2019effects,chung2020effects,tenforde2021influenza}. In an outpatient Influenza VE test-negative design, for example, symptomatic individuals seeking care and meeting eligibility criteria are enrolled and their Influenza virus infection status is subsequently confirmed via a laboratory test. VE against flu infection is then measured by comparing the prevalence of vaccination between the test-positive ``cases'' and test-negative ``controls''~\citep{jackson2013test,jackson2017influenza}. Besides Influenza, the TND and its variants have also been applied to study VE against pneumococcal disease~\citep{broome1980pneumococcal}, dengue~\citep{anders2018awed,utarini2021efficacy}, rotavirus~\citep{boom2010effectiveness,schwartz2017rotavirus}, and other infectious diseases.  Recently, the TND has increasingly been used in post-licensure evaluation of COVID-19 VE~\citep{patel2020postlicensure,dean2021COVID,hitchings2021effectiveness,thompson2021effectiveness,dagan2021bnt162b2,olson2022effectiveness}.  

Test-negative designs are believed to reduce unmeasured confounding bias due to healthcare seeking behavior (HSB), whereby care seekers are more likely to be vaccinated, have healthier behaviors that reduces the risk of infection, and get tested when ill~\citep{jackson2006evidence,shrank2011healthy}. By restricting analysis to care seekers who are tested for the infection in view (e.g. Influenza or COVID-19), the vaccinated and unvaccinated are more likely to share similar HSB and underlying health characteristics. Misclassification of infection status is also reduced because the analysis is restricted to tested individuals~\citep{jackson2013test}. 

Sullivan et al. (2016) used directed acyclic graphs (DAG) to illustrate the rationale behind TND in the context of evaluating VE against Influenza infection, as shown in Figures~\ref{fig:dag-tnd}(a) and (b). We denote Influenza vaccination status by $A$ and Influenza infection by $Y$, so that the arrow $A\rightarrow Y$ represents VE against flu infection. Selection into the TND study sample, denoted by $S$, is triggered by a subject experiencing flu-like symptoms or acute respiratory illness, seeking care at clinics or hospitals, and getting tested for Influenza infection, hence the  $Y\rightarrow S$ edge. Healthcare seeking behavior, denoted by HSB, may affect $S$, $A$, and $Y$ because subjects with certain health seeking proclivities may be more likely to seek care, take annual flu shots, and participate in healthy and preventative behaviors. The above variables are subject to  effects of other clinical or demographic factors, such as age, season and high-risk conditions, included in Figure~\ref{fig:dag-tnd} as confounders $X$. The TND assumes that by restricting recruitment to care seekers, the study subjects have identical healthcare seeking behavior; in other words, conditioning the analysis on $S=1$ necessarily leads to HSB$=1$, which blocks the effects of HSB (Figure~\ref{fig:dag-tnd}(b)). The effects of $X$ are further adjusted for by including these factors in a logistic regression model or by inverse probability weighting~\citep{bond2016regression,thompson2021effectiveness}. 

However, the TND remains subject to potential hidden bias. First, the assumption that all study subjects seeking care are lumped into a single category HSB$=1$ may be unrealistic. It may be more realistic that HSB is not a deterministic function of $S$ and remains a source of confounding bias even after conditioning on $S$. Furthermore, there might be other mismeasured or unmeasured confounders, denoted as $U$. For example, healthcare workers are at increased risk of flu infection due to higher exposure to flu patients and are more likely to seek care and receive vaccination due to health agency guidelines~\citep{black2018influenza}. Previous flu infection history may also be a source of confounding if it alters the likelihood of vaccination and care seeking, while also providing immunity against circulating strains~\citep{sullivan2016theoretical,krammer2019human}. These unmeasured or mismeasured potential sources of confounding, if not properly accounted for, can result in additional confounding bias, as illustrated in Figure~\ref{fig:dag-tnd}(c). Finally,  collider stratification bias is likely present due to conditioning on $S$, which is a common consequence of HSB, other risk factors $(X,U)$, and Influenza infection $Y$~\citep{lipsitch2016observational}. That is, conditioning on $S$
unblocks the  backdoor path $A\leftarrow (X, U, HSB)\rightarrow S\leftarrow Y$, which would in principle be blocked if study subjects had identical levels of HSB and other risk factors~\citep{sullivan2016theoretical}.

\begin{figure}[!htbp]
		\centering
\resizebox{0.88\textwidth}{!}{
\begin{tabular}{cc}	
	\begin{minipage}{0.5\textwidth}%\vskip 2.5em
		\centering
		\begin{tikzpicture}
			
			\tikzset{line width=1pt,inner sep=5pt,
				swig vsplit={gap=3pt, inner line width right=0.4pt},
				ell/.style={draw, inner sep=1.5pt,line width=1pt}}

			\node[shape = circle, ell] (A) at (-2.5, 0) {$A$};
			\node (Aname) at (-2.9, -0.5) {Flu Vac};
			
			\node[shape=ellipse,ell] (Y) at (0,0) {$Y$};
			\node (Yname) at (0,-0.5) {Infection};
			
			\node[shape=ellipse,ell] (HS) at (0,2.1) {$HSB$};
			\node (HSname) at (1.3, 2.67) {\begin{tabular}{l}Healthcare seeking\\behavior\end{tabular}};

			\node[shape=ellipse,ell] (X) at (-2.5,2.1) {$X$};
			\node (Xname) at (-2.9,2.67) {
				\begin{tabular}{l}	Measured\\confounders \end{tabular}
			};

			\node[shape=ellipse,ell] (S) at (1.5,1.05) {$S$};
			\node (Sname) at (2,1.5) {Selected};
			
			\draw[-stealth,line width=0.5pt](A) to (Y);
			
			\foreach \from/\to in {Y/S, HS/S, X/Y, HS/A, X/A, X/HS, HS/Y}
			\draw[-stealth, line width = 0.5pt] (\from) -- (\to);
			
		\end{tikzpicture}
	\end{minipage}
	&
	\begin{minipage}{0.5\textwidth}
		\centering
		\begin{tikzpicture}
			
			\tikzset{line width=1pt,inner sep=5pt,
				swig vsplit={gap=3pt, inner line width right=0.4pt},
				ell/.style={draw, inner sep=1.5pt,line width=1pt}}

			\node[shape = circle, ell] (A) at (-2.5, 0) {$A$};
			\node (Aname) at (-2.9, -0.5) {\textcolor{white}{Flu Vac}};
			
			\node[shape=ellipse,ell] (Y) at (0,0) {$Y$};
			\node (Yname) at (0,-0.5) {\textcolor{white}{Infection}};

			\node[shape=rectangle,ell] (HS) at (0,2.1) {$HSB=1$};
			\node (HSname) at (1.3, 2.67) {\begin{tabular}{l}\textcolor{white}{Healthcare seeking}\\\textcolor{white}{behavior}\end{tabular}};
			
			\node[shape=ellipse,ell] (X) at (-2.5,2.1) {$X$};
			\node (Xname) at (-2.9,2.67) {
				\begin{tabular}{l}	\textcolor{white}{Measured}\\\textcolor{white}{confounders} \end{tabular}
			};

			\node[shape=rectangle,ell] (S) at (1.5,1.05) {$S=1$};
			\node (Sname) at (2,1.5) {\textcolor{white}{Selected}};
			
			\draw[-stealth,line width=0.5pt](A) to (Y);
			
			\foreach \from/\to in {Y/S,  X/Y, X/A, X/HS}
			\draw[-stealth, line width = 0.5pt] (\from) -- (\to);
			
		\end{tikzpicture}
	\end{minipage}
	
	\\
	\textbf{(a)} & \textbf{(b)}\\
	\begin{minipage}{0.5\textwidth}\vskip 2.5em
		
		\begin{tikzpicture}
			
			\tikzset{line width=1pt,inner sep=5pt,
				swig vsplit={gap=3pt, inner line width right=0.4pt},
				ell/.style={draw, inner sep=1.5pt,line width=1pt}}

			\node[shape = circle, ell] (A) at (-2.5, 0) {$A$};
			\node[shape=ellipse,ell] (Y) at (0,0) {$Y$};
			\node[shape=ellipse,ell] (HS) at (0,2.1) {$HSB$};
			\node (HSname) at (1.3, 2.67) {\begin{tabular}{l}Healthcare seeking\\ behavior\end{tabular}};

			\node[shape=ellipse,ell] (X) at (-2.5,2.1) {$X$};
			\node (Xname) at (-2.9,2.67) {
				\begin{tabular}{l}	\textcolor{white}{Measured}\\\textcolor{white}{confounders} \end{tabular}

			};
			
			\node[shape=ellipse,ell] (HCW) at (-1.25, 3) {$U$};
			
			\node (HCWname) at (-1.25, 3.6) {
				\begin{tabular}{l}
					Healthcare worker;\\
					Previous infection
				\end{tabular}
			};

			\node[shape=rectangle,ell] (S) at (1.5,1.05) {$S=1$};

			\draw[-stealth,line width=0.5pt](A) to (Y);
			
			\foreach \from/\to in {Y/S,  X/Y, X/A, X/HS, HS/A, HS/Y, HS/S, HCW/A, HCW/Y, HCW/HS} \draw[-stealth, line width = 0.5pt] (\from) -- (\to);

		\end{tikzpicture}
	\end{minipage}
	
	&
	
	\begin{minipage}{0.5\textwidth}\vskip 2.5em
		\centering
		\begin{tikzpicture}
			
			\tikzset{line width=1pt,inner sep=5pt,
				swig vsplit={gap=3pt, inner line width right=0.4pt},
				ell/.style={draw, inner sep=1.5pt,line width=1pt}}

			\node[shape = circle, ell] (A) at (-2.5, 0) {$A$};
			\node[shape=ellipse,ell] (Y) at (0,0) {$Y$};
			\node[shape=ellipse,ell] (U) at (0,2) {$\mathbf{U}, X$};
			
			\node (Uname) at (0, 3) {
				\begin{tabular}{l}
					U=\{healthcare worker, \\
					previous infection, HSB\}
				\end{tabular}
			};

			\node (Zname) at (-2, 2.67) {\begin{tabular}{l}\textcolor{white}{NCE:}\\ \textcolor{white}{Previous/other Vac}\end{tabular}};
			
			\node[shape=rectangle,ell] (S) at (2.5,0) {$S=1$};

			\draw[-stealth,line width=0.5pt](A) to (Y);

			\foreach \from/\to in {Y/S, U/S, U/Y, U/A}
			\draw[-stealth, line width = 0.5pt] (\from) -- (\to);
			
		\end{tikzpicture}
	\end{minipage}

	\\
	\textbf{(c)} & \textbf{(d)}\\
	\begin{minipage}{0.5\textwidth}\vskip 2.5em
		\begin{tikzpicture}
			
			\tikzset{line width=1pt,inner sep=5pt,
				swig vsplit={gap=3pt, inner line width right=0.4pt},
				ell/.style={draw, inner sep=1.5pt,line width=1pt}}

			\node[shape = circle, ell] (A) at (-2.5, 0) {$A$};
			\node[shape=ellipse,ell] (Y) at (0,0) {$Y$};
			
			\node[shape=ellipse,ell] (U) at (0,2) {$U, X$};
			
			\node[shape=ellipse,ell] (W) at (2.5, 2) {$W$};
			\node (Wname) at (2, 2.67) {\begin{tabular}{l}NCO: \\Other Infection\end{tabular}};
			
			\node[shape=ellipse,ell] (Z) at (-2.5, 2) {$Z$};
			\node (Zname) at (-2, 2.67) {\begin{tabular}{l}NCE:\\ Previous/other Vac\end{tabular}};
			
			\node[shape=rectangle,ell] (S) at (2.5,0) {$S=1$};
			
			\draw[-stealth, line width=0.5pt, bend right,color=white](A) to (S);
			
			\draw[-stealth,line width=0.5pt](A) to (Y);
			
			\draw[dashed,stealth-stealth, line width=0.5pt](A) to (Z);
			
			\draw[dashed, -stealth, line width = 0.5pt] (W) to (Y); 
			
			\draw[dashed, -stealth, line width = 0.5pt] (W) to (S);
			
			\foreach \from/\to in {Y/S, U/S, U/Y, U/A, U/W, U/Z}
			\draw[-stealth, line width = 0.5pt] (\from) -- (\to);
			
		\end{tikzpicture}
	\end{minipage}&
	\begin{minipage}{0.5\textwidth}\vskip 2.5em
		\begin{tikzpicture}
			
			\tikzset{line width=1pt,inner sep=5pt,
				swig vsplit={gap=3pt, inner line width right=0.4pt},
				ell/.style={draw, inner sep=1.5pt,line width=1pt}}

			\node[shape = circle, ell] (A) at (-2.5, 0) {$A$};

			\node[shape=ellipse,ell] (Y) at (0,0) {$Y$};
			
			\node[shape=ellipse,ell] (U) at (0,2) {$U, X$};
			
			\node[shape=ellipse,ell] (W) at (2.5, 2) {$W$};
			\node (Wname) at (2, 2.67) {\begin{tabular}{l}NCO: \\Other Infection\end{tabular}};
			
			\node[shape=ellipse,ell] (Z) at (-2.5, 2) {$Z$};
			\node (Zname) at (-2, 2.67) {\begin{tabular}{l}NCE:\\ Previous/other Vac\end{tabular}};
			
			\node[shape=rectangle,ell] (S) at (2.5,0) {$S=1$};
			
			\draw[dashed,-stealth, line width=0.5pt](A) to (Z);
			
			\draw[-stealth, line width=0.5pt, bend right](A) to (S);
			
			\draw[-stealth,line width=0.5pt](A) to (Y);
			
			\draw[-stealth,line width=0.5pt, dashed](Y) to (W);
			
			\foreach \from/\to in {Y/S, U/S, U/Y, U/A, U/W, U/Z}
			\draw[-stealth, line width = 0.5pt] (\from) -- (\to);
			
		\end{tikzpicture}
	\end{minipage}
	
	\\
	\textbf{(e)}&\textbf{(f)}\\
	
\end{tabular}}

\caption{\label{fig:dag-tnd} Causal relationships of variables in a test-negative design. \cite{sullivan2016theoretical} used (a) to illustrate the causal relationship between variables in a test-negative design in the general population, and used (b) to illustrate the assumption implicit in the common approach to estimate VE from the study data that study subjects have identical healthcare seeking behavior (HSB)~\citep{sullivan2016theoretical}. (c) shows that if $HSB$ remains partially unobserved, then the backdoor paths $A\leftarrow HSB\rightarrow Y$ and $A\leftarrow HSB\rightarrow S=1\leftarrow Y$ indicate unmeasured confounding bias and selection bias, respectively. Other unmeasured confounders, such as occupation as a healthcare worker and previous infection, open additional backdoor paths between $A$ and $Y$ and result in additional confounding bias. (d) shows a simplified DAG from (c) that combines the unmeasured confounders into a single variable $U$.  (e) illustrates our approach to estimate VE leveraging negative control exposure $Z$ and outcome $W$. Dashed arrows indicate effects that are not required. (f) shows a scenario with the $A\rightarrow Y$ arrow where the causal odds ratio can still be identified under additional assumptions.}
\end{figure}
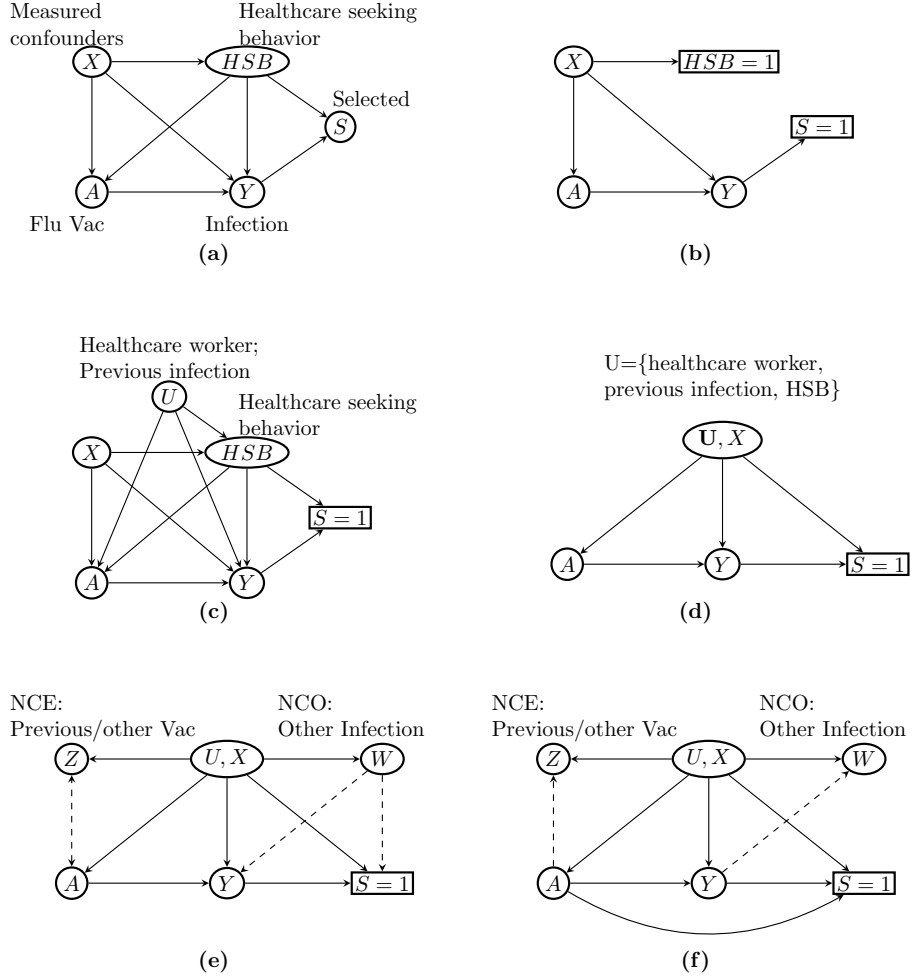

Accounting for these potential sources of bias is well known to be challenging, and potentially infeasible without  additional assumptions or data. This can be seen in Figure~\ref{fig:dag-tnd}(d), which is a simplified version of Figure~\ref{fig:dag-tnd}(c) where the unmeasured confounders $U$ include individuals' occupation as a healthcare worker, previous flu infection, HSB, and so on. Figure~\ref{fig:dag-tnd}(d) indicates that the unmeasured confounders $U$ induce both confounding bias through the path $A\leftarrow U\rightarrow Y$ and collider stratification bias through the path $A\leftarrow U\rightarrow S\leftarrow Y$. In presence of both unmeasured confounding and collider bias, causal bounds may be available ~\citep{gabriel2020causal} but likely too wide to be informative; causal identification in TND therefore remains to date an important and outstanding open problem in the causal inference literature which we aim to resolve.

\subsection{Negative control methods}
In recent years, negative control variables have emerged as powerful tools to detect, reduce and potentially correct for unmeasured confounding bias~\citep{lipsitch2010negative,miao2018identifying,shi2020selective}. The framework requires that at least one of two types of negative control variables are available which  are \textit{a priori} known to satisfy certain conditions: a negative control exposure (NCE) known to have no   direct effect on the primary outcome; or a negative control outcome (NCO), known not to be an effect of the primary exposure. Such negative control variables are only valid and therefore useful to address unmeasured confounding in a given setting to the extent that they are subject to the same source of confounding as the exposure-outcome relationship of primary interest. Thus, the observed association between a valid NCE and the primary outcome (conditional on the primary treatment and observed covariates) or that between a valid NCO and the primary exposure can indicate the presence of residual confounding bias. For example, in a cohort  study to investigate flu VE against hospitalization and death among seniors, to detect the presence of confounding bias due to underlying health characteristics, \citet{jackson2006evidence} used hospitalization/death before and after the flu season as NCOs and found that the association between flu vaccination and hospitalization was virtually the same before and during the flu season, suggesting that the lower hospitalization rate observed among vaccinated seniors versus unvaccinated seniors was partially due to healthy user bias. 

Recently, new causal methods have been developed to not only detect residual confounding when present, but also to potentially de-bias an observational estimate of a treatment causal effect in the presence of unmeasured confounders when both an NCE and an NCO are available, referred to as the double negative control~\citep{miao2018identifying,shi2020multiply,tchetgen2020introduction}. In this recent body of work, the double negative control design was extended in several important directions including settings in which proxies of treatment and outcome confounding routinely measured in well designed observational studies may be used as negative control variables, a framework termed \textit{proximal causal inference};  
longitudinal settings where one is interested in the joint effects of time-varying exposures~\citep{ying2021proximal}, potentially subject to both measured and unmeasured confounding by time-varying factors; and in settings where one aims to  estimate direct and indirect effects in mediation analysis subject to unmeasured confounding or unmeasured mediators~\citep{dukes2021proximal,ghassami2021proximal}. Additional recent papers in this fast-growing literature include \citet{qi2021proximal}, \citet{liu2021regression}, \citet{egami2021identification}, \citet{kallus2021causal}, \citet{imbens2021controlling}, \citet{deaner2018proxy}, \citet{deaner2021many}, \citet{ghassami2021minimax}, \citet{mastouri2021proximal} and \citet{ghassami2022combining}. Importantly, existing identification results in negative control and proximal causal inference literature has been restricted to i.i.d settings~\citep{miao2018confounding} and time series settings~\citep{shi2021theory}, and to date, to the best of our knowledge, outcome-dependent sampling settings such as TND, have not been considered, particularly one where confounding and selection bias might co-exist.

\subsection{Outline}
The rest of the paper is organized as followed: we introduce notation and the identification challenge in view in Sections~\ref{sec:model}. Next we develop the identification strategy and describe a new debiased estimator under a double negative control TND study in Section~\ref{sec:semiparam-model}-\ref{sec:sum}, assuming (1) homogeneous VE across strata defined by all measured and unmeasured confounders and (2) no direct effect of vaccination on selection into the TND sample. In Section~\ref{sec:eff-mod}, we relax the homogeneous VE assumption and describe identification and estimation allowing for VE to depend on observed covariates. In Section~\ref{sec:a-s-arrow}, we relax the assumption of no direct effect of vaccination on selection and introduce the assumptions under which our VE estimator is unbiased on the odds ratio scale. In Section~\ref{sec:sim}, we demonstrate the performance of our method with simulation. In Section~\ref{sec:application}, the approach is further illustrated in an application to estimate  COVID-19 VE against infection in a TND study nested within electronic health records from University of Michigan Health System.  We then conclude with a discussion in Section~\ref{sec:discussion}.

\section{Method}

\subsection{Preliminary:  estimation under no unmeasured confounding and no selection bias\label{sec:model}}
To fix ideas, we first review estimation  assuming all confounders $(U,X)$ are fully observed and the study sample is randomly drawn (rather than selected by testing) from source population, referred to as the ``target population''. 
That is, we observe data on $(A, Y, U, X)$ which are independent and identically distributed in the target population. For each individual, we write $Y(a)$ as the binary potential infection outcome had, possibly contrary to fact, the person's vaccination status been $A=a$, $a=0, 1$. Our goal is to provide identification and estimation strategies for the  causal risk ratio (RR)
defined as $RR= E[Y(1)]/E[Y(0)]$.
Let $\beta_0$ denote the log causal RR, i.e., $RR=\exp(\beta_0)$.
Following ~\citet{hudgens2006causal} and \citet{struchiner2007randomization}, we define VE as one minus the causal RR: $VE=1-RR=1-\exp(\beta_0)$. The potential outcomes and the observed data are related through the following assumptions:

\begin{assumption}\label{assump:ident}(Identification conditions of mean potential outcomes).

\begin{enumerate}[(a)]
    \item (Consistency) $Y(a)=Y$ if $A=a$ almost surely for $a=0,1$;
    \item (Exchangeability) $A\indep Y(a)|U, X$ for $a=0,1$.
    \item (Positivity)  $0<P(A=a|U, X)$ almost surely $a=0,1$.
\end{enumerate}
\end{assumption}

Assumption~\ref{assump:ident}(a) states that the infection status of a subject with vaccination status $A=a$ is equal to the corresponding potential outcome $Y(a)$. This further requires that the treatment is sufficiently well-defined and a subject's potential outcome is not affected by the treatment of other subjects~\citep{cole2009consistency}. Assumption~\ref{assump:ident}(b) states that treatment is exchangeable within strata of $(U,X)$, i.e. there is no unmeasured confounding given $(U,X)$. We develop methods that allow $U$ to be unmeasured in Section~\ref{sec:bridge-approach}. 
Assumption~\ref{assump:ident}(c) states that for all realized values of $(U,X)$ there is at least one individual with an opportunity to get treatment $a=0,1$.

Let $Q(A=a, U, X)=1/P(A=a|U, X)$ denote the inverse of the probability of vaccination status $A=a$ given confounders \citep{rosenbaum1987model}. Under Assumption~\ref{assump:ident}, it is well known that, if $U$ were observed, the mean potential outcome for treatment $a$ in the general population can be identified by inverse probability of treatment weighting (IPTW):
\begin{equation}
    E[Y(a)]= E\left[I(A=a)Q(A=a,U, X)Y\right],\label{eq:iptw}
\end{equation}
for $a=0,1$.
Therefore, the log causal RR $\beta_0$  satisfies the following equation
$$
E\left[{I(A=1)Q(A=1,U, X)Y\exp(-\beta_0)}\right] - E\left[{I(A=0)Q(A=0,U, X)Y}\right]=0.
$$
Equivalently, we have 
\begin{equation}\label{eq:ee-general-population}E[V_0(A, Y, U, X;\beta_0)]=0\end{equation} for the  unbiased estimating function, where 
$V_0(A, Y, U, X;\beta)=(-1)^{1-A}Q(A, U, X)Y\exp(-\beta A)$.

\subsection{Tackling selection bias under a semiparametric risk model}\label{sec:semiparam-model}
Next, consider a TND study for which data $(A,Y,X,U)$ is observed only for the tested individuals with $S=1$. Because $S$ is impacted by other factors such as infection, the estimating function $V_0(A, Y, U, X;\beta_0)$ may not be unbiased with respect to the study sample; i.e. \begin{equation}\label{eq:ee-general-sample}E[V_0(A, Y, U, X;\beta_0)|S=1]\neq 0\end{equation}  without another assumption about the selection process into the TND sample. 

For a study sample of size $n$ from a TND, we denote the $i$-th study subject's variables as $(A_i, Y_i, U_i, X_i)$, $i=1,\dots, n$. For generalizability, we first make the key assumption that vaccination $A$ is unrelated to selection $S$ other than through a subject's infection status $Y$ and confounders $(U,X)$. 

\begin{assumption}[Treatment-independent sampling]\label{assump:trt-indep-samp}  $S\indep A|Y, U, X$.
\end{assumption}
In the test-negative design, this assumption requires that an individual's decision to seek care and get tested only depends on the presence of symptoms and his/her underlying behavioral or socioeconomic characteristics, including HSB (contained in $(U,X)$); a person's vaccination status does not directly Influence their selection process. The DAGs in Figure 1(a)-(e) in fact encode this conditional independence condition. We will relax this assumption in Section~\ref{sec:a-s-arrow}.

\begin{assumption}[No effect modification by a latent confounder]
\label{assump:no-eff-mod}
For $a=0,1$,
\begin{equation}\label{eq:rr-model}
    P(Y=1|A=a, U, X)=\exp(\beta_0a)g(U, X)
\end{equation}
where $g(U, X)$ is an unknown function restricted by $0\leq  P(Y=1|A, U, X)\leq 1$ almost surely. 
\end{assumption}

Assumption~\ref{assump:no-eff-mod} defines a  semiparametric multiplicative risk model which posits that vaccine effectiveness, measured on the RR scale, is constant across $(U,X)$ strata in the target population. In other words, the effect of vaccination $A$ on the risk of infection $Y$ is not modified by confounders $U,X$. In Section~\ref{sec:eff-mod}, we will relax the assumption to allow for effect modification by measured confounders $X$. Infection risk for control subjects $P(Y=1|A=0, U, X)=g(U, X)$ is left unspecified and thus defines the nonparametric component of the model.

Under Assumptions~\ref{assump:ident} and \ref{assump:no-eff-mod}, one can verify that
$\exp(\beta_0)=E[Y(1)]/E[Y(0)]$, which is the marginal causal RR. The potential infection outcome means $E[Y(1)]$ and $E[Y(0)]$ in the target population cannot be identified due to the study selection process without an additional restriction. Nevertheless, the estimating equation~\eqref{eq:ee-general-population} implies that  it may still be possible to identify $\beta_0$ without necessarily identifying $E[Y(0)]$ and $E[Y(1)]$. The following proposition indicates that the same is true when the data are subject to selection bias of certain structure.

\begin{proposition}\label{prop:selection_bias}
Under Assumptions \ref{assump:ident}-~\ref{assump:no-eff-mod}, the parameter $\beta_0$ satisfies
\begin{equation}
    E[V_0(A, Y, U, X;\beta_0)\mid S=1]=0.\label{eq:select_eq}
\end{equation}
\end{proposition}
The proof of Proposition~\ref{prop:selection_bias} is in Appendix~\ref{append:oracle-ee}. From Proposition~\ref{prop:selection_bias}, the IPTW estimating function $V_0$  derived from the target population is also unbiased with respect to the study sample.

Under Assumptions~\ref{assump:ident}-\ref{assump:no-eff-mod}, one can estimate $\beta_0$ with $\hat{\beta}$ as the solution to
\begin{equation}\label{eq:propensity}
    \dfrac{1}{n}\sum_{i=1}^n(-1)^{1-A_i}\widehat{Q}(A_i, U_i, X_i)Y_i\exp(-\hat{\beta_0} A_i)=0,
\end{equation}

where $n$ is the size of the selected sample,  $\widehat{Q}(A_i, U_i, X_i)=\widehat{P}(A=A_i\mid U_i, X_i)$ is the estimated probability of have vaccination status $A=A_i$ given confounders $(U_i,X_i)$. The resulting estimator $$\hat\beta_0=\left[\sum_{i=1}^n\hat Q(A_i, U_i, X_i)A_iY_i\right]/\left[\sum_{i=1}^n\hat Q(A_i, U_i, X_i)(1-A_i)Y_i\right]$$
is essentially the IPTW estimator of marginal RR in~\citet{schnitzer2022estimands} assuming $(U_i, X_i)$'s are all observed.

However, $Q(A,U,X)$ cannot be estimated because $U$ is unobserved. Furthermore, even if $U$ were observed, $\widehat{Q}(A_i, U_i, X_i)$ may not be identified from the TND sample due to selection bias.
In the next section, we describe a new framework to account for unmeasured confounding in a TND setting, leveraging negative control exposure and outcome variables.

\subsection{Tackling unmeasured confounding bias leveraging negative controls \label{sec:bridge-approach}}

\subsubsection{Negative control exposure (NCE) and treatment confounding bridge function}
\label{sec:nce}

As shown in Figure~\ref{fig:dag-tnd}(e), suppose that one has observed a valid possibly vector-valued NCE, denoted as $Z$, which is a priori known to satisfy the following key independence conditions:
\begin{assumption}[NCE independence conditions]\label{assump:nce}
$Z\indep (Y,S)|A, U, X$. 
\end{assumption}
 Assumption~\ref{assump:nce} essentially states that any existing $Z-Y$ association conditional on $(X,A)$ in the target population must be a consequence of their respective association with $U$, therefore indicating the presence of confounding bias. Importantly, the NCE must a priori be known to have no causal effect on infection status~\citep{miao2018confounding}. Likewise, the association between $Z$ and $S$ conditional on $(X,A)$ is completely due to their respective association with $U$. Figure \ref{fig:dag-tnd}(e) presents a graphical illustration of an NCE that satisfies Assumption~\ref{assump:nce}.

In the Influenza VE setting, a candidate NCE can be vaccination status for the preceding year, or other vaccination status such as Tdap (Tetanus, Diphtheria, Pertussis) vaccine, as both are known to effectively provide no protection against the circulating flu strain in a given year. %little effect on the current-year flu infection.  
We now provide an intuitive description of our approach to leverage $Z$ as an imperfect proxy of $U$ for identification  despite being unable to directly observe $U$.

To illustrate the rationale behind identification, ignore selection bias for now and suppose that  $Q(A, U)=\alpha_0 + \alpha_1A + \alpha_2 U$, suppressing  measured confounders $X$. Although $U$ is unobserved, suppose further that $Z$ satisfies $E[Z|A, U] = \gamma_0 + \gamma_1A + \gamma_2U$. Then we have that \begin{align*}
Q(A, U) &= E[q(A, Z)|A, U],\qquad U = E[\tilde U(A, Z)|A, U],
\end{align*} where $\tilde U=(Z - \gamma_0 - \gamma_1A)/\gamma_2$. Replacing $U$ with $\tilde U$ in $Q(A, U)$, we get $q(A, Z)=\alpha_0 + \alpha_1A + \alpha_2\tilde U(A, Z)$, which does not depend on unmeasured confounder $U$ and can represent the inverse probability of vaccination as
 $Q(A, U)=E[q(A, Z)|A, U]$.
If all parameters of $q$ were known, it would naturally follow that the IPTW method in \eqref{eq:iptw} can be recovered by \[
E[Y(a)]=E\{I(A=a)E[q(A, Z)|A, U]Y\}\stackrel{A.\ref{assump:nce}}{=}E\left[I(A=a)q(A, Z)Y\right],
\]
  
Therefore, $\beta_0$ can be identified if the distribution of $(A,Y,Z)$ in the target population is available provided that parameters indexing $q$ can be identified.  

The above insight motivates the following assumption:

\begin{assumption}[treatment confounding bridge function]\label{assump:trt-bridge} There exists a function $q(A, Z, X)$ that satisfies, for every $a$, $u$ and $x$, \begin{equation}\label{eq:trt-bridge-defn}Q(A=a,U=u,X=x)=E\left[q(A, Z, X)|A=a, U=u, X=x\right]\end{equation}
\end{assumption}

The function $q$ that satisfies~(\ref{eq:trt-bridge-defn}) is called a treatment confounding bridge function, as it bridges the observed NCE with the unobserved propensity score~\citep{cui2020semiparametric}. Below we give two examples where the integral equation~\eqref{eq:trt-bridge-defn} can easily be solved and the treatment confounding bridge function $q$ admits a closed form solution.

\begin{example}
\label{eg:binary} (Binary $U$ and $Z$)  Suppose that $U$ is binary, and so is the NCE $Z$. For simplicity we suppress $X$. The integral equation~\eqref{eq:trt-bridge-defn}  can then be written as
$\sum_{z=0}^1 q(a, z)P(Z=z|U=u, A=a) =  {P(A=a|U=u)^{-1}}$, 

or equivalently,
$\sum_{z=0}^1\,p_{za.u}q(a, z)=1$
for each $a, u\in\{0, 1\}$, where  $p_{za.u}=P(Z=z, A=a|U=u)$. Therefore, the treatment confounding bridge function $q(a, z)$ solves the linear equation system
$$P_{Z,A\mid U}\begin{pmatrix}q(a,0)\\q(a,1)\end{pmatrix}=\begin{pmatrix}1\\1\end{pmatrix},\text{ where }
P_{Z,A\mid U}=\begin{pmatrix}
p_{0a.0} & p_{1a.0}\\
p_{0a.1} & p_{1a.1}
\end{pmatrix}.
$$

If the matrix $P_{Z,A\mid U}$ is invertible, then $q(a, z)$ has a closed form solution given by

\begin{equation}\label{eq:trt-bridge-binary}
    q(a, z) = \left[p_{1a.1}-p_{1a.0} + (p_{0a.0}-p_{0a.1}-p_{1a.1}+p_{1a.0})z\right]/\left(p_{0a.0}p_{1a.1}-p_{0a.1}p_{1a.0}\right).
\end{equation}
The result can be extended to the cases where $Z$ is polytomous as detailed in Appendix~\ref{append:eg-categorical}.

\end{example}

\begin{example}
\label{eq:contn}

(Continuous $U$ and $Z$) Suppose the unmeasured confounder $U$ and the NCE $Z$ are continuous. Further assume that 
\begin{align*}
    A|U, X  &\sim Bernoulli([1+\exp(-\mu_{0A} - \mu_{UA}U - \mu_{XA}X)]^{-1})\\
    Z|A, U, X &\sim N(\mu_{0Z} + \mu_{AZ}A +\mu_{UZ}U + \mu_{XZ}X, \sigma_{Z}^2).
\end{align*}

By the derivation in Appendix~\ref{append:trt-bridge-cont}, the treatment confounding bridge function $q(A, Z, X)$ is
\begin{equation}\label{eq:trt-bridge-cont}q(A, Z, X)= 1+ \exp\left[(-1)^A(\tau_0 + \tau_1A + \tau_2Z + \tau_3X)\right]\end{equation}
where 
$\tau_0 = \mu_{0A} - \dfrac{\mu_{UA}\mu_{0Z}}{\mu_{UZ}} - \dfrac{\sigma_z^2\mu_{UA}^2}{2\mu_{UZ}^2}$,
$\tau_1 = \dfrac{\sigma_z^2\mu_{UA}^2}{\mu_{UZ}^2} - \dfrac{\mu_{UA}\mu_{AZ}}{\mu_{UZ}},$
$\tau_2 = \mu_{UA} / \mu_{UZ},$ and
$\tau_3 = \mu_{XA} - \mu_{XZ}\mu_{UA}/\mu_{UZ}.$
\end{example}

Formally, Equation~\eqref{eq:trt-bridge-defn} defines a Fredholm integral equation of the first kind, with treatment confounding bridge function $q(A, Z, X)$ as its solution~\citep{cui2020semiparametric}. Heuristically, the existence of a treatment confounding bridge function requires that variation in $Z$ induced  by $U$ is sufficiently correlated with variation in $A$ induced by $U$. For instance, in Example~\ref{eg:binary}, existence of a treatment confounding bridge function requires that the matrix $P_{Z,A\mid U}$

is nonsingular. In Example~\ref{eq:contn}, the existence of a treatment confounding bridge function amounts to the condition $\mu_{UZ}\neq 0$, which again requires $Z\nindep U|A, X$. \citet{cui2020semiparametric}  provided formal conditions sufficient for the existence of the treatment confounding bridge function satisfying   Equation~\eqref{eq:trt-bridge-defn}. These conditions are reproduced for completeness in Appendix~\ref{append:trt-bridge-existence}.

Thus, under Assumption~\ref{assump:trt-bridge}, we propose to construct a new unbiased estimating function for $\beta_0$ by replacing $Q(A, U, X)$ with $q(Z,A,X)$ in $U_0(A, Y, U, X;\beta_0)$.

\begin{theorem}\label{thm:causal-rr-identifiability} (Moment restriction of $\beta_0$)
Under Assumptions \ref{assump:ident}-\ref{assump:trt-bridge}, we have that $$E[V_1(A, Y, Z, X;\beta_0)\mid S=1]=0$$ where
$V_1(A, Y, Z, X;\beta_0)=(-1)^{1-A}q(A, Z, X)Y\exp(-\beta_0 A).$
\end{theorem}

The proof of Theorem~\ref{thm:causal-rr-identifiability} is in Appendix~\ref{append:ee-rr-unbiasedness}. In practice, if one can consistently estimate the treatment confounding bridge function  $ q(A, Z, X)$ with $\widehat q(A, Z, X)$, Theorem~\ref{thm:causal-rr-identifiability} suggests estimating $\beta_0$ by solving the estimating equation
\begin{equation}\label{eq:ee-risk-ratio}
     \dfrac{1}{n}\sum_{i=1}^n(-1)^{1-A_i}\widehat q(A_i, Z_i, X_i)Y_i\exp(-\beta_0 A_i)=0,
\end{equation}
which results in a closed form estimator
$$\widehat{\beta}_0 = \log\left(\dfrac{\sum\widehat{q}(A_i, Z_i, X_i)A_iY_i}{\sum\widehat{q}(A_i, Z_i, X_i)(1-A_i)Y_i}\right).$$

Importantly, although \eqref{eq:trt-bridge-defn} may not have a unique solution, any solution uniquely identifies the causal log RR $\beta_0$. 
The result in Theorem~\ref{thm:causal-rr-identifiability} cannot directly be applied in practice because the treatment confounding bridge function is not identifiable even if random samples from the target population were available -- solving~\eqref{eq:trt-bridge-defn} requires additional information about $U$ which is unobserved. For instance, in Example 1 one is unable to directly estimate $q(a, z)$ because $p_{za.u}$  in \eqref{eq:trt-bridge-binary} cannot directly be estimated from the observed data.

\subsubsection{Negative control outcome (NCO) for identification of treatment confounding bridge function }

For identification and estimation of $q$, we leverage negative control outcomes (NCO) to construct feasible estimating equations for the treatment confounding bridge function as in \citet{cui2020semiparametric}. Similar to NCEs, NCOs can be viewed as imperfect proxies of $U$. However, unlike NCEs, a valid NCO, denoted by $W$, is a measured covariate which is (i) known a priori not to be a causal effect of either the primary exposure $A$ nor the NCE $Z$; and (ii) is associated with $(A,Z)$ conditional on $X$ only to the extent that it is associated with $U$.

Formally, we make the following assumption.

\begin{assumption}\label{assump:nco}(NCO Independence Conditions).
\begin{multicols}{2}
(a) $W\indep A|U, X$;

(c) $S\indep Z|A, U, X, W, Y$.

(b) $W \indep Z|A, U, X, Y$;
\end{multicols}
\end{assumption} 

Assumptions~\ref{assump:nco}(a) and (b) formalize the requirement that neither the primary exposure nor NCE have direct effects on the NCO. Assumption~\ref{assump:nco}(c) complements Assumption~\ref{assump:nce} and states that conditioning on $W$ in addition to $(A, U, X, Y)$ does not alter the conditional independence of $Z$ with $S$. In flu VE studies, a candidate NCO can be an infection whose risk is not causally affected by either $A$ or $Z$.  For example, if the NCE is selected to be Tdap vaccination, then a potential NCO may be current-year respiratory syncytial virus infection, as its risk is unlikely to be affected by Influenza or Tdap vaccination. Recent outpatient visits for other acute illnesses can also serve as NCO, such as blepharitis, wrist/hand sprain, lipoma, ingrowing nail, etc.~\citep{leung2011herpes}. Figure~\ref{fig:dag-tnd}(e) illustrates an NCO $W$ that satisfies Assumptions~\ref{assump:nco}(a) and (b).

 Similar to \citet{cui2020semiparametric}, we leverage the availability of an NCO as an additional proxy to identify the treatment confounding bridge function. However, a complication arises due to lack of a random sample from the target population, a key requirement in the approach outlined in \citet{cui2020semiparametric}. In general, it is not possible to obtain sufficient information about neither the distribution of $W$ nor that of $U$ in the target population from the TND data without an additional structural assumption~\citep{bareinboim2012controlling}. In the following, we avoid imposing such an additional structural assumption by leveraging an important feature of infectious diseases such as Influenza and COVID-19; mainly that contracting such an infection is a rare event in most target populations of interest, and therefore information from the target population that is relevant for estimating the treatment confounding bridge function can be recovered from the test-negative control group. Formally, we make the following rare disease assumption.

\begin{assumption}[Rare infection]\label{assump:rare-disease}\text{There exist a small positive number $\delta>0$ such that }\begin{equation}\label{eq:rare-disease-assumption} 
    P(Y=1|A=a, W=w,U=u, X=x)\leq \delta,\qquad   \mbox{for almost every $a, w, u, x$}
\end{equation}
\end{assumption}Assumption~\ref{assump:rare-disease} states that infected subjects, whether vaccinated or not and regardless of their negative control outcomes, only constitute a small proportion of each $(U, X)$ stratum in the general population; specifically, the assumption implies that $\frac{1}{1-\delta}\leq \frac{P(A,Z|U=u, X=x,Y=0)}{P(A,Z|U=u, X=x)}\leq 1-\delta$. Thus, under Assumptions \ref{assump:trt-indep-samp}, \ref{assump:nce} and~\ref{assump:rare-disease}, $P(A=a,Z=z|U=u, X=x)\approx P(A=a,Z=z| U=u, X=x,Y=0,S=1)$ for all $a,z,x,u$. We now introduce a key property of the treatment confounding bridge function in Theorem~\ref{thm:trt-bridge-identification}, which is proved in Appendix~\ref{append:trt-bridge-identification}.
\begin{theorem}[Identification of the treatment confounding bridge function]\label{thm:trt-bridge-identification} Under Assumptions~\ref{assump:trt-indep-samp}, \ref{assump:nce}, \ref{assump:trt-bridge},
\ref{assump:nco}, and \ref{assump:rare-disease}, for $a=0, 1$ we have that\begin{align*}\label{eq:trt-bridge-approx}
    & \dfrac{(1-\delta)^3}{P(A=a|W, X, Y=0, S=1)} < E[q(a, Z, X)|W, A=a, X, Y=0, S=1]\\&\qquad\qquad < \dfrac{1}{(1-\delta)^3P(A=a|W, X, Y=0, S=1)}
\end{align*}
\end{theorem}Thus, provided $\delta\approx 0$, Theorem~\ref{thm:trt-bridge-identification} suggests that an approximation to the treatment confounding bridge function can be obtained by solving the following integral equation involving only observed data
\begin{equation}\label{eq:trt-bridge-identification}
    E[q^*(A, Z, X)|W, A=a, X, Y=0, S=1] = 1/P(A=a|W, X, Y=0, S=1).
\end{equation}
as long as a solution exists. Accordingly, hereafter suppose that the following assumption holds.
 
\begin{assumption}[Existence of a unique solution to \eqref{eq:trt-bridge-identification}]\label{assump:solvability}
There exists a unique square-integrable function $q^*(A, Z, X)$ that satisfies \eqref{eq:trt-bridge-identification}.
\end{assumption}

Heuristically, uniqueness of a solution to~\eqref{eq:trt-bridge-identification} requires that variation in $W$ is sufficiently informative about variation in $Z$, in the sense that there is no source of variation in $W$ that is not associated with a corresponding source of variation in $Z$. See Appendix~\ref{append:completeness} for further elaboration of completeness conditions and \citet{newey2003instrumental,d2011completeness} for related use of the assumption in the literature. Below we briefly illustrate Assumption~\ref{assump:solvability}  in the examples of Section~\ref{sec:nce}. \begin{newexample}{1} Suppose $U$ and $Z$ are both binary, and a binary NCO $W$ is also observed. Let $p'_{za.w}=P(Z=z,A=a|W=w, Y=0, S=1)$ for $z,a,w\in\{0, 1\}$, then solving~\eqref{eq:trt-bridge-identification} is equivalent to solving the system of linear equations \begin{align*}
p'_{0a.0}q^*(a,0)  + p'_{1a.0}q^*(a,1) = 1;\qquad
p'_{0a.1}q^*(a,0) + p'_{1a.1}q^*(a,1) = 1,
\end{align*} giving $q^*(a, z) = \left[p'_{1a.1}-p'_{1a.0} + (p'_{0a.0}-p'_{0a.1}-p'_{1a.1}+p'_{1a.0})z\right]/\left(p'_{0a.0}p'_{1a.1}-p'_{0a.1}p'_{1a.0}\right)$.
Note that the probabilities $p'_{za.w}$ can all be estimated from the study sample.
\end{newexample}

We emphasize that the solution to Equation~\eqref{eq:trt-bridge-identification} is ultimately an approximation to the (non-identifiable) treatment confounding bridge function in the target population. The accuracy of this approximation relies on the extent to which the rare disease assumption holds in the target population of interest. We study the potential bias resulting from a departure of this key assumption in the Appendix~\ref{append:rare-disease}. We further observe that, under the null hypothesis of no vaccine effectiveness, or if $W$ has no direct effects on $Y$ or $S$, then the function $q^*(A, Z, X)$ equals the treatment confounding bridge exactly, even for a non-rare disease outcome, as stated in the corollary below. We prove Corollary~\ref{cor:null-preserving} in Appendix~\ref{append:trt-bridge-identification}. 

\begin{corollary}\label{cor:null-preserving}
Under the Assumptions of Theorem~\ref{thm:causal-rr-identifiability} and Assumption~\ref{assump:solvability}, if there is no vaccine effect against infection, such that $Y\indep A| U, X$.
\end{corollary}

From Theorem~\ref{thm:trt-bridge-identification}, we immediately have the following corollary which provides a basis for estimation of $q^*(A, Z, X)$ from the observed TND  sample. 
\begin{corollary}\label{cor:trt-bridge-ee}
Under the conditions listed in Theorem~\ref{thm:trt-bridge-identification}, for any function $m(W, A, X)$, the solution $q^*(A, Z, X)$ to Equation~\eqref{eq:trt-bridge-identification} also solves the population moment equation 
\begin{equation}\label{eq:trt-bridge-ee-unbiasedness}
    E\left[m(W, A, X)q^*(A, Z, X) - m(W, 1, X) - m(W, 0, A)|Y=0, S=1\right]=0.
\end{equation}
\end{corollary}
We prove Corollary~\ref{cor:trt-bridge-ee} in Appendix~\ref{append:ee-trt-bridge}. In practical situations where a parametric model $q^*(A, Z, X;\tau)$ for the treatment confounding bridge function might be appropriate, where $\tau$ is an unknown finite dimensional parameter indexing the model, Corollary~\ref{cor:trt-bridge-ee} suggests one can estimate $\tau$ by solving the estimating equation
\begin{equation}\label{eq:ee-trt-bridge}\dfrac{1}{n}\sum_{i=1}^n(1-Y_i)[m(W, A, X)q(A, Z, X;\tau)-m(W, 1, X) - m(W, 0, X)]=0,\end{equation}where $m(W, A, X)$ is a user-specified function whose dimension is no smaller than $\tau$'s. 

\begin{newexample}{1$'$}
If $Z$ and $W$ are both binary, rather than solving the system of equations implied by~\eqref{eq:trt-bridge-identification}, one can instead  specify a saturated model for the treatment confounding bridge function:\begin{equation}\label{eq:trt-bridge-model-saturated}
    q^*(A, Z;\tau) =\tau_0 + \tau_1Z+\tau_2A + \tau_3ZA 
\end{equation}and estimate $\tau=(\tau_0,\tau_1,\tau_2,\tau_3)^T$ by solving \eqref{eq:ee-trt-bridge} with $m(W, A)=(1, W, A, WA)^T$. Extension to $Z$ and $X$ with multiple categories is straightforward.
\end{newexample}

\begin{newexample}{2}
In case of continuous $(U, X, Z)$, result \eqref{eq:trt-bridge-cont} suggests the model
\begin{equation}\label{eq:trt-bridge-model-continuous}
    q^*(A, Z, X;\tau) = 1+\exp\left[(-1)^A (\tau_0 + \tau_1 A + \tau_2Z + \tau_3X)\right].
\end{equation}
If a univariate NCO $W$ is available, we may solve \eqref{eq:ee-trt-bridge} with $m(W, A, X)=(1, W, A, X)^T$.
\end{newexample}

\subsection{Estimation and Inference}
\label{sec:sum}
In the previous sections, we have defined the causal parameter of interest $\beta_0$ as stratum-specific log risk ratio, introduced the treatment confounding bridge function as a key ingredient to identification of $\beta_0$, and presented a strategy to estimate the  treatment confounding bridge function leveraging an available NCO. We summarize the steps of our estimation framework in Algorithm~\ref{alg:nc-tnd} and present the large-sample properties of the resulting estimator $(\widehat\beta,\widehat\tau)$ in Theorem~\ref{thm:ee-joint}.

\begin{algorithm}[!htbp]
\caption{\label{alg:nc-tnd}Negative control method to estimate vaccine effectiveness from a test-negative design} %\textcolor{red}{[TODO: R package]}}
\begin{algorithmic}[1]
\State Identify the variables in the data according to Figure~\ref{fig:dag-tnd}(e), in particular the NCEs and NCOs.
\State Estimate the treatment confounding bridge function by solving Equation~\eqref{eq:ee-trt-bridge} with a suitable parametric model $q^*(A, Z, X;\tau)$ and a user-specified function $m(W, A, X)$. Write $\widehat\tau$ as the resulting estimate of $\tau$.
\State  Estimate $\beta_0$ by solving Equation~(\ref{eq:ee-risk-ratio}) with $\widehat{q}(A, Z, X) = q^*(A, Z, X;\widehat\tau)$. The resulting estimator of $\beta_0$ is
    \begin{equation}\label{eq:rr-hat}\widehat\beta = \log\left(\dfrac{\sum q^*(A_i, Z_i, X_i;\widehat\tau)A_iY_i}{\sum q^*(A_i, Z_i, X_i;\widehat\tau)(1-A_i)Y_i}\right);
\end{equation}
 The estimated vaccine effectiveness is 
    $\widehat{VE} = 1 - \exp(\widehat\beta).$
\end{algorithmic}
\end{algorithm}

\begin{theorem}[Inference based on $(\widehat\beta,\widehat\tau)$]\label{thm:ee-joint}
Under Assumptions \ref{assump:ident}-\ref{assump:solvability} and suitable regularity conditions provided in Appendix~\ref{append:ee-joint} , the estimator $(\widehat\beta, \widehat\tau)$ in Algorithm~\ref{alg:nc-tnd}, or equivalently, the solution to the estimating equation $\dfrac{1}{n}\sum_{i=1}^n G_i(\beta, \tau)=0$
is regular and asymptotically linear with the i-th influence function
$$IF_i(\beta, \tau) = -\left[\Omega(\beta, \tau)^T\Omega(\beta, \tau)\right]^{-1}\Omega(\beta, \tau)^TG_i(\beta, \tau),$$
where
$$G_i(\beta, \tau) = \left(\begin{array}{c}  {(-1)}^{1-A_i}q^*(A_i, Z_i, X_i; \tau)Y_i\exp(-\beta A_i)\\
(1-Y_i)\left[m(W_i, A_i, X_i)q^*(A_i, Z_i, X_i;\tau) - m(W_i, 1, X_i) - m(W_i, 0, X_i)\right]\end{array}\right)$$
and $\Omega(\beta, \tau)=\left(E\left[\partial G_i(\beta, \tau)/\partial \beta^T\right], E\left[\partial G_i(\beta, \tau)/\partial \tau^T\right]\right).$\end{theorem}

    The proof follows from standard estimating equation theory (See~\citet{van2000asymptotic} Theorem 5.21). An immediate consequence of Theorem~\ref{thm:ee-joint} is that we may estimate the variance-covariance matrix of $(\widehat\beta,\widehat\tau)$ with
\begin{equation}\label{eq:sandwich}
    \widehat\Sigma_n = \left[\widehat\Omega(\widehat\beta, \widehat\tau)^T\widehat\Omega(\widehat\beta, \widehat\tau)\right]^{-1}\widehat\Omega(\widehat\beta, \widehat\tau)^T\widehat{Var}(G_i(\widehat\beta,\widehat\tau))\widehat\Omega(\widehat\beta, \widehat\tau)^T\left[\widehat\Omega(\widehat\beta, \widehat\tau)^T\widehat\Omega(\widehat\beta, \widehat\tau)\right]^{-1}/n,
\end{equation}
where $
\widehat\Omega(\beta,\tau)=\left(\hat E\left[\partial G_i(\beta, \tau)/\partial \beta^T\bigg|_{\beta = \hat\beta,\tau = \hat\tau}\right], \hat E\left[\partial G_i(\beta, \tau)/\partial \tau^T\bigg|_{\beta = \hat\beta,\tau = \hat\tau}\right]\right).$ Here $\hat E$ and $\widehat{Var}$ denote the expectation and variance with respect to the empirical distribution, respectively. A two-sided $\alpha$-level Wald-type confidence interval of VE can then be obtained as
$$\left(1 - \exp\left(\widehat\beta -z_{1-\alpha/2}\sqrt{\widehat\Sigma_{n, 1, 1}}\right), 1 - \exp\left(\widehat\beta +z_{1-\alpha/2}\sqrt{\widehat\Sigma_{n, 1, 1}}\right)\right)$$
where $\widehat\Sigma_{n, 1, 1}$ is the $(1,1)$-th entry of $\widehat\Sigma_n$ and $z_{1-\alpha/2}$ is the $(1-\alpha/2)$-th quantile of a standard normal distribution.

The estimator $\hat\beta$ and the above confidence interval are constructed under the assumption that the disease is rare in the target population; for non-rare diseases, $\hat\beta$ is in general going to be biased and the confidence interval may not be well-calibrated. However, by Corollary~\ref{cor:null-preserving}, under the null hypothesis of no vaccine effects, the estimated $q^*(A, Z, X)$ converges to the true treatment confounding bridge function and $\hat\beta$ is consistent for $\beta_0=0$. This implies that while our methods are approximately asymptotically unbiased for rare infections, they provide a valid test of no vaccine effect even if the infection is not rare. 

\subsection{Accounting for effect modification by measured confounders}
\label{sec:eff-mod}

So far we have operated under Assumption~\ref{assump:no-eff-mod} that VE is constant across levels of $(U,X)$. As we now show, this assumption can be relaxed to allow for potential effect modification with respect to $X$ without compromising identification. This extension is particularly important because empirical evidence has indeed suggested that flu vaccine effectiveness may vary across sex and age groups~\citep{chambers2018should}; and similar effect heterogeneity is of key interest in case of COVID-19~\citep{fernandez2021effect}.

Instead of Assumption~\ref{assump:no-eff-mod}, we consider a less stringent assumption:
\begin{assumption}[No effect modification by unmeasured confounders]\label{assump:no-eff-mod-by-u}
\begin{equation}\label{eq:no-eff-mod-by-u}
    P(Y=1|A=a, U, X)=\exp(\beta_0(X)a)g(U,X)
\end{equation}
where $\beta_0(x)$ are $g(u,x)$ are unknown functions of $x$ and $u,x$ respectively.
\end{assumption}
Under condition~\ref{assump:ident}, Assumption~\ref{assump:no-eff-mod-by-u} further implies that
$\beta_0(x)=E[Y(1)|X=x]/E[Y(0)|X=x]$, i.e. the conditional causal RR as a function of $x$. 
Similar to Theorem~\ref{thm:causal-rr-identifiability}, we have:

\begin{theorem}\label{thm:rr-identification-2}
Under Assumptions~\ref{assump:ident}, \ref{assump:trt-indep-samp}, \ref{assump:nce}, \ref{assump:trt-bridge} and \ref{assump:no-eff-mod-by-u}, we have that  $E[V_3(A, Y, Z, X;\beta_0)|S=1]=0$, where
$V_3(A, Y, Z, X;\beta_0)=(-1)^{1-A}q(A, Z, X)\exp(-\beta_0(X)A)$.
\end{theorem}

The proof of Theorem~\ref{thm:rr-identification-2} is identical to that of Theorem~\ref{thm:causal-rr-identifiability} with $\beta_0A$ replaced with $\beta_0(X)A$. Identification and estimation of the treatment confounding bridge function are also essentially identical to that of Corollary~\ref{cor:trt-bridge-ee}. Therefore, it is straightforward to extend Algorithm~\ref{alg:nc-tnd} to allow effect modification by measured confounders. We describe the algorithm and the large sample properties of the resulting estimator in Appendix~\ref{append:tnd-alg-2}.

 \subsection{Estimating VE under treatment-induced selection into TND sample}\label{sec:a-s-arrow}
 
Thus far, unbiasedness of the estimating function $V_0$ has crucially relied  on Assumption~\ref{assump:trt-indep-samp} that $A$ does not have a direct effect on $S$. In some settings, the assumption may be violated if an infected person who is vaccinated is on average more likely to present to the ER than an unvaccinated infected person with similar symptoms, so that treatment or vaccination-induced selection into the TND sample is said to be present. In such settings, both estimators $\hat\beta$ and $\hat\beta(X)$ produced by Algorithms~\ref{alg:nc-tnd} and \ref{alg:nc-tnd-2} may be severely biased because Assumption~\ref{assump:trt-indep-samp} may no longer be valid. Crucially, we note that this form of selection bias can be present even in context of a randomized trial in which vaccination/treatment is assigned completely at random, if the outcome is ascertained using a TND, for example in the cluster-randomized test-negative design studies of community-level dengue intervention effectiveness \citet{anders2018awed,jewell2019analysis,dufault2020analysis,wang2022randomization}.  In this Section, we provide sufficient conditions for identification under treatment-induced selection. In this vein, consider the following assumptions:

\begin{newassumption}{\ref{assump:trt-indep-samp}} $P(S=1|A=a, Y=1, U, X)/P(S=1|A=a, Y=0,U, X)=\exp(h(U, X))$ for $a=0,1.$ 

\end{newassumption}
That is, the risk ratio association between infection status and selection into the TND sample is independent of vaccination status. Furthermore,  

\begin{newassumption}{\ref{assump:no-eff-mod}}(No effect modification by confounders on the OR scale).
$$\dfrac{P(Y=1|A=1, U, X)/P(Y=0|A=1, U, X)}{P(Y=1|A=0, U, X)/P(Y=0|A=0,U, X)}=\exp(\beta_0').$$
\end{newassumption}
Recall that Assumption~\ref{assump:no-eff-mod} posited a constant vaccination causal effect on the RR scale across levels of $(U,X)$, while Assumption 3$'$ posits that the corresponding causal effect on the odds ratio scale is constant across levels of $(U,X)$. In case of a rare infection in the target population, the OR and RR are approximately equal, in which case VE is well approximated by $1-OR$.

Furthermore, identification relies on the following modified definition of a treatment confounding bridge function:
\begin{newassumption}{\ref{assump:trt-bridge}}
    There exists a treatment confounding bridge function $\tilde q$ such that for $a=0, 1$, \begin{equation}\label{eq:trt-bridge-alternative}E[\tilde q(a, Z, X)|A=a, U, X]=1/P(A=a|U, X, Y=0, S=1)\qquad\text{almost surely}.\end{equation}
\end{newassumption}

Note that if the infection is rare in the target population in the sense of Assumption~\ref{assump:rare-disease}, then the treatment confounding bridge function defined in Assumption~\ref{assump:trt-bridge} in Section~\ref{sec:nce} satisfies~\eqref{eq:trt-bridge-alternative} approximately .

We now introduce the identification of the OR with the following theorem:
 \begin{newt}{1}
Under Assumptions \ref{assump:ident}, \ref{assump:trt-indep-samp}$'$, \ref{assump:no-eff-mod}$'$, \ref{assump:nce} and \ref{assump:trt-bridge}$'$, we have 
$$E[\tilde V_1(A, Y, Z, X;\beta_0')\mid S=1]=0$$ where $\tilde V_1(A, Y, Z, X;\beta)$ is the same as $V_1$ defined in Theorem~\ref{thm:causal-rr-identifiability} except with ${\tilde q}$ replacing $q$.

\end{newt}
 Importantly, the theorem establishes that the estimating function $V_1$ previously developed in the paper can under certain conditions, remain unbiased for the odds ratio association of vaccination with testing positive for the infection, even in the presence of treatment-induced selection into the TND sample. We leave the proof of Theorem 1$'$ to Appendix~\ref{append:causal-or-ee}.

 Estimation of the treatment confounding bridge function ${\tilde q}(A, Z, X)$ requires a negative control outcome that satisfies:

\begin{newassumption}{\ref{assump:nco}}(NCO Independence Conditions) $W\indep (A, Z,  S)|U, X, Y.$
\end{newassumption}
In addition to Assumptions~\ref{assump:nco}, this last assumption requires that neither $Y$ nor $S$ is a causal effect of $W$. Figure~\ref{fig:dag-tnd}(f) illustrates a DAG that satisfies our assumptions regarding $(Z,W)$. As can be verified in the graph, Assumption 6$'$ is needed to ensure that collider stratification bias induced by the path $A\rightarrow [S=1]\leftarrow W$ upon conditioning on $S=1$ is no longer present. Identification of the function $\tilde q$ is given below:
 
 \begin{newt}{\ref{thm:trt-bridge-identification}}
 Under assumptions~\ref{assump:nce}, \ref{assump:trt-bridge}$'$ and
\ref{assump:nco}$'$, for $a=0, 1$ we have that
\begin{align*}
    E[\tilde q(a, Z, X)|A=a, W, X, Y=0, S=1]= 1/P(A=a|W, X, Y=0, S=1)
\end{align*}
 \end{newt}

 We prove Theorem~\ref{thm:trt-bridge-identification}$'$ in Appendix~\ref{append:tilde-q-identification}. As a result of Theorem~\ref{thm:trt-bridge-identification}$'$, the parameters in the treatment confounding bridge function can be estimated by solving moment equation \eqref{eq:ee-trt-bridge}.

  In summary, the above discussion suggests that one can continue to use Algorithm~\ref{alg:nc-tnd} to estimate VE in presence of treatment induced selection bias, albeit on the OR scale and under a modified set of negative control conditions. Algorithm~\ref{alg:nc-tnd-2} can similarly be justified under treatment-induced selection with assumptions in this section, except that $\beta_0'$ in Assumption~\ref{assump:no-eff-mod}$'$ is replaced by the conditional log RR $\beta_0'(X)$.
 
As a side note, Assumption~\ref{assump:trt-indep-samp}$'$ automatically holds under Assumption~\ref{assump:trt-indep-samp}, and hence the above results in this section also apply to the setting in previous sections that is illustrated in Figure~\ref{fig:dag-tnd}(e). We present this statement in the following corollary.
 
\begin{corollary}\label{cor:ee-or}
 Under Assumptions~\ref{assump:ident}, \ref{assump:trt-indep-samp}, \ref{assump:no-eff-mod}$'$, \ref{assump:nce} and \ref{assump:trt-bridge}$'$, we have 
$$E[\tilde V_1(A, Y, Z, X;\beta_0')\mid S=1]=0.$$
 \end{corollary}
 
 With Assumption~\ref{assump:trt-indep-samp}, the treatment confounding bridge function $\tilde q$ can be estimated by solving the moment equation~\eqref{eq:ee-trt-bridge} either under under Assumption~\ref{assump:nco}$'$ and \ref{assump:solvability}, or Assumptions~\ref{assump:nco}, \ref{assump:rare-disease} and \ref{assump:solvability} as an approximation under the rare disease assumption. Corollary~\ref{cor:ee-or} leads to an interesting observation: under the treatment-independent sampling (Assumption~\ref{assump:trt-indep-samp}), the estimator $\hat\beta$ from Algorithm~\ref{alg:nc-tnd} can be viewed as either log RR or log OR, depending on the set of assumptions made.

\section{Simulation Study}
\label{sec:sim}

To assess the empirical performance of our proposed method, we consider two settings with different types of confounding and negative control variables, and perform corresponding simulation studies. 

In the first setting, we consider no measured confounder, a binary unmeasured confounder $U$, a binary NCE $Z$ and a binary NCO $W$. To trigger selection among subjects with $Y=0$, we let $D$ be a binary indicator of the presence of other flu like illnesses. The treatment confounding bridge function is thus given by \eqref{eq:trt-bridge-binary}. We assume the distribution of $Y$ is Bernoulli with a log-linear risk model:
$Y|A, U\sim Bernoulli(\exp(\eta_{0Y}+\beta_0 A + \eta_{UY}U))$. We consider values of $\beta_0$ to be -1.609, -0.693, -0.357 or 0, corresponding to a risk ratio of 0.2, 0.5, 0.7 or 1. We assume the selection $S$ only equals one with nonzero probability if at least one of $Y$, $W$ and $D$ equals one, and is independent of $A$ and $Z$ conditional on other variables. The resulting prevalence of $Y$ in the target population is 0.75\% among the unvaccinated individuals and 0.55\%, 0.65\%, 0.72\% or 0.75\% among the vaccinated individuals, corresponding to four values of $\beta_0$.  Next, we consider a setting where $X$, $U$, $Z$ and $W$ are all univariate continuous variables.
We generate the infection outcome using a log-linear model
$$Y|A, U, X\sim Bernoulli(\exp(\mu_{0Y} + \beta_0A + \mu_{UY}U + \mu_{XY}X + \mu_{UXY}UX)).$$
We generate $A$ 
and $Z$ following Example~\ref{eq:contn} in Section~\ref{sec:bridge-approach}. As such the treatment confounding bridge function is given by Equation \eqref{eq:trt-bridge-cont}.  The probability of $S=1$ is 1 only if at least one of $Y$ and $D$ is nonzero. The resulting prevalence of $Y$ in the target population is 0.34\% among the unvaccinated individuals and 0.24\%, 0.28\%, 0.31\% or 0.34\% among the vaccinated individuals, corresponding to four values of $\beta_0$. Appendix~\ref{append:sim-binary} and ~\ref{append:sim-continuous} give more details on the data-generating mechanism for the two settings.

In each scenario, we simulate a target population of size $N=7,000,000$ and implement $1,000$ simulation iterations. For both settings, we evaluate the performance of three estimators for $\beta_0$:
\begin{itemize}
    \item NC estimator: our proposed estimator given by Algorithm~\ref{alg:nc-tnd}. In the first setting, we use a saturated parametric model \eqref{eq:trt-bridge-model-saturated} for the treatment confounding bridge function in the first setting, with $m(W,A)=(1,W,A,WA)^T$; in the second setting, we use model~\eqref{eq:trt-bridge-model-continuous} and $m(W,A,X)=(1,W,A,X)^T$.
    \item NC-Oracle estimator: the estimated treatment confounding bridge function in Algorithm~\ref{alg:nc-tnd} is only an approximation under Assumption~\ref{assump:rare-disease}, whose bias may affect the estimation for $\beta_0$, as derived in Appendix~\ref{append:rare-disease}. We therefore include NC-Oracle estimator  that uses the true treatment confounding bridge function $q(A, Z, X)$. Appendices \ref{append:trt-bridge-cont} include derivation of the true treatment confounding bridge function under the continuous $(X, U, Z, W)$ setting.
    \item Logistic regression: we also consider a logistic regression model of $Y$ on $A$ (and $X$ in the second setting), overlooking the unmeasured confounders $U$. This is a common choice for covariate adjusted analyses of test-negative designs but ignores biases caused by $U$~\citep{bond2016regression}. We comment  in Appendix~\ref{append:logit-reg} that the estimator is appropriate in the absence of unmeasured confounders except for potential model misspecification.

\end{itemize}

We note that the target parameter $\beta_0$ for NC estimator and NC-Oracle estimator is 
    log causal RR, while
    logistic regression gives log causal OR. However, the two parameters are approximately equal under Assumption~\ref{assump:rare-disease} where the infection risk is low in the target population.

Figure~\ref{fig:sim} shows the bias of three estimators considered and the coverage of their 95\% confidence intervals. In both settings, both NC and NC-Oracle are essentially unbiased whereas logistic regression gives a biased estimate in all scenarios. NC-Oracle exhibits slightly higher precision than NC, which implies that estimating the treatment confounding bridge function in the TND is only slightly more variable. The 95\% confidence intervals for NC and NC-Oracle both achieve nominal coverage, whereas logistic regression-based confidence intervals under-cover severely. We repeated the simulation under a non-rare disease setting in Appendix~\ref{append:non-rare}. In such scenarios, while NC-Oracle estimator is still unbiased with calibrated 95\% confidence intervals, the NC estimator is biased in general except when $\beta_0=0$. We conclude that the proposed NC estimator is unbiased of the log causal RR either under a rare disease setting or under a non-rare disease setting with no vaccine effect.

\begin{figure}[h]
    \centering
    \begin{tabular}{cc}
    \includegraphics[width = 0.45\textwidth]{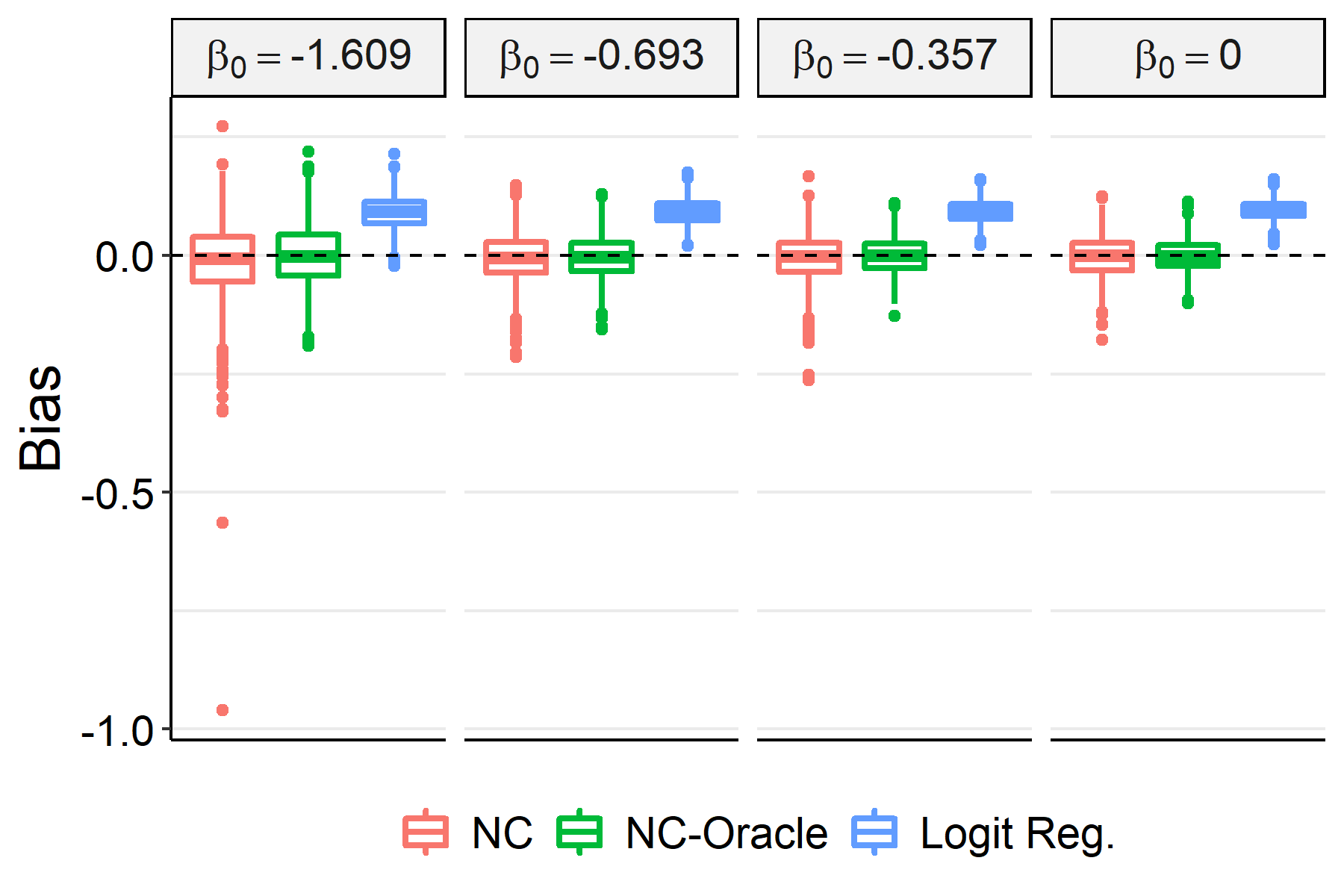}&
    \includegraphics[width = 0.45\textwidth]{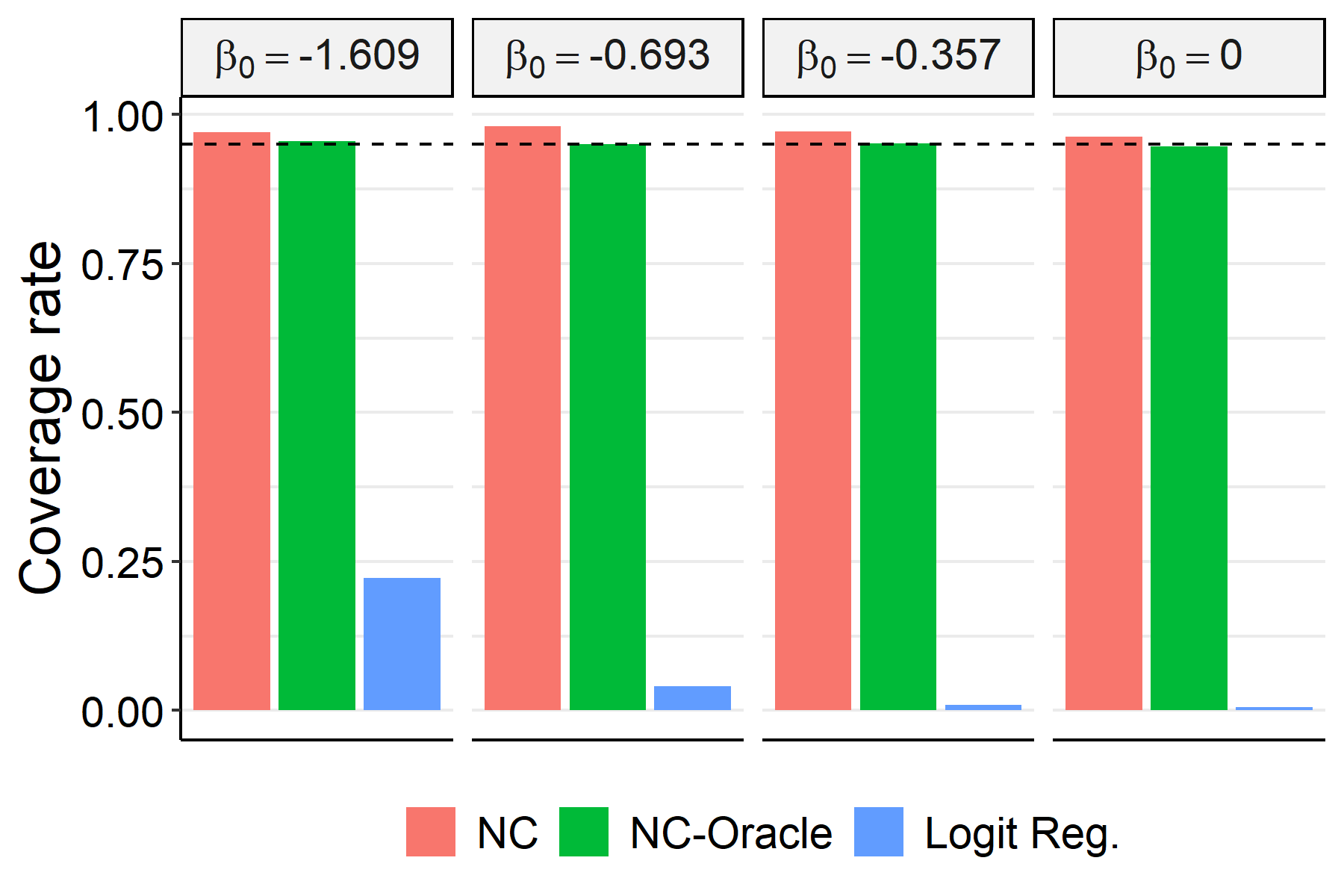}\\
    \multicolumn{2}{c}{(a)}\\
    \includegraphics[width = 0.45\textwidth]{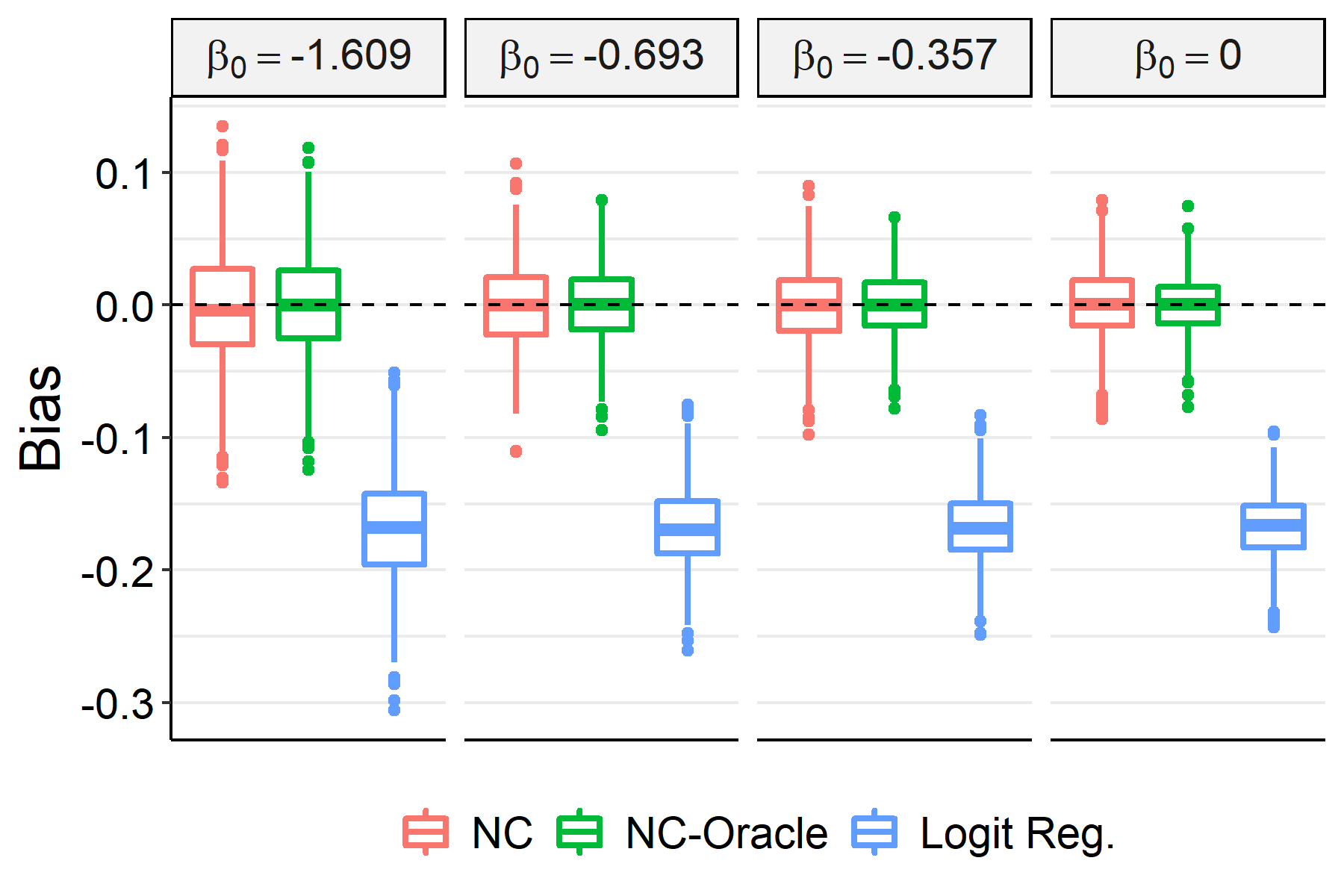}&
    \includegraphics[width = 0.45\textwidth]{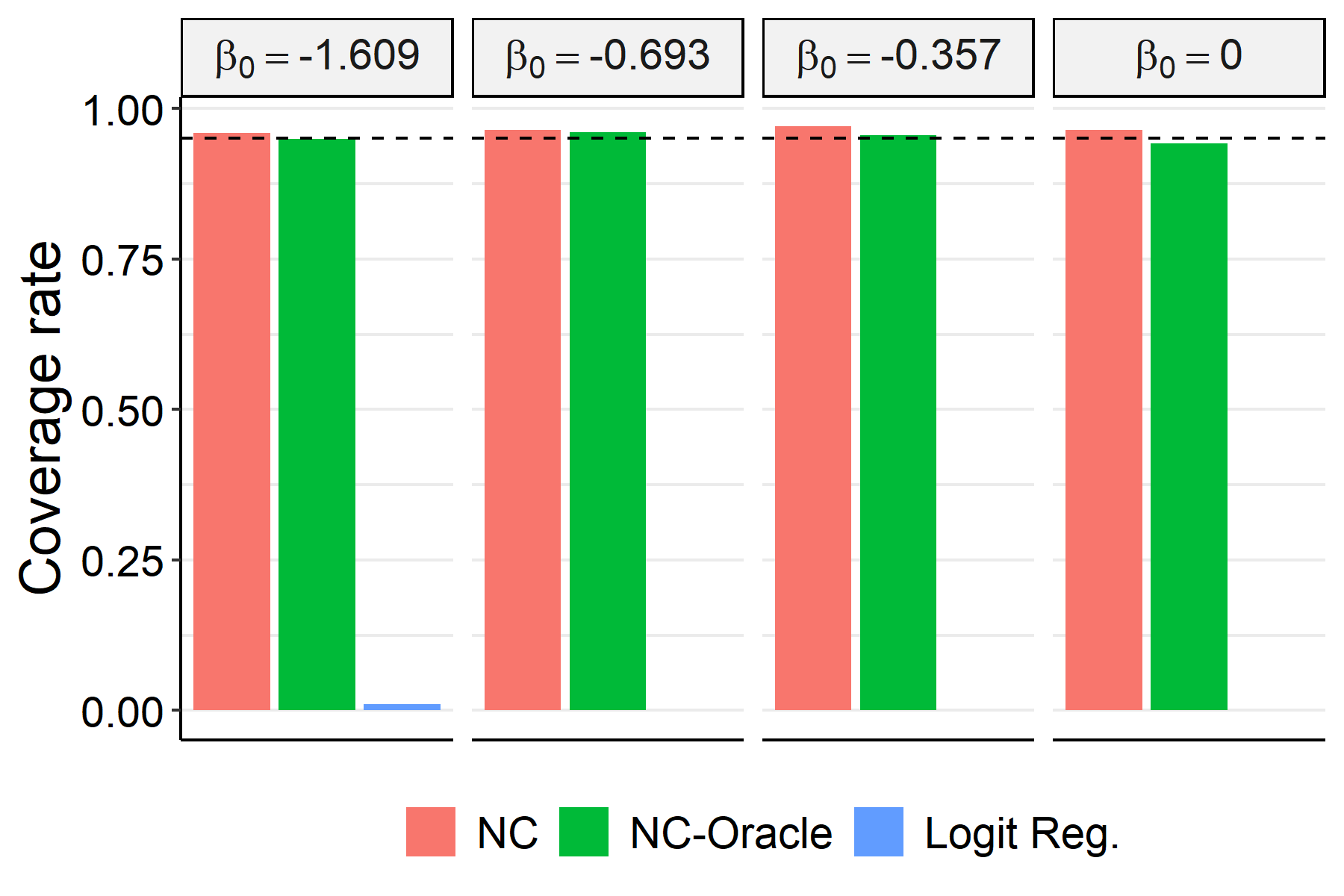}\\
    \multicolumn{2}{c}{(b)}\\
    \end{tabular}
    \caption{Bias (left) and coverage rates of 95\% confidence interval (right) for the oracle estimator (NC-Oracle), GMM estimator (NC-GMM) and logistic regression (Logit Reg.) with a (a) binary or (b) continuous unmeasured confounder.}
    \label{fig:sim}
\end{figure}

\section{Application}
\label{sec:application}

We applied our proposed method to a TND study of COVID-19 VE against COVID-19 infection nested in the University of Michigan Health System. The selected study sample includes patients who interacted with the University of Michigan Health System and experienced COVID-19 symptoms, had suspected exposure to COVID-19 virus, or sought to screen for COVID-19 infection, between April 5, 2021 and December 7, 2021. In addition, the selected test-positive subjects had at least one positive lab tests for COVID-19 infection after April 5. Vaccination history was obtained through electronic health records. A study subject was considered fully vaccinated if they received at least one dose of Johnson \& Johnson's Janssen vaccine or at least two doses of Moderna or Pfizer vaccine. If a subject tested positive before or within 14 days after their first dose of Janssen vaccine  or within 14 days after their second dose of Moderna or Pfizer vaccine, they were considered unvaccinated~\citep{moline2021effectiveness}.

We took immunization visits before December 2020 as NCE since COVID-19 vaccines were not available before December 2020 and immunization before was unlikely to affect the risk of COVID-19 infection; nor that of the selected NCOs we describe next. For NCO, we selected a binary indicator of having at least one of the following ``negative control outcome'' conditions after April 5, 2021: arm/leg cellulitis, eye/ear disorder,  gastro-esophageal disease, atopic dermatitis, and injuries. Such candidate NCE and NCO are likely to satisfy the requisite conditional independence conditions for them to be valid negative control variables and to be related to a patient's latent HSB. We adjusted for age groups ($<$18, between 18 and 60, or $\geq 60$), gender, race (white or non-white), Charlson comorbidity score $\geq$ 3, and the calendar month of a test-positive subject's first positive COVID test or a test-negative subject's last COVID test.  Table~\ref{tab:data_app} in Appendix~\ref{append:umhs-tables} summarizes the distribution of negative control variables, demographic variables and COVID-19 infection among vaccinated and unvaccinated subjects.

Because NCE is expected not to be associated with either the outcome or NCE in a fully adjusted analysis unless there is unmeasured confounding, we first fit regression models to detect presence of residual confounding bias. Conditioning on the baseline covariates, in both vaccinated and unvaccinated groups, NCE is significantly associated with COVID-19 infection  ($p<0.001$) and NCO ($p<0.001$) in corresponding adjusted logistic regression models, suggesting the presence of hidden biases~(See Appendix~\ref{append:umhs-tables} Table \ref{tab:reg-nce}, \ref{tab:reg-nco}).

We implemented Algorithm~\ref{alg:nc-tnd} to estimate VE. We specified a linear model for the treatment confounding bridge function with an interaction term between COVID-19 vaccination and the NCE, and set the function $m$ to include one (for an intercept term), COVID-19 vaccination, the NCO, and baseline covariates, as well as all two-way interactions. For comparison, we also implemented a logistic regression model, which gives an unbiased estimate of causal odds ratio under the no unmeasured confounding assumption,  adjusting for gender and age groups. In this case, The VE can be approximated by one minus the odds ratio of COVID-19 infection against vaccination, as COVID-19 infection rate is known to be low across strata in the target population.

 The double negative control Algorithm~\ref{alg:nc-tnd} estimated a causal log-RR associated with vaccination of -2.80 (95\% CI: -3.08, -2.54) for COVID-19 infection in the target population, and the estimated VE was  94.0\% (95\% CI: 92.1\%, 95.4\%). The logistic regression estimates for the same log-RR was -3.18 (95\% CI: -3.35, -3.01\%), yielding VE estimate of
95.8\% (95\% CI: 95.1\%, 96.5\%). Table~\ref{tab:umhs-nc-results} and \ref{tab:logit-reg-umhs} in Appendix~\ref{append:umhs-tables} give detailed output of Algorithm~\ref{alg:nc-tnd} and the logistic regression model respectively. Although there is significant evidence of hidden bias as summarized in Tables~\ref{tab:reg-nce} and \ref{tab:reg-nco}, the magnitude of detected bias is not appreciable, approximately $12\%$ on the log-risk ratio scale, corresponding to a 1.8\% difference on vaccine efficacy scale. A Hausman chi-squared test statistic~\citep{hausman1978specification} comparing the two estimates on the log-RR scale is 8.84, giving a p-value of 0.003, indicating that the double negative control VE estimate  is significantly smaller than that given by logistic regression. 

The VE estimated with a standard logistic regression model is larger than VE estimates from RCTs which reported VE=94$\%$ for Moderna and VE=95$\%$ for Pfizer vaccines, respectively~\citep{baden2021efficacy,sadoff2021safety,polack2020safety}. We hypothesize that this may be due to some degree of confounding by HSB and related factors, which our proposed double NC approach appears to account for to some extent, recovering VE estimates more consistent with those of RCTs.

\section{Discussion}
\label{sec:discussion}

In this article, we have introduced a statistical method for estimating vaccine effectiveness in a test-negative design. The approach leverages negative control variables to account for hidden bias due to residual confounding and/or selection mechanism into the TND sample. Negative control variables abound in practice, such as vaccination history which is routinely collected in insurance claims and electronic health records. Hence the proposed method may be particularly useful in such real world settings to obtain improved estimates of vaccine effectiveness. 

The TND is a challenging setting in causal inference where selection bias and unmeasured confounding co-exist, selection is outcome-dependent, and unmeasured confounders also impact selection.  As a result, the causal effect of interest is in general not identified from such studies~\citep{cai2012identifying}. Nevertheless, we establish that progress can be made under a semiparametric multiplicative model, provided the outcome is rare in the target population, and double negative control variables are available. To this end, this article showcases the potential power of negative control methods and proximal causal inference in epidemiologic research~\citep{shi2020selective,tchetgen2020introduction}.

We focused on the outpatient TND, where recruitment is restricted to subjects who seek care voluntarily. TNDs have also been applied to inpatient settings for studying VE against, for example, flu hospitalization~\citep{foppa2016case,feng2016influenza}. In inpatient TNDs, differential access to healthcare and underlying health characteristics between vaccinated and unvaccinated subjects are likely the main causes of confounding bias~\citep{feng2016influenza}. Our methods are still applicable in such settings, but negative control variables should be selected to be relevant to the source of unmeasured confounding mechanism. For example, previous vaccination and hospitalization outside the flu season or hospitalization due to other flu-like illnesses are viable candidate NCE and NCO, respectively~\citep{jackson2006evidence}.

 Our approach is suitable for post-market TND studies where real-world vaccine effectiveness is of interest and vaccination history is obtained retrospectively, possibly through electronic health records. For vaccine efficacy in a controlled trial setting, \citet{wang2022randomization} recently developed estimation and inference of RR in cluster-randomized TND, aiming to correct for bias due to differential HSB induced by the intervention being unblinded. Because of randomization, they considered HSB as a post-treatment variable and proposed a log-contrast estimator which corrects for selection bias by leveraging a valid test-negative outcome, under an assumption that either (i) the vaccine does not have a causal effect in the population, and the causal impact of vaccination on selection is equal for test-positive and -negative subsamples; or (ii) among care seekers, the incidence of test-negative outcomes does not differ between vaccinated and unvaccinated, and the intervention effect among care seekers is generalizable to the whole population. We note that even under randomization, identification conditions given in Section~\ref{sec:a-s-arrow} are neither stronger nor weaker than those of \citet{wang2022randomization} described above, as neither set of assumptions appear to imply the other. An important advantage of our proposed methods is that they can be used to account for selection bias in a TND study irrespective of randomization.

 Our methods target RR as a measure of VE instead of the more common OR~\citep{jackson2006evidence,sullivan2016theoretical}. These two measures are approximately equal for rare infections.  \citet{schnitzer2022estimands} recently considered estimation of a marginal causal RR in the TND sample and justified the use of an inverse probability of treatment weighted (IPTW) estimator in a setting in which an unmeasured common cause of infection and selection into the TND sample does not cause vaccination (and thus there is no unmeasured confounding). Instead, our methods allow for an unmeasured common cause of vaccination, infection and selection into the TND sample; however in order to estimate a causal RR, we invoke both, an assumptions of no effect modification by an unmeasured confounder, and a rare-disease condition. As we establish, the latter assumption is not needed if there is no vaccine effect against infection outcome.  In Section~\ref{sec:a-s-arrow}, we establish that under a homogeneous OR vaccine effect measure condition, and an alternative definition of the treatment bridge function, our methods can identify a causal effect of the vaccine on the odds ratio scale without invoking the rare disease condition.
 
 Throughout the article, we have assumed diagnostic tests are accurate and individuals who seek care are sparsely distributed, such that the vaccination of a given subject in the TND sample does not protect another study subject  from infection, , i.e. there is no interference in the TND sample, a common assumption in TND literature. This assumption may be violated if members of the same households present in the ER in which case block interference must be accounted using results from interference literature~\citep{hudgens2008toward,tchetgen2012causal}.  Sensitivity analysis may be considered to evaluate how violation of these assumptions can potentially bias inferences about VE.
 
\if1\blind{
\section*{Acknowledgements}
 The authors thank Dr. Lili Zhao, Dr. Chen Shen and University of Michigan Precision Health Initiative for providing the University of Michigan Health System data.
}\fi

%\bibliographystyle{agsm}
%\bibliography{references}
\pagebreak
\appendix

\section{Proof of Proposition~\ref{prop:selection_bias}}\label{append:oracle-ee}

We need to prove
 $$E[(-1)^{1-A}\dfrac{1}{P(A|U, X)}Y\exp(-\beta_0A)|S=1]=0$$

Since \begin{align*}
    & E\left[(-1)^{1-A}\dfrac{1}{P(A|U, X)}Y\exp(-\beta_0A)\bigg |S=1\right]\\
    =& E\left[(-1)^{1-A}\dfrac{1}{P(A|U, X)}\exp(-\beta_0A)\bigg|Y=1,S=1\right]P(Y=1|S=1)\\
    =& E\left\{E\left[(-1)^{1-A}\dfrac{1}{P(A|U, X)}\exp(-\beta_0A)\bigg|U, X, Y=1, S=1\right]\bigg|Y=1,S=1\right\}P(Y=1|S=1),
\end{align*}
the result immediately follows if we can show that
\begin{equation}E\left[(-1)^{1-A}\dfrac{1}{P(A|U, X)}\exp(-\beta_0A)\bigg|U, X, Y=1,S=1\right]=0.\label{eq:propensity-conditional}
\end{equation}

By Assumption~\ref{assump:trt-indep-samp}, the left-hand side of (\ref{eq:propensity-conditional}) equals
\begin{align*}
    & E\left[(-1)^{1-A}\dfrac{1}{P(A|U, X)}\exp(-\beta_0A)\bigg|U, X, Y=1\right].
\end{align*}

We further have
\begin{align*}
    & E\left[(-1)^{1-A}\dfrac{1}{P(A|U, X)}\exp(-\beta_0A)\bigg|U, X, Y=1\right]\\
    =& \sum_{a=0}^1 (-1)^{1-a}\dfrac{1}{P(A=a|U, X)}\exp(-\beta_0 a)P(A=a|U, X, Y=1)\\
    =& \sum_{a=0}^1 (-1)^{1-a}\dfrac{1}{P(A=a|U, X)}\exp(-\beta_0 a)\dfrac{P(A=a|U, X)P(Y=1|A=a, U, X)}{P(Y=1|U, X)}\\
    \stackrel{A.\ref{assump:no-eff-mod}}{=}& \sum_{a=0}^1 (-1)^{1-a}\exp(-\beta_0a)\dfrac{\exp(\beta_0 a)P(Y=1|A=0, U, X)}{P(Y=1|U, X)}\\
    =& \sum_{a=0}^1(-1)^{1-a}\dfrac{P(Y=1|A=0, U, X)}{P(Y=1|U, X)}\\
    =& 0.
\end{align*}

\pagebreak
\section{Existence of solutions to \eqref{eq:trt-bridge-defn}}\label{append:trt-bridge-existence}
In this section, we provide the conditions of existence of solutions to \eqref{eq:trt-bridge-defn}. The conditions for \eqref{eq:trt-bridge-identification} can be similarly derived. The results in this section directly adapted from Appendix B of \citet{cui2020semiparametric}.

Let $L^2\{F(t)\}$ denote the Hilbert space of all square-integrable functions of $t$ with respect to distribution function $F(t)$, equiped with inner product $\langle g_1, g_2\rangle=\int g_1(t)g_2(t)dF(t)$. Let $T_{a,x}$ denote the operator $L^2\{F(z|a,x)\}\rightarrow L^2\{F(u|a, x)\}$, $T_{a,x}q=E[q(Z)|A=a, U=u, X=x]$ and let $(\lambda_{a,x,n},\varphi_{a,x,n},\phi_{a,x,n})$ denote a singular value decomposition of $T_{a,x}$. The solution to \eqref{eq:trt-bridge-defn} exists if:
\begin{enumerate}[(1)]
    \item $\int\int f(z|a, u, x)f(u|a,z,x)dzdu<\infty$;
    \item $\int P^{-2}(A=a|U=u,X=x)f(u|a,x)du<\infty$;
    \item $\sum_{n=1}^\infty \lambda_{a,x,n}^{-2}|\langle P^{-1}(A=a|U=u,X=x),\phi_{a,x,n}\rangle|^2<\infty$.
\end{enumerate}

\pagebreak
\section{Proof of Theorem~\ref{thm:causal-rr-identifiability}}
\label{append:ee-rr-unbiasedness}

We need to prove $$E[(-1)^{1-A}q(A,Z,X)Y\exp(-\beta_0A)|S=1]=0$$

Since \begin{align*}
    & E[(-1)^{1-A}q(A, Z, X)Y\exp(-\beta_0A)|S=1]\\
    =& E[(-1)^{1-A}q(A, Z, X)\exp(-\beta_0A)|Y=1,S=1]P(Y=1|S=1)\\
    =& E\left\{E\left[(-1)^{1-A}q(A, Z, X)\exp(-\beta_0A)|U, X, Y=1, S=1\right]|Y=1,S=1\right\}P(Y=1|S=1),
\end{align*}
it suffices to show that

$$E[(-1)^{1-A}q(A, Z, X)\exp(-\beta_0A)|U, X, Y=1, S=1]=0.$$

By Assumption~\ref{assump:trt-indep-samp}, the left-hand side is
\begin{align*}
    & E[e^{-\beta_0 A}(-1)^{1-A}q(A, Z, X)|U, X,Y=1] \\
    =& E\{e^{-\beta_0 A}(-1)^{1-A}E[q(A, Z, X)|A, U, X, Y=1]|U,X, Y=1\}\\
    \stackrel{A.\ref{assump:nce}}{=}& E\{e^{-\beta_0 A}(-1)^{1-A}E[q(A, Z, X)|A, U,X]|U, X, Y=1\}\\
    \stackrel{A.\ref{assump:trt-bridge}}{=}& E\{e^{-\beta_0 A}(-1)^{1-A}\dfrac{1}{P(A|U,X)}|U,X, Y=1\}\\
    =& 0
\end{align*}
The last equality is proved in Appendix~\ref{append:oracle-ee}.

\pagebreak

\section{Treatmeng bridge function and estimation with categorical NCE, NCO and unmeasured confounders.}
\label{append:eg-categorical}
We consider a categorical unmeasured confounder $U$ with categories $u_1,\dots,u_J$ and NCE $Z$ with categories $z_1,\dots, z_I$. We assume there are no other covariates $X$. We write $p_{ia.k}=P(Z=z_i, A=a|U=u_j)$ for $i=1,\dots, I$, $a=0,1$, and $j=1,\dots, J$.

Similar to before, the treatment confounding bridge function $q(a, z)$ should satisfy 

\begin{equation}
\sum_{i=1}^I\,p_{ia.j}q(a, z_i)=1
\end{equation}
for all $i,j$. Therefore, any solution of the following equation system, if exists, is a treatment confounding bridge function:
\begin{equation}
\begin{array}{c}
p_{1a.1}q(a, z_1)  + p_{2a.1}q(a, z_2) + \dots + p_{Ia.1}q(a,z_I) = 1\\
p_{1a.2}q(a, z_1)  + p_{2a.2}q(a, z_2) + \dots + p_{Ia.2}q(a, z_I) = 1\\
\dots\\
p_{1a.J}q(a, z_1)  + p_{2a.J}q(a, z_2) + \dots + p_{Ia.J}q(a, z_I) = 1
\end{array}\label{eq:categorical-trt-bridge}
\end{equation}
for $a=0,1$.   We denote the probability matrix
\begin{equation}\label{eq:prob-mat-categorical}P_{Z,A\mid U}=\begin{pmatrix}
p_{1a.1} & p_{2a.1} & \dots & p_{Ia.1}\\
p_{1a.2} & p_{2a.2} & \dots & p_{Ia.2}\\
\vdots & & & \vdots\\
p_{1a.J} & p_{2a.J} & \dots & p_{Ia.J}
\end{pmatrix}.\end{equation}
A treatment confounding bridge function exists if the matrix $P_{Z,A\mid U}$ is invertible.

 Suppose besides a categorical NCE $Z$ with levels $z_1,\dots,z_I$,  we also have a categorical NCO $W$ with levels $w_1,\dots, w_K$. The integral equation~\eqref{eq:trt-bridge-identification} is equivalent to the linear equation system

\begin{equation}
\begin{array}{c}
p'_{1a.1}q^*(a, z_1)  + p'_{2a.1}q^*(a, z_2) + \dots + p'_{Ia.1}q^*(a,z_I) = 1\\
p'_{1a.2}q^*(a, z_1)  + p'_{2a.2}q^*(a, z_2) + \dots + p'_{Ia.2}q^*(a, z_I) = 1\\
\dots\\
p'_{1a.K}q^*(a, z_1)  + p'_{2a.K}q^*(a, z_2) + \dots + p'_{Ia.K}q^*(a, z_I) = 1
\end{array}
\end{equation}
for $a=0,1$, where $p'_{ia.k}=P(Z=z_i, A=a|W=w_k,Y=0,S=1)$. We denote that matrix
\begin{equation}\label{eq:prob-mat-categorical-est}P'(a)=\begin{pmatrix}
p'_{1a.1} & p'_{2a.1} & \dots & p'_{Ia.1}\\
p'_{1a.2} & p'_{2a.2} & \dots & p'_{Ia.2}\\
\vdots & & & \vdots\\
p'_{1a.K} & p'_{2a.K} & \dots & p'_{Ia.K}
\end{pmatrix}.\end{equation}

Then Assumption~\ref{assump:solvability} is equivalent to the condition that $P'(a)$ is invertible with $I=K$ for $a=0,1$, in which case ~\eqref{eq:trt-bridge-identification} has a unique solution 
$$(q^*(a,z_1),\dots, q^*(a,z_I))^T=[P'(a)]^{-1}\mathbf 1_I.$$
The probabilities $p'_{ia.k}$'s can all be estimated from the study data.

\pagebreak
\section{Derivation of the treatment confounding bridge function in Example~\ref{eq:contn}}

\label{append:trt-bridge-cont}
By Assumption~\ref{assump:trt-bridge}, the treatment confounding bridge function $q(A, Z, X)$ should satisfy 
$$E[q(A, Z, X)|U=u, A=a, X=x] = \dfrac{1}{P(A=a|U=u, X=x)}$$

for all $a$, $u$ and $X$.

We write $q(A, Z, X) = 1 + r(Z, A, X)$, then
$$E[r(Z, a, X)|U, A=a, X] = \dfrac{1 - P(A=a|U, X)}{P(A=a|U, X)}.$$

Consider 
$$r(Z, A, X) = \exp((-1)^{A}(\tau_0 + \tau_1A + \tau_2Z + \tau_3 X)),$$

then 
\begin{align*}
    &\dfrac{P(A=1|U=u, X=x)}{1 - P(A=1|U=u, X=x)}\\ =& \int r(z, 0, x)f(z|u, x, A=0)dz\\
    =& \int \exp(\tau_0 + \tau_2z + \tau_3x)\dfrac{1}{\sqrt{2\pi\sigma_z^2}}\exp\left(-\dfrac{(z - \mu_{0z}  - \mu_{UZ}u - \mu_{XZ}x)^2}{2\sigma_z^2}\right)dz\\
    =& \exp(\tau_0 + \tau_3x)\exp\left(\tau_2(\mu_{0Z}+\mu_{UZ}u + \mu_{XZ}x) + \dfrac{\sigma_z^2\tau_2^2}{2}\right)\\
    =& \exp(\tau_0 + \tau_2\mu_{0Z} + \dfrac{\sigma_z^2\tau_2^2}{2} + \tau_2\mu_{UZ} u + (\tau_3+\mu_{XZ}\tau_2)x)
\end{align*}

and
\begin{align*}
    &\dfrac{1 - P(A=1|U=u, X=x)}{P(A=1|U=u, X=x)}\\ =& \int r(z, 1, x)f(z|u, x, A=1)dz\\
    =& \int \exp(-\tau_0 - \tau_1 - \tau_2z - \tau_3x)\dfrac{1}{\sqrt{2\pi\sigma_z^2}}\exp\left(-\dfrac{(z - \mu_{0Z} -\mu_{AZ}- \mu_{UZ}u - \mu_{XZ}x)^2}{2\sigma_z^2}\right)dz\\
    =& \exp(-\tau_0 - \tau_1 - \tau_3x)\exp\left(-\tau_2(\mu_{0Z} + \mu_{AZ} + \mu_{UZ}u + \mu_{XZ}x) + \dfrac{\sigma_z^2 \tau_2^2}{2}\right)\\
    =& \exp(-\tau_0 - \tau_1 - \tau_2\mu_{0Z} -\tau_2\mu_{AZ}+ \dfrac{\sigma_z^2\tau_2^2}{2} - \tau_2\mu_{UZ}u - (\tau_3 + \tau_2\mu_{XZ})x)
\end{align*}
This requires that 
$$\tau_1 + \tau_2 \mu_{AZ} = \sigma_z^2\tau_2^2$$

Because
$$P(A=1|U, X) = \expit(\mu_{0A} + \mu_{UA} U + \mu_{XA} X),$$
we conclude the parameters in the bridge function are
$$\tau_2 = \mu_{UA} / \mu_{UZ},$$
$$\tau_3 = \mu_{XA} - \mu_{XZ}\tau_2 = \mu_{XA} - \mu_{XZ}\mu_{UA}/\mu_{UZ}$$
$$\tau_1 = \sigma_z^2\tau_2^2 - \tau_2 \mu_{AZ} = \dfrac{\sigma_z^2\mu_{UA}^2}{\mu_{UZ}^2} - \dfrac{\mu_{UA}\mu_{AZ}}{\mu_{UZ}}$$
$$\tau_0 = \mu_{0A} - \tau_2\mu_{0Z} - \dfrac{\sigma_z^2\alpha_2^2}{2} = \mu_{0A} - \dfrac{\mu_{UA}\mu_{0Z}}{\mu_{UZ}} - \dfrac{\sigma_z^2\mu_{UA}^2}{2\mu_{UZ}^2}$$

\pagebreak
\section{Proof of Theorem~\ref{thm:trt-bridge-identification} and Corollary~\ref{cor:null-preserving}}
\label{append:trt-bridge-identification}
We first introduce a few properties due to the rare disease assumption:

\begin{lemma}\label{lemma:rare-disease-approximation} Under Assumptions~\ref{assump:trt-indep-samp},~\ref{assump:nce},~\ref{assump:nco} and~\ref{assump:rare-disease}, for every $a$, $w$, $u$ and $x$, we have
\begin{enumerate}[(a)]
    \item \begin{align*}
        &P(Y=1|W=w, U=u, X=x)< \delta, && P(Y=1|A=a, U=u, X=x) < \delta,\\
        &P(Y=1|U=u, X=x) <\delta.
    \end{align*}
    \item $$1-\delta < \dfrac{P(A=a|U=u, X=x, Y=0, S=1)}{P(A=a|U=u, X=x)} < \dfrac{1}{1-\delta}$$
    
    \item $$(1-\delta)^2 <\dfrac{P(A=a|W=w, U=u, X=x, Y=0, S=1)}{P(A=a|U=u, X=x, Y=0, S=1)} < \dfrac{1}{(1-\delta)^2}.$$

    \item $$f(z|A=a, U=u, X=x)= f(z|W=w, A=a, U=u, X=x, Y=0, S=1).$$
\end{enumerate}
\end{lemma}

\begin{proof}
\begin{enumerate}[(a)]
    \item For every $a$, $w$, $u$ and $x$, we have
    \begin{align*}
        &P(Y=1|W=w, U=u, X=x)\\ &= \sum_{a}P(Y=1|A=a, W=w, U=u, X=x)P(A=a|W=w, U=u, X=x)\\
        &\stackrel{A.\ref{assump:rare-disease}}{<}\delta \sum_aP(A=a|W=w, U=u, X=x)\\
        &= \delta.
    \end{align*}
    The rest follows similarly.
    \item For every $a$, $u$ and $x$, we have 
    \begin{align*}
        P(A=a|U=u, X=x, Y=0) &= P(A=a|U=u, X=x)\times \dfrac{P(Y=0|A=a, U=u, X=x)}{P(Y=0|U=u, X=x)}.
    \end{align*}
    By Lemma~\ref{lemma:rare-disease-approximation}(a), we have
    $$1-\delta < \dfrac{P(Y=0|A=a, U=u, X=x)}{P(Y=0|U=u, X=x)} = \dfrac{P(A=a|U=u, X=x, Y=0)}{P(A=a|U=u, X=x)}<\dfrac{1}{1-\delta}.$$
    The result follows by noticing that $P(A=a|U=u, X=x, Y=0)=P(A=a|U=u, X=x, Y=0, S=1)$ due to Assumption~\ref{assump:trt-indep-samp}.
    \item For every $a$, $w$, $u$ and $x$, we have
\begin{align*}
    &P(A=a|W=w, U=u, X=x, Y=0) \\&= P(A=a|W=w, U=u, X=x)\times \dfrac{P(Y=0|A=a, W=w, U=u, X=x)}{P(Y=0|W=w, U=u, X=x)}\\
    &\stackrel{A.\ref{assump:nco}(a)}{=}P(A=a|U=u, X=x) \times \dfrac{P(Y=0|A=a, W=w, U=u, X=x)}{P(Y=0|W=w, U=u, X=x)}\\
    &= P(A=a|U=u, X=x, Y=0)\times \dfrac{P(Y=0|U=u, X=x)}{P(Y=0|A=a, U=u, X=x)}\\&\quad \times \dfrac{P(Y=0|A=a, W=w, U=u, X=x)}{P(Y=0|W=w, U=u, X=x)}
\end{align*}
By Lemma~\ref{lemma:rare-disease-approximation}(a), we have
$$1-\delta < \dfrac{P(Y=0|U=u, X=x)}{P(Y=0|A=a, U=u, X=x)}<\dfrac{1}{1-\delta}$$ and $$1 - \delta < \dfrac{P(Y=0|A=a, W=w, U=u, X=x)}{P(Y=0|W=w, U=u, X=x)}< \dfrac{1}{1-\delta}.$$

We therefore have
\begin{align*}
    (1-\delta)^2 <\dfrac{P(A=a|W=w, U=u, X=x, Y=0)}{P(A=a|U=u, X=x, Y=0)} < \dfrac{1}{(1-\delta)^2}.
\end{align*}
Finally, by Assumptions~\ref{assump:trt-indep-samp} and~\ref{assump:nco}(c), we have
\begin{align*}
    &P(A=a|U=u, X=x, Y=0) = P(A=a|U=u, X=x, Y=0, S=1),\\
    &P(A=a|W=w, U=u, X=x, Y=0) = P(A=a|W=w, U=u, X=x, Y=0, S=1).\\
\end{align*}
We conclude that 
\begin{align*}
    (1-\delta)^2 <\dfrac{P(A=a|W=w, U=u, X=x, Y=0, S=1)}{P(A=a|U=u, X=x, Y=0, S=1)} < \dfrac{1}{(1-\delta)^2}.
\end{align*}

\item \begin{align*}
    &f(z|W=w, A=a, U=u, X=x, Y=0, S=1)\\
    \stackrel{A.\ref{assump:nco}(c)}{=}&f(z|W=w, A=a, U=u, X=x, Y=0)\\
    \stackrel{A.\ref{assump:nco}(b)}{=}&f(z|A=a, U=u, X=x)
\end{align*}
\end{enumerate}

\end{proof}

Therefore, we have
 \begin{align*}
    & E\{q(a, Z, X)|A=a, W, X, Y=0,S=1\}\\
    =& \int q(a, z, X)f(z|A=a, W, X, Y=0, S=1)\,\mathrm dz\\
    =& \int \int q(a, z, X)f(z|A=a, W, U=u, X, Y=0, S=1)f(u|A=a, W, X, Y=0, S=1)\,\mathrm dz\,\mathrm du\\
    \stackrel{L.\ref{lemma:rare-disease-approximation}(d)}{=}& \int \left\{\int q(a, z, X)f(z|A=a, U=u, X)\,\mathrm du\right\}f(u|A=a, W, X, Y=0, S=1)\,\mathrm dz\\
    \stackrel{A.\ref{assump:trt-bridge}}{=}& \int \dfrac{1}{P(A=a|U=u, X)}f(u|A=a, W, X, Y=0, S=1)du\\
    \stackrel{L.\ref{lemma:rare-disease-approximation}(b)}{<}& \dfrac{1}{1-\delta}\int \dfrac{1}{P(A=a|U=u, X, Y=0, S=1)}f(u|A=a, W, X, Y=0, S=1)du\\
   \stackrel{L.\ref{lemma:rare-disease-approximation}(c)}{<}& \dfrac{1}{(1-\delta)^3}\int \dfrac{1}{P(A=a|W, U=u, X, Y=0, S=1)}f(u|A=a, W, X, Y=0, S=1)du\\
    =& \dfrac{1}{(1-\delta)^3}\int \dfrac{f(u|W, X, Y=0, S=1)f(u|A=a, W, X, Y=0, S=1)}{P(A=a|W, X, Y=0, S=1)f(u|A=a, W, X, Y=0, S=1)}\,\mathrm du\\
    =&\dfrac{1}{(1-\delta)^3}\dfrac{1}{P(A=a|W, X, Y=0, S=1)}\int f(u|W, X, Y=0, S=1)\,\mathrm du\\
    =& \dfrac{1}{(1-\delta)^3}\dfrac{1}{P(A=a|W, X, Y=0, S=1)}
\end{align*}

and 

\begin{align*}
    & E\{q(a, Z, X)|A=a, W, X, Y=0,S=1\}\\
    =& \int \dfrac{1}{P(A=a|U=u, X)}f(u|A=a, W, X, Y=0, S=1)\,\mathrm du\\
    \stackrel{L.\ref{lemma:rare-disease-approximation}(b)}{>}& (1-\delta)\int \dfrac{1}{P(A=a|U=u, X, Y=0, S=1)}f(u|A=a, W, X, Y=0, S=1)\,\mathrm du\\
   \stackrel{L.\ref{lemma:rare-disease-approximation}(c)}{>}& (1-\delta)^3\int \dfrac{1}{P(A=a|W, U=u, X, Y=0, S=1)}f(u|A=a, W, X, Y=0, S=1)\,\mathrm du\\
    =& (1-\delta)^3\int \dfrac{f(u|W, X, Y=0, S=1)f(u|A=a, W, X, Y=0, S=1)}{P(A=a|W, X, Y=0, S=1)f(u|A=a, W, X, Y=0, S=1)}\,\mathrm du\\
    =&(1-\delta)^3\dfrac{1}{P(A=a|W, X, Y=0, S=1)}\int f(u|W, X, Y=0, S=1)\,\mathrm du\\
    =& (1-\delta)^3\dfrac{1}{P(A=a|W, X, Y=0, S=1)}.
\end{align*}

To prove Corollary~\ref{cor:null-preserving}, we have
\begin{align*}
     & E\{q(a, Z, X)|A=a, W, X, Y=0,S=1\}\\
    =& \int \dfrac{1}{P(A=a|U=u, X)}f(u|A=a, W, X, Y=0, S=1)\,\mathrm du\\
    =& \int \dfrac{P(A=a|U=u,X,Y=0,W, S=1)}{P(A=a|U=u, X)P(A=a|W, X, Y=0, S=1)}f(u| W, X, Y=0, S=1)\,\mathrm du\\
    \stackrel{A.\ref{assump:nco}}{=}&\int \dfrac{P(A=a|U=u,X,Y=0,W)}{P(A=a|U=u, X)P(A=a|W, X, Y=0, S=1)}f(u| W, X, Y=0, S=1)\,\mathrm du
\end{align*}
If $A\indep Y|U, X$, then together with Assumption~\ref{assump:nco} we have $P(A=a|U=u, X, Y=0, W)=P(A=a|U=u, X)$, whereby the above equals
$$\dfrac{1}{P(A=a|W, X, Y=0, S=1)}\int f(u|W, X, Y=0, S=1)du = \dfrac{1}{P(A=a|W, X, Y=0, S=1)}.$$

\pagebreak
\section{Discussion of Assumption~\ref{assump:solvability}}\label{append:completeness}

Similar to \eqref{eq:trt-bridge-defn}, Equation~\eqref{eq:trt-bridge-identification} defines a Fredholm integral equation of the first kind, yet only involves observed data. Although a treatment confounding bridge function $q(A, Z, X)$ must satisfy~\eqref{eq:trt-bridge-identification}, there is no guarantee that solving~\eqref{eq:trt-bridge-identification} gives a treatment confounding bridge function if multiple solutions exist. When a solution to~\eqref{eq:trt-bridge-identification} exists, we give the following assumptions that guarantee  the uniqueness of solution.

\begin{assumption}[Completeness]
\label{assump:completeness} $\quad$

\begin{enumerate}[(a)]
    \item For any square-integrable function $g$, if $E[g(Z)|A, U, X]=0$ almost surely, then $g(Z)=0$ almost surely;
    \item For any square-integrable function $h$, if $E[h(U)|A, W, X, Y=0, S=1]$ almost surely, then $h(U)=0$ almost surely.
\end{enumerate}

\end{assumption}

Intuitively, Assumption~\ref{assump:completeness} is a statement on the information contained in the unmeasured confounders vs. in the negative control variables -- Assumption~\ref{assump:completeness}(a) requires the confounders $U$ are informative enough about $Z$ in the target population and Assumption~\ref{assump:completeness}(b) requires the NCO $W$ is informative enough about $U$ in the subgroup $Y=0,S=1$, so that no information is lost when taking the two conditional expectations. Completeness conditions similar to Assumption~\ref{assump:completeness} were originally introduced by Lehmann and Scheffe to identify the so-called unbiase minimum risk estimator~\citep{lehmann2012completeness1,lehmann2012completeness2}. In econometrics and causal inference literature, completeness conditions have been employed to achieve identifiability for a variety of nonparametric or semiparametric models, such as instrumental variable regression~\citep{newey2003instrumental,d2011completeness}, measurement error models~\citep{hu2008instrumental}, synthetic control~\citep{shi2021theory}, and previous works in negative control methods~\citep{miao2018confounding,cui2020semiparametric,ying2021proximal}. The completeness condition holds for a wide range of distributions. \citet{newey2003instrumental} and \citet{d2011completeness} provided justifications in exponential families and discrete distributions with finite support.  \citet{andrews2011examples} constructed a broad class of bivariate distributions that satisfy the completeness condition.

Assumption~\ref{assump:completeness} have several immediate consequences:

\begin{proposition}\label{prop:completeness}
\begin{enumerate}[(a)]
    \item Under Assumptions~\ref{assump:nce}, \ref{assump:nco} and \ref{assump:completeness}, for any square integrable function $g$ such that $E[g(Z)|A, W, X, S=0, Y=1]=0$ almost surely, then $g(Z)=0$ almost surely.
    \item Under Assumptions~\ref{assump:completeness}(a) and \ref{assump:ident}-\ref{assump:trt-bridge}, the treatment confounding bridge function is unique. That is, if two square-integrable functions $q(A, Z, X)$ and $q_1(A, Z, X)$ satisfy
    \begin{align*}&E[q(a, Z, x)|A=a, U=u, X=x]=E[q_1(a, Z, x)|A=a, U=u, X=x]\\=&\dfrac{1}{P(A=a|U=u, X=x)}\end{align*}
    for all $a,u,x$ almost surely, then $q(A, Z, X)=q_1(A, Z, X)$ almost surely.
    \item Under Assumptions 
    \ref{assump:nce}, \ref{assump:nco}, \ref{assump:solvability} and \ref{assump:completeness}, Equation~\eqref{eq:trt-bridge-identification} has a unique solution $q^*(A, Z, X)$.
\end{enumerate}

\end{proposition}

We prove Proposition~\ref{prop:completeness} below. Proposition~\ref{prop:completeness} states that the completeness conditions in Assumption~\ref{assump:completeness} and the definitions of the negative control variables lead to a third completeness condition that only involves the study data. Proposition~\ref{prop:completeness} states the uniqueness of the treatment confounding bridge function, $q(A, Z, X)$, and the solution to \eqref{eq:trt-bridge-identification}, $q^*(A, Z, X)$. The function $q^*(A, Z, X)$ can therefore be identified from the study data alone and is a good approximation of $q(A, Z, X)$ by Theorem~\ref{thm:trt-bridge-identification}.

\begin{proof}
\begin{enumerate}[(a)]
    \item Suppose a square-integrable function $g(Z)$ satisfies 
    $$E[g(Z)|A, W, X, Y=0, S=1]=0\quad\mbox{almost surely}.$$
    
    The left-hand side equals
    \begin{align*}
        &E\left\{E\left[g(Z)|U, A, W, X, Y=0, S=1 \right]|A, W, X, Y=0, S=1\right\}\\
        \stackrel{A.\ref{assump:nco}}{=}& E\left\{E\left[g(Z)|U, A, X, Y=0, S=1 \right]|A, W, X, Y=0, S=1\right\}\\
        \stackrel{A.\ref{assump:nce}}{=}& E\left\{E\left[g(Z)|U, A, X \right]|A, W, X, Y=0, S=1\right\}
    \end{align*}
    
    By Assumption~\ref{assump:completeness}(b), $E\left[g(Z)|U, A, X \right]=0$ almost surely. Then by Assumption~\ref{assump:completeness}(a), $g(Z)=0$ almost surely.
    \item We have $$E[q(A, Z, X)-q_1(A, Z, X)|A, U, X]=0$$ 
    almost surely. By Assumtion~\ref{assump:completeness}(a), we have
    $$q(A, Z, X)-q_1(A, Z, X)=0$$
    almost surely, or $q(A, Z, X)=q_1(A, Z, X)$ almost surely.
    \item If two square-integrable functions $q^*(A, Z, X)=q_1^*(A, Z, X)$ satisfies
    \begin{align*}
        &E[q^*(A, Z, X)|A, W, X, Y=0, S=1] = E[q_1^*(A, Z, X)|A, W, X, Y=0, S=1]\\
        =& \dfrac{1}{P(A|W, X, Y=0, S=1)}
    \end{align*}
    almost surely, then
    $$E\left\{[q^*(A, Z, X)-q_1^*(A, Z, X)]|A, W, X, Y=0, S=1\right\}=0.$$
    Under Assumptions \ref{assump:nce}, \ref{assump:nco} and \ref{assump:completeness}, Proposition~\ref{prop:completeness}(a) holds and therefore
    $$q^*(A, Z, X)=q_1^*(A, Z, X)$$
    almost surely.
\end{enumerate}
\end{proof}

\pagebreak
\section{Proof of Corollary~\ref{cor:trt-bridge-ee} and further discussion}
\label{append:ee-trt-bridge}

We first prove equation~(\ref{eq:trt-bridge-ee-unbiasedness})
\begin{align*}
    & E[m(W, A, X)q^*(A, Z, X)-m(W, 1, X) - m(W, 0, X)|Y=0, S=1]\\
    =& E\left\{m(W, A, X)E\left[q^*(A, Z, X)|W, A, X, Y=0, S=1\right] - m(W, 1, X) - m (W, 0, X)|Y=0, S=1\right\}\\
    \stackrel{\eqref{eq:trt-bridge-identification}}{=}& E\left\{m(W, A, X)\dfrac{1}{P(A|W, X, Y=0, S=1)}-m(W, 1, X) - m(W, 0, X)|Y=0, S=1\right\}\\
    =& E\left\{ E\left[m(W, A, X)\dfrac{1}{P(A|W, X, Y=0, S=1)} |W, X, Y=0, S=1\right]- m(W, 1, X) - m(W, 0, X)|Y=0, S=1\right\}\\
    =& E\left\{m(W, 1, X) + m(W, 0, X) - m(W, 1, X) - m(W, 0, X)|Y=0, S=1\right\}\\
    =0
\end{align*}

In fact, one can show that any regular and asymptotically normal estimator of $\tau$ that satisfies~(\ref{eq:trt-bridge-identification}) has influence function of the form
\begin{align*}
    IF(W, Z, A, X) &= -\bigg\{E\bigg[\dfrac{\partial q(A, Z, X;\tau)}{\partial \tau}\bigg|_{\tau = \tau_0}m(W, A, X)(1-Y)\bigg| S=1\bigg]\bigg\}^{-1}(1-Y)\\&\qquad \bigg(m(W, A, X)q(A, Z, X) - m(W, 1, X) - m(W, 0, X)\bigg)
\end{align*}
for an arbitrary function $m(W, A, X)$. Therefore, any regular and asymptotically normal estimator of $\tau$ corresponds to the solution of the estimating equation~(\ref{eq:ee-trt-bridge}) for some function $m(W, A, X)$.

  To prove this result, we see that for any parametric submodel that satisfies (\ref{eq:trt-bridge-identification})
    and is indexed by $s$ such that the true distribution corresponds to $s=0$, we have
    \begin{align*}
        E_s\bigg\{(1-Y)\bigg[q(A, Z, X;\tau_s) - \dfrac{1}{f_s(A|W, X, Y=0, S=1)}\bigg]m(W, A, X)\bigg|S=1\bigg\}=0
    \end{align*}
    and thus 
    \begin{align*}
        \partial E_s\bigg\{(1-Y)\bigg[q(A, Z, X;\tau_s) - \dfrac{1}{f_s(A|W, X, Y=0, S=1)}\bigg]m(W, A, X)\bigg| S=1\bigg\}/\partial s\bigg |_{s=0}=0.
    \end{align*}
    
    Note that 
    \begin{align*}
        &\dfrac{\partial}{\partial s}E_s[(1-Y)q(A, Z, X;\tau_s)m(W, A, X)|S=1]\\
        =& E\bigg[\dfrac{\partial q(A, Z, X;\tau)}{\partial \tau}\bigg|_{\tau = \tau_0}m(W, A, X)(1-Y)\bigg| S=1\bigg]\dfrac{\partial\tau_s}{\partial s}\bigg|_{s=0}+\\
        &\qquad E[(1-Y)q(A, Z, X)m(W, A, X)S(O|S=1)|S=1]\\
    \end{align*}
    and \begin{align*}
        &\dfrac{\partial}{\partial s}E_s\bigg[\dfrac{(1-Y)m(W, A, X)}{f_s(A|W, X, Y=0, S=1)}\bigg|S=1\bigg] \\
        =& E\bigg[-(1-Y)m(W, A, X)\dfrac{\dfrac{\partial}{\partial s}f_s(A|W, X, Y=0, S=1)}{f^2(A|W, X, Y=0, S=1)}\bigg| S=1\bigg] +\\
        &\qquad E\bigg[\dfrac{m(W, A, X)(1-Y)}{f(A|W, X, Y=0, S=1)}S(W, A, X, Y=0|S=1)\bigg| S=1\bigg]\\
        =& E\bigg[-\dfrac{m(W, A, X)(1-Y)}{f(A|W, X, Y=0, S=1)}S(A|W, X, Y=0, S=1)\bigg|S=1\bigg] +\\
        &\qquad E\bigg[\dfrac{m(W, A, X)(1-Y)}{f(A|W, X, Y=0, S=1)}S(W, A, X, Y=0|S=1)\bigg| S=1\bigg]\\
        = & E\bigg[\dfrac{m(W, A, X)(1-Y)}{f(A|W, X, Y=0, S=1)}S(W, X, Y=0|S=1)\bigg| S=1\bigg]\\
        =& E\bigg[(1-Y)\{m(W, 1, X) - m(W, 0, X)\}S(W, X, Y=0| S=1)\bigg|S=1\bigg]\\
        =& E\bigg[(1-Y)\{m(W, 1, X) - m(W, 0, X)\}S(O| S=1)\bigg|S=1\bigg]
    \end{align*}
    Rearranging the terms, we have
    \begin{align*}
        &\dfrac{\partial \tau_s}{\partial s}\bigg|_{s=0}=\\
        & E\bigg[-\bigg\{E\bigg[\dfrac{\partial q(A, Z, X;\tau)}{\partial \tau}\bigg|_{\tau = \tau_0}m(W, A, X)(1-Y)\bigg| S=1\bigg]\bigg\}^{-1}(1-Y)\\&\qquad \bigg(m(W, A, X)q(A, Z, X) - m(W, 1, X) - m(W, 0, X)\bigg)S(O|S=1)\bigg|S=1\bigg]
    \end{align*}

\pagebreak

\section{Discussion on the rare disease assumption~\ref{assump:rare-disease}}\label{append:rare-disease}

We first describe a crucial identity that links the log risk ratio $\beta_0$ and the treatment confounding bridge function $q(A, Z, X)$.
\begin{lemma} 
Under Assumptions~\ref{assump:ident}, \ref{assump:trt-indep-samp}, \ref{assump:no-eff-mod} and \ref{assump:trt-bridge}, we have
\begin{equation}\label{eq:rr-estimand}\beta_0 =\log\left(\dfrac{E[q(A, Z, X)I(A=1, Y=1)| S=1]}{E[q(A, Z, X)I(A=0, Y=1)|S=1]}\right) \end{equation}
\end{lemma}

\begin{proof}
The right hand side of \eqref{eq:rr-estimand} is
\begin{align*}
    & \log\left(\dfrac{E[q(A, Z, X)I(A=1, Y=1)|S=1]}{E[q(A, Z, X)I(A=0, Y=1)|S=1]}\right)\\
    =&\log \left(\dfrac{E[q(A, Z, X)|A=1, Y=1, S=1]}{E[q(A, Z, X)|A=0, Y=1, S=1]}\right) + \log\left(\dfrac{P(Y=1, A=1|S=1)}{P(Y=1, A=0|S=1)}\right)
\end{align*}

Note that for $a=0,1$, 
\begin{align*}
    &E[q(A, Z, X)|A=a, Y=1, S=1] \\
    = &E\left\{E[q(A, Z, X)|U, X, A=a, Y=1, S=1]|A=a, Y=1, S=1\right\}\\
    \stackrel{A.\ref{assump:trt-indep-samp}}{=} & E\left\{E[q(A, Z, X)|U, X, A=a]|A=a, Y=1, S=1\right\}\\
    \stackrel{A.\ref{assump:trt-bridge}}{=} &E\left\{\dfrac{1}{P(A=a|U, X)}\bigg|A=a, Y=1, S=1\right\}\\
    =& \int \dfrac{1}{P(A=a|U=u, X=x)}f(u, x|A=a, Y=1, S=1)\,\mathrm du\,\mathrm dx\\
    =& \int \dfrac{1}{P(A=a|U=u, X=x)}\times\\&\dfrac{P(A=a|U=u, X=x)P(Y=1|A=a, U=u, X=x)P(S=1|Y=1, A=a, U=u, X=x)}{P(A=a, Y=1, S=1)}\,\mathrm du\,\mathrm dx\\
    \stackrel{A.\ref{assump:trt-indep-samp},A.\ref{assump:no-eff-mod}}{=}& \int \dfrac{\exp(\beta_0a)P(Y=1|A=0, U=u, X=x)P(S=1|Y=1, U=u, X=x)}{P(A=a, Y=1, S=1)}\,\mathrm du\,\mathrm dx\\
    =& \dfrac{\exp(\beta_0 a)}{P(A=a, Y=1, S=1)}\int P(Y=1|A=0, U=u, X=x)P(S=1| Y=1, U=u, X=x)\,\mathrm du\,\mathrm dx.
\end{align*}

Therefore, the right-hand side of \eqref{eq:rr-estimand} equals
\begin{align*}
    &\beta_0 + \log\left(\dfrac{P(A=0, Y=1, S=1)}{P(A=1, Y=1, S=1)}\right) + \log\left(\dfrac{P(Y=1, A=1|S=1)}{P(Y=1, A=0|S=1)}\right)\\
    =& \beta_0 + \log\left(\dfrac{P(A=0, Y=1| S=1)P(S=1)}{P(A=1, Y=1| S=1)P(S=1)}\right) + \log\left(\dfrac{P(Y=1, A=1|S=1)}{P(Y=1, A=0|S=1)}\right)\\
    =& \beta_0
\end{align*}

\end{proof}

Let $q^*$ be the function that satisfies~\eqref{eq:trt-bridge-identification}:
$$E[q^*(a, Z, X)|A=a, W, X, Y=0, S=1] = \dfrac{1}{P(A=a|U=u, X, Y=0, S=1)}.$$

We introduce an additional regularity condition:
\begin{assumption}[Uniform continuity]\label{assump:regularity} For any fixed positive square-integrable function $g(U)$ and a small positive number $0<\eta<1$, there exists some $0<\gamma = \gamma(g, \eta)<0$ such that $1-\gamma < \dfrac{E[g_1(U)|A, W, X, Y=0, S=1]}{E[g(U)|A, W, X, Y=0, S=1]}<\dfrac{1}{1-\gamma}$ a.e. implies $1-\eta<\dfrac{g_1(U)}{g(U)}<\dfrac{1}{1-\eta}$ a.e., where $g_1(U)$ is a positive square integrable function. 

\end{assumption}

Assumption~\ref{assump:regularity} requires that the inverse mapping of $g\rightarrow E[g(U)|A,W,X, Y=0, S=1]$ is sufficiently smooth. By Theorem~\ref{thm:trt-bridge-identification} and Assumption~\ref{assump:solvability}, we have
\begin{align*}
    (1 - \delta)^3 &< \dfrac{E[q^*(A, Z, X)|A, W, X, Y=0, S=1]}{E[q(A, Z, X)|A, W, X, Y=0, S=1]}\\
    &=  \dfrac{E\{E[q^*(A, Z, X)|A, W, U, X, Y=0, S=1]|A, W, X, Y=0, S=1\}}{E\{E[q(A, Z, X)|A, W, U, X, Y=0,S=1]|A, W, X, Y=0, S=1\}}\\
    &\stackrel{A.\ref{assump:nco}}{=} \dfrac{E\{E[q^*(A, Z, X)|A,U, X]|A, W, X, Y=0, S=1\}}{E\{E[q(A, Z, X)|A, U, X]|A, W, X, Y=0, S=1\}}\\
    &\stackrel{A.\ref{assump:trt-bridge}}{=}\dfrac{E\{E[q^*(A, Z, X)|A,U, X]|A, W, X, Y=0, S=1\}}{E\{E[\dfrac{1}{P(A|U, X)}|A, U, X]|A, W, X, Y=0, S=1\}}\\
    &<\dfrac{1}{(1-\delta)^3}
\end{align*}
By Assumption~\ref{assump:regularity}, this implies \begin{equation}\label{eq:q-approx-bound}\dfrac{1-\eta(\delta)}{P(A|U, X)}<E[q^*(A, Z, X)|A,U, X]<\dfrac{1}{(1-\eta(\delta))P(A|U, X)}.\end{equation}
The constant $\eta(\delta)$ is determined by the smoothness of the inverse mapping of $g\mapsto E[g(U)|A, W, X, Y=0, S=1]$.

\begin{equation}\label{eq:trt-bridge-est}E[q^*(A, Z, X)|U=u,  X, Y=0, S=1] = \dfrac{1}{P(A=a|U=u,X, Y=0, S=1)}\end{equation} for almost all $a$ and $u$. Let 
\begin{align*}\beta_0^* &=\log\left(\dfrac{E[q^*(A, Z, X)AY| S=1]}{E[q^*(A, Z, X)(1-A)Y|S=1]}\right) \\&=\log\left(\dfrac{E[q^*(A, Z, X)|A=1, Y=1, S=1]}{E[q^*(A,Z, X)|A=0, Y=1, S=1]}\right) + \log\left(\dfrac{P(Y=1, A=1|S=1)}{P(Y=1, A=0|S=1)}\right).\end{align*}

Under mild regularity, the estimator $\widehat\beta$ is regular and asymptotically linear for $\beta_0^*$, and therefore $$\widehat\beta= \beta_0^* + O_p(1/\sqrt n) = \beta_0 + (\beta_0^* - \beta_0) + O_p(1/\sqrt n).$$

It suffices to study 
$$\beta_0^* -\beta_0 = \log\left(\dfrac{E[q^*(A, Z, X)|A=1, Y=1, S=1]}{E[q^*(A, Z, X)|A=0, Y=1, S=1]}\right) -\log\left(\dfrac{E[q(A, Z, X)|A=1, Y=1, S=1]}{E[q(A, Z, X)|A=0, Y=1, S=1]}\right) $$

Notice that 
\begin{align*}
    &E[q^*(A, Z, X)|A=a, Y=1, S=1]\\ 
    =& E\{E[q^*(A, Z, X)|A=a, U, X, Y=1, S=1]|A=a, Y=1, S=1\}\\
    \stackrel{A.\ref{assump:nce}}{=}& E\{E[q^*(A, Z, X)|A=a, U, X]|A=a, Y=1, S=1\}\\%&&\mbox{(since $Z\indep (Y, S)|A, U$)}\\
    \stackrel{A.\ref{assump:nce}}{=} & E\{E[q^*(A, Z, X)|A=a, U, X, Y=0]|A=a, Y=1, S=1\}\\%&&\mbox{(since $Z\indep (Y, S)|A, U$)}\\
    \stackrel{\eqref{eq:q-approx-bound}}{<}& \dfrac{1}{1-\eta(\delta)}E\left\{\dfrac{1}{P(A=a|U, X)} \bigg|A=a, Y=1, S=1\right\}\\
    \stackrel{A.\ref{assump:trt-bridge}}{=}&\dfrac{1}{1-\eta(\delta)}E\left\{E[q(A, Z, X)|A, U, X] \bigg|A=a, Y=1, S=1\right\}\\
    \stackrel{A.\ref{assump:nce}}{=}&\dfrac{1}{1-\eta(\delta)}E\left\{E[q(A, Z, X)|A, U, X, Y=1, S=1] \bigg|A=a, Y=1, S=1\right\}\\
    =& \dfrac{1}{1-\eta(\delta)}E[q(A, Z, X)|A=a, Y=1, S=1]
\end{align*}
and similarly, 

$$E[q^*(A, Z, X)|A=a, Y=1, S=1]>(1-\eta(\delta))E[q(A, Z, X)|A=a, Y=1, S=1].$$

Therefore, we have

$$1-\eta(\delta) < \dfrac{E[q^*(A, Z, X)|A=a, Y=1, S=1]}{E[q(A, Z, X)|A=a, Y=1, S=1]} < \dfrac{1}{1-\eta(\delta)}.$$

We conclude that
$$2\log(1-\eta(\delta))=\log((1-\epsilon)^2) < \beta_0^* - \beta_0 < \log\left(\dfrac{1}{(1-\epsilon)^2}\right)=-2\log(1-\epsilon)$$

and thus

$$|\widehat\beta - \beta_0| < -2\log(1-\eta(\delta)) + O_p(\dfrac{1}{\sqrt n}).$$

\pagebreak
\section{Regularity conditions and proof of Theorem~\ref{thm:ee-joint}}\label{append:ee-joint}

We denote $\tau_0^*$ as the true value of $\tau$ such that $q(A, Z, X;\tau_0^*)=q^*(A, Z, X)$. We will give the regularity conditions and proof that $(\widehat\beta,\widehat\tau)$ is a regular and asymptotically linear estimator of $(\beta_0^*, \tau_0^*)$. Here $\beta_0^*$ and $q^*$ are the biased versions of $\beta_0$ and $q$ defined in Appendix~\ref{append:rare-disease}, respectively, although the biases are negligible when the infection is rare. Following Appendix~\ref{append:rare-disease}, $\widehat\beta$ is also a regular and asymptotically linear estimator of $\beta_0$ if 
$$\sup_{a, w, u, x}P(Y=1|A=a, W=w, U=u, X=x)< \delta_n,$$
Assumption~\ref{assump:regularity} holds and $$\log(1-\eta(\delta_n))=o_p(\dfrac{1}{\sqrt n}).$$

A set of regularity conditions are
\begin{enumerate}[R.1]
    \item The function $\tau\mapsto q(A, Z, X;\tau)$ is Lipschitz in a neighborhood of $\tau_0^*$; that is, for every $\tau_1$ and $\tau_2$ in a neighborhood of $\tau_0$ and a measurable function $\dot q(A, Z, X)$ with $E[\dot q(A, Z, X)]<\infty$, we have
    $\lVert q(A, Z, X;\tau_1) - q(A, Z, X;\tau_2)\rVert\leq \dot q(A, Z, X)\lVert \tau_1 - \tau_2\rVert$;
    \item $E[q(A, Z, X;\tau_0^*)^2]<\infty$ and $E[m(W, A, X)^2]<\infty$;
    \item The function $\tau\mapsto q(A, Z, X;\tau)$ is differentiable at $\tau_0^*$. The derivative matrix $\Omega(\beta_0^*, \tau^*)$ is nonsingular;
    \item $\dfrac{1}{n}\sum_{i=1}^n G_i(\widehat\beta,\widehat\tau)=o_p(n^{-1/2})$ and $(\widehat\beta,\widehat\tau)\overset{p}{\rightarrow}(\beta_0^*,\tau_0^*)$.
\end{enumerate}

The condition R.1 and the fact that $\beta\mapsto\exp(-\beta A)$ is Lipschitz in a neighborhood of $\beta_0^*$ imply the function $(\beta,\tau)\rightarrow M_i(\beta,\tau)$ is Lipschitz in a neighborhood of $(\beta_0^*,\tau_0^*)$ for every $i$. The remaining proof follows \citet{van2000asymptotic} Theorem 5.21.

\pagebreak
\section{Estimating conditional causal RR in the presence of effect modification by measured confounders}\label{append:tnd-alg-2}

Algorithm~\ref{alg:nc-tnd-2} below describes a straightforward extension of Algorithm~\ref{alg:nc-tnd} to estimation the conditional vaccine effectiveness $VE(x)=1-\exp(\beta_0(x))$ under Assumption~\ref{assump:no-eff-mod-by-u} and a parametric model $\beta_0(X;\alpha)$ indexed by a finite-dimensional parameter $\alpha$.

\begin{algorithm}
\caption{\label{alg:nc-tnd-2}Negative control method to estimate conditional vaccine effectiveness from a test-negative design }
\begin{algorithmic}[1]
\State Identify the variables in the data according to Figure~\ref{fig:dag-tnd}(c), in particular the NCEs and NCOs.
\State Estimate the treatment confounding bridge function by solving the equation~(\ref{eq:ee-trt-bridge}) with a suitable parametric model $q^*(A, Z, X;\tau)$ and a user-specified function $m(W, A, X)$. Write $\widehat\tau$ as the resulting estimate of $\tau$.
\State  Estimate $\alpha$ by solving 
\begin{equation}
    \dfrac{1}{n}\sum_{i=1}^n(-1)^{1-A_i}C(X_i)q^*(A_i, Z_i, X_i;\widehat\tau)\exp(-\beta_0(X_i;\alpha)A_i)=0
\end{equation}
Denote the resulting estimate of $\alpha$ as $\widehat\alpha$. The estimated conditional vaccine effectiveness is 
    $$\widehat{VE}(x) = 1 - \exp(\beta_0(x;\widehat\alpha)).$$
\end{algorithmic}
\end{algorithm}

We describe the large-sample properties of the estimator $(\widehat\alpha, \widehat\tau)$ in the theorem below.

\begin{theorem}[Inference of $(\widehat\alpha,\widehat\tau)$]\label{thm:ee-joint-2}
Under Assumptions \ref{assump:ident}- \ref{assump:trt-indep-samp}, \ref{assump:nce}-\ref{assump:trt-bridge},  \ref{assump:nco}, \ref{assump:no-eff-mod-by-u} and suitable regularity conditions listed at the end of this section, the estimator $(\widehat\alpha, \widehat\tau)$ in Algorithm~\ref{alg:nc-tnd-2}, or equivalently, the solution to the estimating equation $\dfrac{1}{n}\sum_{i=1}^n \tilde G_i(\alpha, \tau)=0$
is regular and asymptotically linear with influence function
$$\widetilde{IF}(\alpha, \tau) = -\left[\tilde\Omega(\alpha, \tau)^T\tilde\Omega(\alpha, \tau)\right]^{-1}\tilde\Omega(\alpha, \tau)^T\tilde G_i(\alpha, \tau),$$

where
$$\tilde G_i(\alpha, \tau) = \left(\begin{array}{c}  {(-1)}^{1-A_i}C(X_i)q^*(A_i, Z_i, X_i; \tau)Y_i\exp(-\beta_0(X;\alpha) A_i)\\
(1-Y_i)\left[m(W_i, A_i, X_i)q^*(A_i, Z_i, X_i;\tau) - m(W_i, 1, X_i) - m(W_i, 0, X_i)\right]\end{array}\right)$$
and
\begin{align*}\tilde\Omega(\alpha, \tau)&=\left(E\left[\dfrac{\partial \tilde G_i(\alpha, \tau)}{\partial \alpha^T}\right], E\left[\dfrac{\partial \tilde G_i(\alpha, \tau)}{\partial \tau^T}\right]\right).
\end{align*}
Here $C(X)$ is a user-specified function of $X$ with the sample dimension as $\alpha$. 

\end{theorem}

Suppose that in Algorithm~\ref{alg:nc-tnd-2}, one specifies  $\beta_0(X;\alpha)=X^T\alpha$, then a natural choice for $C(X)$ is $C(X)=X$. A sandwich estimator of the asymptotic variance of $(\widehat\alpha,\widehat\tau)$ can be deduced from previous derivations. 
Under Assumption~\ref{assump:no-eff-mod-by-u}, we have shown that one can identify $VE(X)$, however, one may be unable to identify the population marginal $VE$ without an additional assumption. % because the distribution of $X$ is unknown in the target population. 
 Interestingly, we note that the population marginal risk ratio would remain non-identified even if one had access to a random sample from the target population to inform the marginal distribution of $X$. Specifically, as shown in~\citet{huitfeldt2019collapsibility}, the marginal $RR=E[Y(1)]/E[Y(0)]$ can be written as $RR=E[RR(X)|Y(0)=1]$, i.e. the average risk ratio among subjects who would contract say Influenza had they possibly contrary to fact, not been vaccinated against Influenza. However, the distribution of $X$ within the group $Y(a=0)=1$  cannot be identified in presence of unmeasured confounding, thus ruling out identification of the population marginal $RR$.

Define $q^*$ as before and $$\beta_0^*(x)=\log\left(\dfrac{E[q^*(A, Z, x)I(A=1, Y=1)| S=1, X=x]}{E[q^*(A, Z, x)I(A=0, Y=1)|S=1,X=x]}\right).$$
Denote $\alpha_0^*$ as the value of $\alpha$ such that $\beta_0(x;\alpha_0^*)=\beta_0^*(x)$. Below we give the set of regularity conditions such that $(\widehat\alpha,\widehat\tau)$ is a regular and asymptotically linear estimator of $(\alpha_0^*,\tau_0^*)$:

\begin{enumerate}[R'.1]
    \item The function $\tau\mapsto q(A, Z, X;\tau)$ is Lipschitz in a neighborhood of $\tau_0^*$ and $\alpha\mapsto \beta_0(X;\alpha)$ is Lipschitz in a neighborhood of $\alpha_0^*$. 
    \item $E[q(A, Z, X;\tau_0^*)^2]<\infty$,  $E[m(W, A, X)^2]<\infty$, $E[C(X)^2]<\infty$ and $E[\exp(-2\beta_0(X;\alpha_0^*))]<\infty$. ;
    \item The function $\tau\mapsto q(A, Z, X;\tau)$ is differentiable at $\tau_0^*$ and $\alpha\mapsto\beta_0(X;\alpha)$ is differentiable at $\alpha_0^*$. The derivative matrix $\tilde\Omega(\alpha_0^*, \tau^*)$ is nonsingular;
    \item $\dfrac{1}{n}\sum_{i=1}^n \tilde G_i(\widehat\alpha,\widehat\tau)=o_p(n^{-1/2})$ and $(\widehat\alpha,\widehat\tau)\overset{p}{\rightarrow}(\alpha_0^*,\tau_0^*)$.
\end{enumerate}

\pagebreak
\section{Proof of Theorem 1$'$}
\label{append:causal-or-ee}

\begin{lemma}\label{lemma:or}
Under Assumptions 2$'$, 3$'$ and 6$'$, we have
\begin{align*}&P(A=a, Y=y|U,W, X, S=1)\\&\qquad=\dfrac{1}{c}P(A=a|Y=0, U,W, X, S=1)P(Y=y|A=0,U,W, X, S=1)\exp(\beta_0'ay)\end{align*}
where $c=\sum_{a^*,y^*}P(A=a^*|Y=0, U,W, X, S=1)P(Y=y^*|A=0,U,W, X, S=1)\exp(\beta_0'a^*y^*)$.
\end{lemma}

\begin{proof}
Note that
\begin{align*}
    &\dfrac{P(Y=1|A=1, U,W, X, S=1)P(Y=0|A=0, U,W, X, S=1)}{P(Y=1|A=0, U,W, X, S=1)P(Y=0|A=1, U,W, X, S=1)}\\
    \stackrel{A.6'}{=} &\dfrac{P(Y=1|A=1, U, X, S=1)P(Y=0|A=0, U, X, S=1)}{P(Y=1|A=0, U, X, S=1)P(Y=0|A=1, U, X, S=1)}\\
    =& \exp(\beta_0')\times \dfrac{P(S=1|Y=1, A=1, U, X)}{P(S=1|Y=0, A=1, U, X)}\times \dfrac{P(S=1|Y=0, S=0, U, X)}{P(S=1|Y=1, S=0, U, X)}\\
    \stackrel{A.2'}{=}& \exp(\beta_0')\times \exp(h(U, X))\times \exp(-h(U, X))\\
    =& \exp(\beta_0')
\end{align*}

The result follows after \citet{chen2003note}.
\end{proof}

To prove Theorem 1$'$, we need to show $$E[(-1)^{1-A}q(A,Z,X)Y\exp(-\beta_aA)|S=1]=0$$

It suffices to prove that 
$$E[(-1)^{1-A}q(A, Z, X)Y\exp(-\beta_aA)|U, W, X, S=1]=0.$$

The left-hand side is 
\begin{align*}
    &\sum_{a, y}\int (-1)^{1-a}y\exp(-\beta_0' a)q(a, z, X)P(A=a, Y=y|U,X, S=1)f(z|U, X, A=a, Y=y, S=1)dz\\
    \stackrel{L.\ref{lemma:or}}{=}&\sum_{a, y}\int (-1)^{1-a}y\exp(-\beta_0' a)q(a, z, X)\times \\&\quad \dfrac{1}{c}P(A=a|Y=0, U, X, S=1)P(Y=y|A=0, U,X, S=1)\exp(\beta_0' ay)f(z|U, X, A=a, Y=y, S=1)dz\\
    =& \sum_{a}\int (-1)^{1-a}\exp(-\beta_0' a)q(a, z, X)\times \\&\quad \dfrac{1}{c}P(A=a|Y=0, U, X, S=1)P(Y=1|A=0, U, X, S=1)\exp(\beta_0' a)f(z|U, X, A=a, Y=1, S=1)dz\\
     =& \sum_{a} \dfrac{(-1)^{1-a}}{c}P(A=a|Y=0, U, X, S=1)P(Y=1|A=0, U, X, S=1)\times\\
     &\quad\int q(a, z, X)f(z|U,X, A=a, Y=1, S=1)dz\\
     \stackrel{A.\ref{assump:trt-bridge}'}{=}&\sum_{a} \dfrac{(-1)^{1-a}}{c}P(A=a|Y=0, U, X, S=1)P(Y=1|A=0, U, X, S=1)\times\\
     &\quad \dfrac{1}{P(A=a|Y=0, U, X, S=1)}\\
     =& \sum_{a} \dfrac{(-1)^{1-a}}{c}P(Y=1|A=0, U, X, S=1)\\
     =& 0
\end{align*}

\pagebreak

\section{Proof of Theorem~\ref{thm:trt-bridge-identification}$'$}\label{append:tilde-q-identification}

\begin{align*}
    & E[\tilde q(a, Z, X)|A=a, W, X, Y=0, S=1]\\
    =& E\left\{E[\tilde q(a, Z, X)|A=a, U, W, X, Y=0, S=1]|A=a, W, X, Y=0, S=1\right\}\\
    \stackrel{A.\ref{assump:nco}'}{=}& E\left\{E[\tilde q(a, Z, X)|A=a, U, X, Y=0, S=1]|A=a, W, X, Y=0, S=1\right\}\\
    \stackrel{A.\ref{assump:nce}'}{=}& E\left\{E[\tilde q(a, Z, X)|A=a, U, X]|A=a, W, X, Y=0, S=1\right\}\\
    \stackrel{A.\ref{assump:trt-bridge}'}{=}&E\left\{\dfrac{1}{P(A=a| U, X, Y=0, S=1)}|A=a, W, X, Y=0, S=1\right\}\\
    =& \int \dfrac{1}{P(A=a| U=u, X, Y=0, S=1)}f(u|A=a, W, X, Y=0, S=1)du \\
    =& \int \dfrac{1}{P(A=a| U=u, X, Y=0, S=1)}\dfrac{f(u| W, X, Y=0, S=1)P(A=a|U=u, W, X, Y=0, S=1)}{P(A=a|W, X, Y=0, S=1)}du\\
    \stackrel{A.\ref{assump:nco}'}{=}& \int \dfrac{1}{P(A=a| U=u, X, Y=0, S=1)}\dfrac{f(u| W, X, Y=0, S=1)P(A=a|U=u, X, Y=0, S=1)}{P(A=a|W, X, Y=0, S=1)}du\\
    =& \dfrac{1}{P(A=a|W, X, Y=0, S=1)}\int f(u| W, X, Y=0, S=1)du\\
    =& \dfrac{1}{P(A=a|W, X, Y=0, S=1)}
\end{align*}

\pagebreak

\section{Simulation setting with binary unmeasured confounder}\label{append:sim-binary}
We generate the data of a size $N=7,000,000$ general population according to Figure~\ref{fig:dag-tnd}, with the following distribution:

\begin{align*}
    U & \sim Bernoulli(p_U);\\
    Z|U & \sim Bernoulli(p_{0Z} + p_{UZ}U);\\
    A|U & \sim Bernoulli(p_{0A} + p_{UA}U);\\
   Y|A, Y & \sim Bernoulli(\exp(\eta_{0Y} + \beta_0A + \eta_{UY}U));\\
    W|U &\sim Bernoulli(p_{0W}+ p_{UW}U);\\
    D|U & \sim Bernoulli(p_{0D} + p_{UD}U);\\
    S|Y, D, W, U &\sim Bernoulli(\max(Y, D, W)(p_{YS} + p_{UYS}U)).
\end{align*}

In the above data generating process, to mimic a test-negative design platform, we created a binary indicator $D$ for flu-like diseases other than $W$. The study sample contains subjects with $S=1$. The distribution of $S$ indicates that only subjects with at least one of $Y$, $D$ and $W$ equal to one will be recruited into the study sample. We chose the parameters as in Table~\ref{tab:param_binary}, which resulted in an average study sample size of between around 48,000 and around 52,000.

\begin{table}[!htbp]
    \centering
    \caption{Parameter values of the data generating distribution in the simulation with a binary unmeasured confounder.}
    \begin{tabular}{cc||cc}
    \hline
    Parameter     & Value & Parameter & Value  \\
    \hline
    $p_U$ &  0.5   & $p_{0Z}$ & 0.2\\
    $p_{UZ}$ & 0.4 & $p_{0A}$ & 0.2\\
    $p_{UA}$ & 0.4 & $\eta_{0Y}$ & $\log(0.01)$\\
    $\beta_0$ & $\log(0.2)$, $\log(0.5)$, $\log(0.7)$, or $0$ & $\eta_{UY}$ &$\log(0.5)$\\
    $p_{0W}$ & 0.02 & $p_{UW}$ & 0.02\\
    $p_{0D}$ & 0.02 & $p_{UD}$ & -0.015\\
    $p_{YS}$ & 0.1 & $p_{UYS}$ & 0.4\\
    \hline
    \end{tabular}
    
    \label{tab:param_binary}
\end{table}

To obtain the true treatment confounding bridge function, by \eqref{eq:trt-bridge-defn}, we have
$$\sum_z q(a, z)f(z|u, a) = \sum_z q(a, z)f(a|u, z)f(z|u)/f(a|u) = 1/f(a|u)$$ and thus
\begin{align*}
&\sum_z q(a, z)f(a|u, z)f(z|u)\\
=& q(a, 0)[p_{0A} + p_{UA}u]^a[1-(p_{0A} + p_{UA}u)]^{1-a}[1 - (p_{0Z}+ p_{UZ}u)]+\\
& q(a, 1)[p_{0A} + p_{UA}u]^a[1-(p_{0A} + p_{UA}u)]^{1-a}(p_{0Z}+p_{UZ}u)\\=&1\end{align*}
for each $u, a$. We obtain that
\begin{align*}
    \begin{pmatrix}
    q(0, 0)\\q(0, 1)
    \end{pmatrix} &= \begin{pmatrix}
    [1 - p_{0A}](1 - p_{0Z}) &[1-p_{0A}]p_{0Z}\\
    [1 - (p_{0A} + p_{UA})][1 - (p_{0Z} + p_{UZ})] & [1-(p_{0A} + p_{UA})](p_{0Z} + p_{UZ})
    \end{pmatrix}^{-1}\begin{pmatrix}
    1\\1
    \end{pmatrix}\\
    \begin{pmatrix}
    q(1, 0)\\q(1, 1)
    \end{pmatrix} &= \begin{pmatrix}
    p_{0A}(1 - p_{0Z}) & p_{0A}p_{0Z}\\
    (p_{0A} + p_{UA})[1 - (p_{0Z} + p_{UZ})] & (p_{0A} + p_{UA})(p_{0Z} + p_{UZ})
    \end{pmatrix}^{-1}\begin{pmatrix}
    1\\1
    \end{pmatrix}
\end{align*}

\pagebreak
\section{Simulation setting with a continuous unmeasured confounder}\label{append:sim-continuous}

In simulation studies with a continuous unmeasured confounder, we generate data for a general population of 7,000,000 individuals following the distribution below:
\begin{align*}
    U & \sim Uniform(0, 1);\\
    X & \sim Uniform(0, 1); \\
    A|U, X & \sim Binomial(\expit(\mu_{0A} + \mu_{UA}U + \mu_{XA}X))\\
    Z|A, U, X & \sim N(\mu_{0Z} + \mu_{AZ}A + \mu_{XZ}X + \mu_{UZ}U, \sigma_Z^2);\\
    W|U, X & \sim N(\mu_{0W} + \mu_{XW}X + \mu_{UW}U, \sigma_W^2);\\
    Y|A, U, X & \sim Binomial(\expit(\mu_{0Y} + \beta A + \mu_{UY}U + \mu_{XY}X) +\mu_{UXY}UX)\\
    D|U, X & \sim Binomial(\expit(\mu_{0D} + \mu_{XD}X + \mu_{UD}U))\\
    S|U, X, Y, D & \sim \max(Y, D)\times Binomial(\expit(\mu_{0S}+\mu_{XS}X + \mu_{US}U + \mu_{UXS}UX)).
\end{align*}

We chose the parameter values according to Table~\ref{tab:param_continuous}, which results in an average study sample size of between around 43,000 and around 47,000.

\begin{table}[!htbp]
    \centering
    \caption{Parameter values of the data generating distribution in the simulation with a continuous unmeasured confounder.}
    \begin{tabular}{cc||cc}
    \hline
    Parameter     & Value & Parameter & Value  \\
    \hline
    $\mu_{0A}$ &  -1   & $\mu_{UA}$ & -1\\
    $\mu_{XA}$ & 0.25 & $\mu_{0Z}$ & 0\\
    $\mu_{AZ}$ & 0.25 & $\mu_{XZ}$ & $0.25$\\
    $\mu_{UZ}$ & 4 & $\sigma_Z$ & 0.25\\
    $\mu_{0Y}$ & $\log(0.01)$ &
    $\beta_0$ & $\log(0.2)$, $\log(0.5)$, $\log(0.7)$, or $0$\\  
    $\mu_{UY}$ & -2 &
    $\mu_{XY}$ & -0.25 \\
    $\mu_{0W}$ & 0 & $\mu_{XW}$ & 0.25\\
    $\mu_{UW}$ & 2 & $\sigma_W$ & 0.25 \\
    $\mu_{0D}$ & $\log(0.01)$ & $\mu_{XD}$ & 0.25 \\
    $\mu_{UD}$ & -0.2 & $\mu_{0S}$ & -1.4 \\
    $\mu_{XS}$ & 0.5 & $\mu_{US}$ & 2 \\
    $\mu_{UXS}$ & 1 & & \\
    \hline
    \end{tabular}
    
    \label{tab:param_continuous}
\end{table}

\pagebreak
\section{Logistic regression as a na\"ive approximate estimator of log risk ratio, ignoring unmeasured confounders}
\label{append:logit-reg}
Figure \ref{fig:dag-no-u} shows a causal diagram of a test-negative design with no unmeasured confounders. Again we assumed selection into the study $S$ is independent of the subjects' treatment status $A$, given the subjects' infection status $Y$ and other covariates $X$.

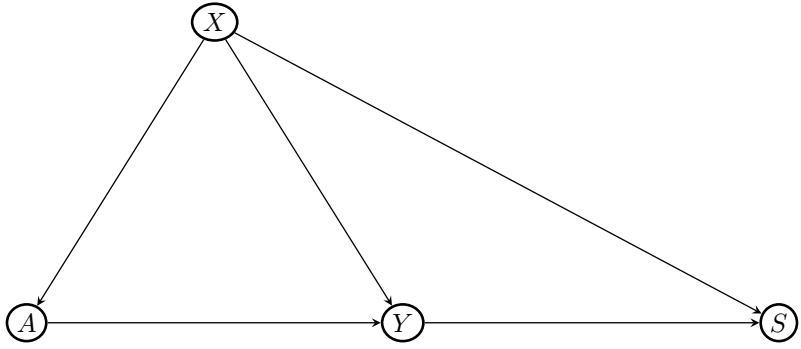
\begin{figure}[!htbp]
		\centering
		\begin{tikzpicture}
		
		\tikzset{line width=1pt,inner sep=5pt,
			swig vsplit={gap=3pt, inner line width right=0.4pt},
			ell/.style={draw, inner sep=1.5pt,line width=1pt}}

\node[shape = ellipse, ell] (A) at (-5, 0) {$A$};

\node[shape=ellipse,ell] (Y) at (0,0) {$Y$};

\node[shape=ellipse,ell] (S) at (5,0) {$S$};

\node[shape=ellipse,ell] (X) at (-2.5,4) {$X$};

\foreach \from/\to in {X/A, X/Y, X/S, A/Y, Y/S}
\draw[-stealth, line width = 0.5pt] (\from) -- (\to);

\end{tikzpicture}
		\caption{\label{fig:dag-no-u} Directed acylic graph of a test-negative design with no unmeasured confounders. }
	\end{figure}

In this scenario, we note that the conditional odds ratio given $X$ in the study population equals the conditional odds ratio in the general population, i.e.
\begin{align*}
    &\dfrac{P(Y=1|A=1, X, S=1)}{P(Y=1|A=0, X, S=1)}\bigg /\dfrac{P(Y=0| A=1, X, S=1)}{P(Y=0|A=0, X, S=1)}\\
    =& \dfrac{P(Y=1|A=1, X)P(S=1|Y=1, A=1, X)}{P(Y=1|A=0, X)P(S=1|Y=1, A=0, X)}\bigg /\dfrac{P(Y=0| A=1, X)P(S=1|Y=0, A=1, X)}{P(Y=0|A=0, X)P(S=1|Y=0, A=0, X)}\\
    \stackrel{A.\ref{assump:trt-indep-samp}}{=}& \dfrac{P(Y=1|A=1, X)P(S=1|Y=1, X)}{P(Y=1|A=0, X)P(S=1|Y=1, X)}\bigg /\dfrac{P(Y=0| A=1, X)P(S=1|Y=0,X)}{P(Y=0|A=0, X)P(S=1|Y=0, X)}\\
    =& \dfrac{P(Y=1|A=1, X)}{P(Y=1|A=0, X)}\bigg /\dfrac{P(Y=0| A=1, X)}{P(Y=0|A=0, X)}
\end{align*}

This implies that we may estimate the stratum-specific odds ratio in the general population by fitting a logistic regression model to the study data:
$$P(Y=1|A, X, S=1) = \expit(\gamma_0 + \gamma_XX + \gamma_A A)$$
where $\expit(x) = \exp(x)/[1 + \exp(x)]$, and $\gamma_A$ is the log odds ratio. When the outcome is rare in the general population, i.e. $P(Y=0|A=a, X=x)\approx 1$ for any $a, x$, the log odds ratio $\gamma_A$ also approximates the log risk ratio.

\pagebreak
\section{Simulation for non-rare diseases}\label{append:non-rare}

To investigate the performance of our method for non-rare diseases, we repeat the simulation with the same setup with binary or continuous confounders. For the simulation with binary confounders, we set $\eta_{0Y}$ to be $\log(0.20)$; for the simulation with continuous confounders, we set $\mu_{0Y}$ to be $\log(0.20)$ -- both correspond to an infection risk of 20\% for subjects with $A=U=X=0$.

The results are in Figure~\ref{fig:sim-nonrare}. While the NC-Oracle estimator remains unbiased and maintains calibrated confidence intervals, the NC estimator is in general biased with under-covered confidence intervals. Notably, the NC estimator is unbiased with calibrated confidence intervals under the null hypothesis where $\beta_0=0$. The logistic regression estimator is biased and has unpredictable behavior.

\begin{figure}[h]
    \centering
    \begin{tabular}{cc}
    \includegraphics[width = 0.45\textwidth]{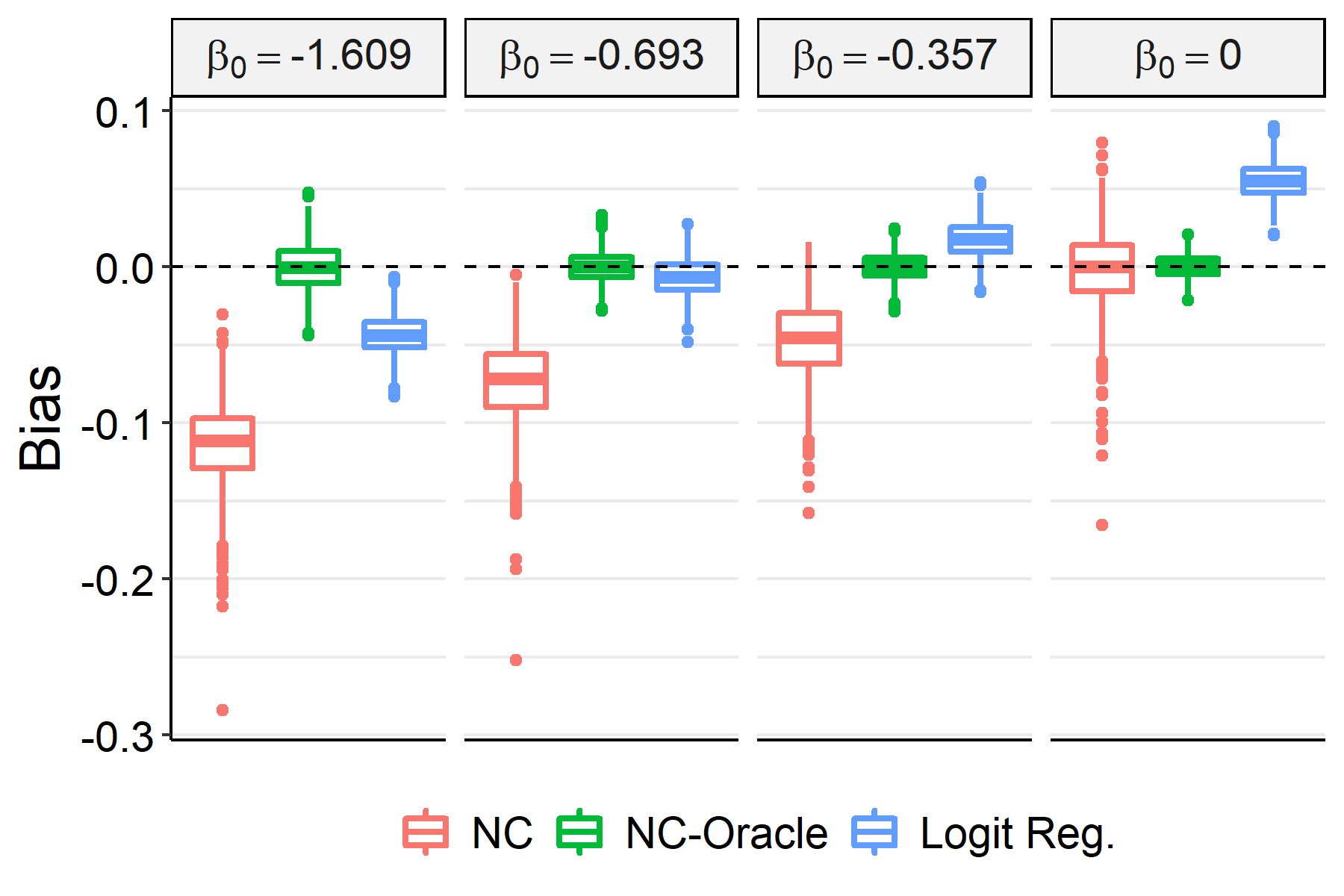}&
    \includegraphics[width = 0.45\textwidth]{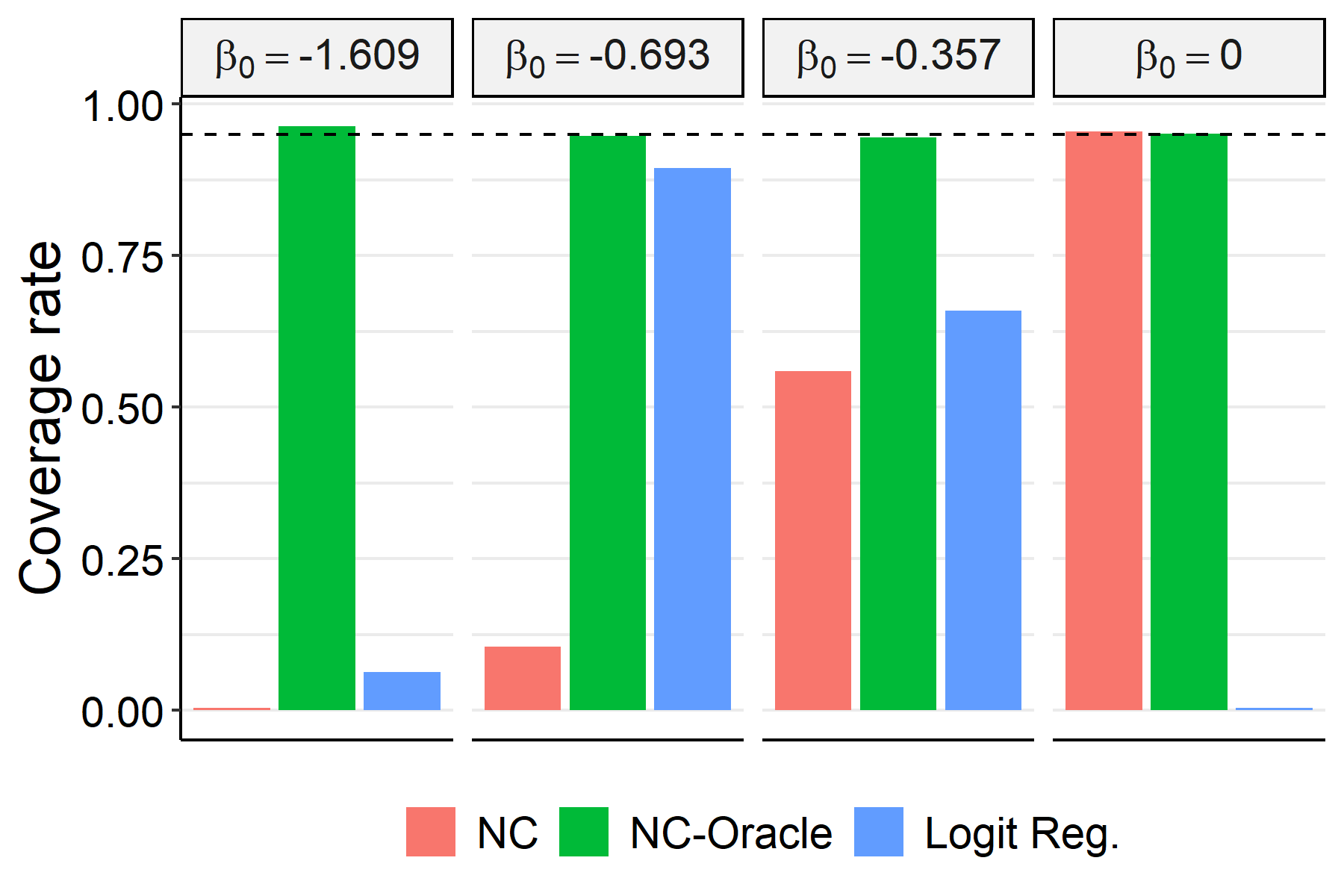}\\
    \multicolumn{2}{c}{(a)}\\
    \includegraphics[width = 0.45\textwidth]{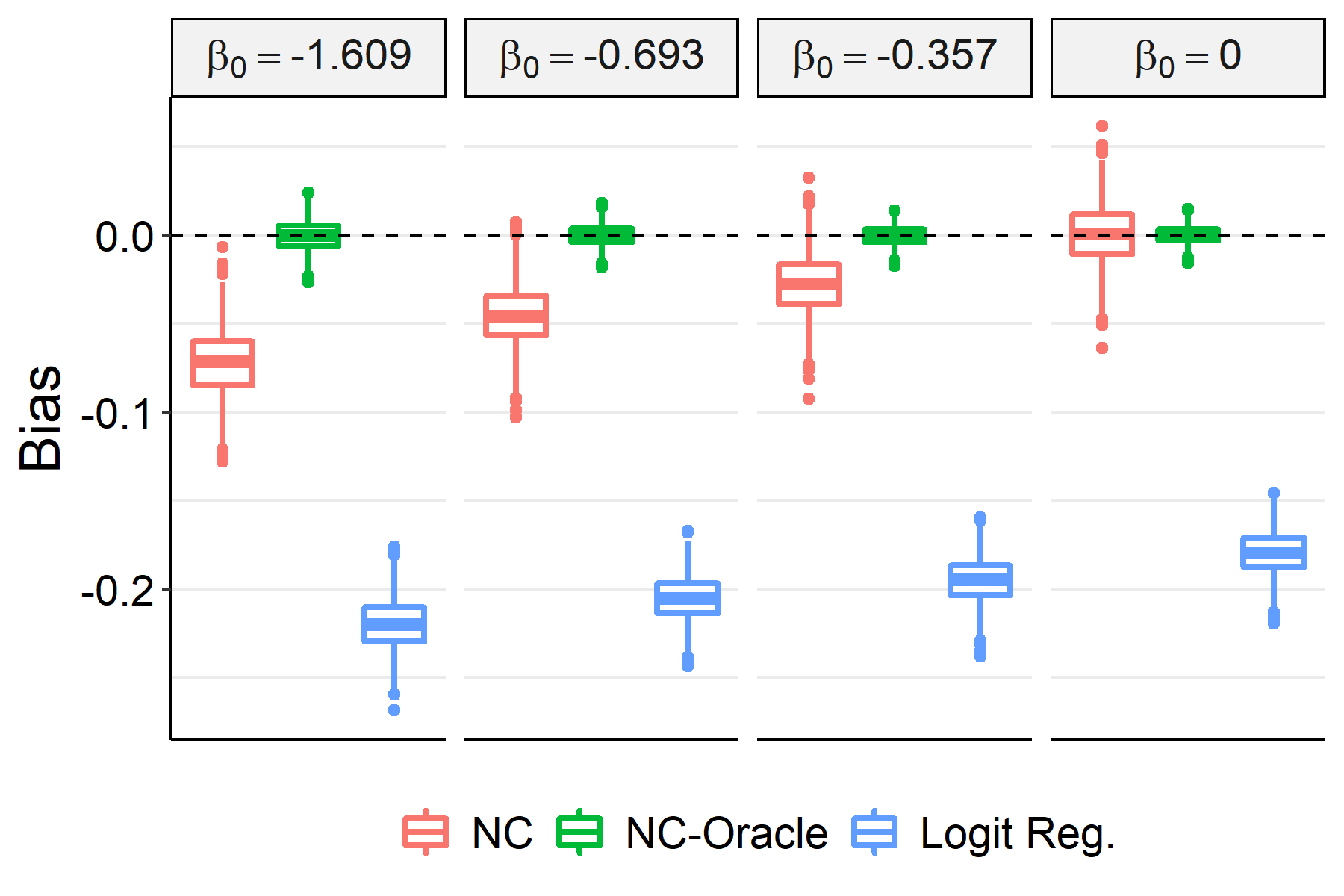}&
    \includegraphics[width = 0.45\textwidth]{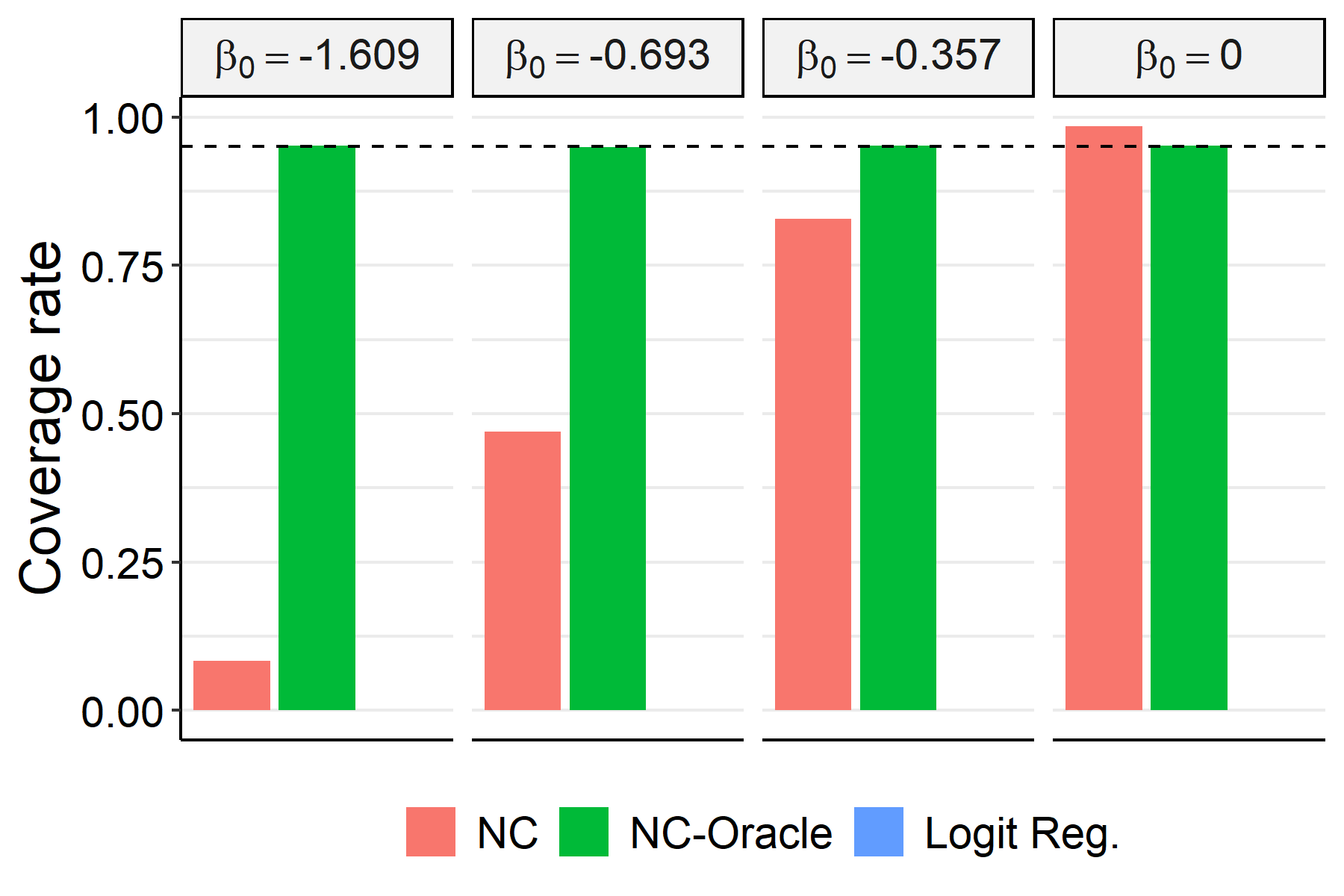}\\
    \multicolumn{2}{c}{(b)}\\
    \end{tabular}
    \caption{Bias (left) and coverage rates of 95\% confidence interval (right) for the oracle estimator (NC-Oracle), GMM estimator (NC-GMM) and logistic regression (Logit Reg.) with a binary or unmeasured confounder, where $\eta_{0y}$ or $\mu_{0Y}$ is $\log(0.20)$.}
    \label{fig:sim-nonrare}
\end{figure}

\pagebreak
\section{Detailed results of University of Michigan Health System Data analysis}\label{append:umhs-tables}

\begin{table}[!htbp]
    \centering
    \caption{Descriptive statistics of University of Michigan Health System COVID-19 Data. Variables were summarized as counts  (percentage\%).}
    \begin{tabular}{l|rr}
    \hline
    & Unvaccinated (N=777) & Vaccinated (N=39,377)\\
    \hline
    Vaccine types \\
    $\quad$Pfizer-BioNTech & / & 20,239 (51.4\%) \\
    $\quad$Moderna & / & 10,719 (27.2\%) \\
    $\quad$Johnson \& Johnson's Janssen & / & 1,405 (3.6\%)\\
    $\quad$Other & / & 7,014 (17.8\%) \\
    COVID-19 Infection & 456 (58.7\%) & 2,752 (7.0\%)\\
    NCE: Immunization before Dec 2020 & 433 (55.7\%) & 18,049 (45.8\%)\\
    NCO conditions \\
    $\quad$Arm/leg cellulitis & 7 (0.9\%) & 158 (0.4\%)\\
     $\quad$Eye/ear disorder & 6 (0.8\%) & 514 (1.3\%)\\
%     $\quad$Influenza & 2 (0.2\%) & 74 (0.2\%)\\
     $\quad$Gastro-esophageal disease & 51 (6.6\%) & 3,164 (8.0\%)\\
     $\quad$Atopic dermatitis & 1 (0.1\%) & 41 (0.1\%)\\
     $\quad$Injuries & 62 (8.0\%) & 3,652 (9.3\%)\\
     $\quad$General adult examination & 72 (9.3\%)& 4,657 (11.8\%)\\
     No. of NCO conditions $\geq$ 1 & 160 (20.6\%) & 10,280 (26.1\%)\\
     Age &&\\
     $\quad\leq 18$ & 93 (12.0\%) & 2,602 (6.6\%) \\ 
     $\quad\geq 18, <60$ & 543 (69.9\%) & 23,286 (59.1\%) \\
     $\quad\geq$ 60& 141 (18.1\%\%) & 13,489 (34.2\%)\\
     Male & 332 (42.7\%) & 16,151 (41.0\%)\\
     White & 557 (71.7\%) & 30,562 (77.6\%)\\
     Charlson score $\geq$ 3 & 50 (6.5\%)& 3,175 (8.1\%)\\
    \hline
    \end{tabular}
    \label{tab:data_app}
\end{table}

\pagebreak

\begin{table}[!htbp]
    \centering
    \caption{Logistic regression of COVID-19 infection on COVID-19 vaccination, the NCE (previous immunization) and other baseline covariates.}
    \begin{tabular}{l|rrr}
    \hline
         & Est. & S.E. & p-value \\
         \hline
    (Intercept)   &  0.26 &0.12 & $<$0.036 \\
    COVID-19 vaccination   &  -3.17 & 0.08 & $<$0.001\\
    Previous Immunization & 0.16 &0.04& $<$0.001\\
    Age $\geq$18, $\le$60 &0.58& 0.09& $<$0.001\\ 
    Age $\geq$60 &0.33&0.09&$<$0.001\\
    Male &0.06&0.04&0.149\\
    White & 0.23 & 0.05 & $<$0.001\\
    Charlson score $\geq$3 & 0.21& 0.07& 0.003\\
    Calendar month &&&\\
    $\quad$April& -0.48 & 0.08& $<$0.001 \\
    $\quad$May & -1.40 & 0.09 & $<$0.001 \\
    $\quad$June & -2.23 & 0.14& $<$0.001 \\
    $\quad$July & -1.23 & 0.10 & $<$0.001\\
    $\quad$August & -0.51 & 0.07 & $<$0.001\\
    $\quad$September & -0.30 & 0.06 & $<$0.001 \\
    $\quad$October & -0.06 & 0.05 & 0.240\\
    \hline
    \end{tabular}
    \label{tab:reg-nce}
\end{table}

\pagebreak

\begin{table}[!htbp]
    \centering
    \caption{Logistic regression of having at least one NCO conditions on COVID-19 vaccination, the NCE (previous immunization) and other  covariates.}
    \begin{tabular}{l|rrr}
       \hline
         & Est. & S.E. & p-value \\
         \hline
    (Intercept)   &  -2.42 &0.11 & $<$0.001 \\
    COVID-19 vaccination   &  0.23 & 0.09 & 0.014\\
    Previous Immunization & 0.68 &0.02& $<$0.001\\
    Age $\geq$18, $\le$60 &0.47& 0.06& $<$0.001\\ 
    Age $\geq$60 & 1.00 & 0.06 & $<$0.001\\
    Male & -0.05 & 0.02 & 0.049\\
    White & 0.13 & 0.03 & $<$0.001\\
    Charlson score $\geq$3 & 0.29& 0.04& $<$0.001\\
    Calendar month &&&\\
    $\quad$April& -0.04 & 0.05& $<$0.472 \\
    $\quad$May & 0.12 & 0.04 & 0.004 \\
    $\quad$June & 0.13 & 0.04& 0.002 \\
    $\quad$July & 0.12 & 0.05 & 0.013\\
    $\quad$August & 0.10 & 0.04 & 0.022\\
    $\quad$September & -0.01 & 0.04 & 0.712 \\
    $\quad$October & 0.040 & 0.04 & 0.312\\
    \hline
    \end{tabular}
    \label{tab:reg-nco}
\end{table}

\pagebreak

\begin{table}[!htbp]
    \centering
    \caption{Estimates, standard errors, and 95\% confidence intervals of parameters for the NC analysis with the University of Michigan Health System data by Algorithm~\ref{alg:nc-tnd}.}
    \begin{tabular}{l|rrr}
    \hline
         & Est. & S.E. & 95\% CI \\
         \hline
    $\beta$ (log causal RR)  & -2.80 & 0.14& (-3.08, -2.54)\\
    VE ($1-\exp(\beta)$)  & 94.0\% & / & (92.1\%, 95.4\%) \\
    $\tau$: Intercept & -29.14 & 53.86 & (-134.70,76.42)\\
    $\tau$: COVID-19 vaccination & 30.48 & 53.85 & (-75.06, -136.02)\\
    $\tau$: Previous immunization & 234.71 & 90.11 & (58.09, 411.32) \\
    $\tau$: COVID-19 vaccination $\times$ Previous immunization  &-235.76& 90.18& (-412.51, -59.01)\\ 
     $\tau$: Age $\geq$18, $\le$60 &-0.04& 0.18& (-0.40, 0.31)\\
    $\tau$: Age $\geq$60 & 0.15 & 0.18 & (-0.19, 0.50)\\
    $\tau$: Male & -0.04 & 0.08 & (-0.20,0.12)\\
    $\tau$: White & -0.07 & 0.09 & (-0.25, 0.11)\\
    $\tau$: Charlson score $\geq$3 & 0.31& 0.15& (0.01, 0.60)\\
    $\tau$: Calendar month &&&\\
    $\qquad$April& -0.52 & 0.23& (-0.98, -0.06) \\
    $\qquad$May & 0.24 & 0.15 & (-0.06, 0.54) \\
    $\qquad$June & 0.35 & 0.14& (0.08, 0.62) \\
    $\qquad$July & 0.31 & 0.15 & (0.02, 0.60)\\
    $\qquad$August & 0.25 & 0.15 & (-0.04, 0.55)\\
    $\qquad$September & 0.38 & 0.13 & (0.12, 0.64) \\
    $\qquad$October & 0.22 & 0.13 & (-0.03, 0.48)\\
    \hline
    \end{tabular}
    \label{tab:umhs-nc-results}
\end{table}

\pagebreak

\begin{table}[!htbp]
    \centering
    \caption{Estimates, standard errors and 95\% confidence intervals of parameters in a logistic regression of COVID-19 infection on COVID-19 vaccination and other baseline covariates.}
    \begin{tabular}{l|rrr}
    \hline
         & Est. & S.E. & 95\% CI\\
         \hline
    (Intercept)   & 0.36 & 0.12 & (0.12, 0.60) \\
    COVID-19 vaccination   & -3.18 & 0.08 & (-3.35, -3.01) \\
  Age $\geq$18, $\le$60 &0.56& 0.09& (0.39, 0.73)\\
    Age $\geq$60 & 0.32 & 0.09 & (0.14, 0.51)\\
    Male & 0.05 & 0.04 & (-0.03,0.13)\\
    White & 0.24 & 0.05 & (0.14, 0.33)\\
    Charlson score $\geq$3 & 0.25& 0.07& (0.11, 0.39)\\
     Calendar month &&&\\
    $\quad$April& -0.48 & 0.08& (-0.63, -0.33) \\
    $\quad$May & -1.40 & 0.09 & (-1.58, -1.21) \\
    $\quad$June & -2.23 & 0.14& (-2.49, -1.69) \\
    $\quad$July & -1.23 & 0.10 & (-1.42, -1.03)\\
    $\quad$August & -0.51 & 0.07 & (-0.65, -0.37)\\
    $\quad$September & -0.30 & 0.06 & (-0.42, -0.19) \\
    $\quad$October & -0.06 & 0.05 & (-0.17, 0.04)\\
    \hline
    \end{tabular}
    \label{tab:logit-reg-umhs}
\end{table}

\pagebreak
\printbibliography

@article{jewell2019analysis,
  title={Analysis of cluster-randomized test-negative designs: cluster-level methods},
  author={Jewell, Nicholas P and Dufault, Suzanne and Cutcher, Zoe and Simmons, Cameron P and Anders, Katherine L},
  journal={Biostatistics},
  volume={20},
  number={2},
  pages={332--346},
  year={2019},
  publisher={Oxford University Press}
}

@article{shi2020multiply,
  title={Multiply robust causal inference with double-negative control adjustment for categorical unmeasured confounding},
  author={Shi, Xu and Miao, Wang and Nelson, Jennifer C and Tchetgen Tchetgen, Eric J},
  journal={Journal of the Royal Statistical Society: Series B (Statistical Methodology)},
  volume={82},
  number={2},
  pages={521--540},
  year={2020},
  publisher={Wiley Online Library}
}

@article{shi2020selective,
  title={A selective review of negative control methods in epidemiology},
  author={Shi, Xu and Miao, Wang and Tchetgen Tchetgen, Eric},
  journal={Current Epidemiology Reports},
  pages={1--13},
  year={2020},
  publisher={Springer}
}

@article{miao2018confounding,
  title={A confounding bridge approach for double negative control inference on causal effects},
  author={Miao, Wang and Shi, Xu and Tchetgen Tchetgen, Eric},
  journal={arXiv e-prints},
  pages={arXiv--1808},
  year={2018}
}

@article{miao2018identifying,
  title={Identifying causal effects with proxy variables of an unmeasured confounder},
  author={Miao, Wang and Geng, Zhi and Tchetgen Tchetgen, Eric J},
  journal={Biometrika},
  volume={105},
  number={4},
  pages={987--993},
  year={2018},
  publisher={Oxford University Press}
}

@article{tchetgen2020introduction,
  title={An Introduction to Proximal Causal Learning},
  author={Tchetgen Tchetgen, Eric J and Ying, Andrew and Cui, Yifan and Shi, Xu and Miao, Wang},
  journal={arXiv preprint arXiv:2009.10982},
  year={2020}
}

@article{sullivan2016theoretical,
  title={Theoretical basis of the test-negative study design for assessment of influenza vaccine effectiveness},
  author={Sullivan, Sheena G and Tchetgen Tchetgen, Eric J and Cowling, Benjamin J},
  journal={American journal of epidemiology},
  volume={184},
  number={5},
  pages={345--353},
  year={2016},
  publisher={Oxford University Press}
}

@article{jackson2013test,
  title={The test-negative design for estimating influenza vaccine effectiveness},
  author={Jackson, Michael L and Nelson, Jennifer C},
  journal={Vaccine},
  volume={31},
  number={17},
  pages={2165--2168},
  year={2013},
  publisher={Elsevier}
}

@article{cui2020semiparametric,
  title={Semiparametric proximal causal inference},
  author={Cui, Yifan and Pu, Hongming and Shi, Xu and Miao, Wang and Tchetgen Tchetgen, Eric},
  journal={arXiv preprint arXiv:2011.08411},
  year={2020}
}

@article{qi2021proximal,
  title={Proximal Learning for Individualized Treatment Regimes Under Unmeasured Confounding},
  author={Qi, Zhengling and Miao, Rui and Zhang, Xiaoke},
  journal={arXiv preprint arXiv:2105.01187},
  year={2021}
}

@article{liu2021regression,
  title={Regression-based negative control of homophily in dyadic peer effect analysis},
  author={Liu, Lan and Tchetgen Tchetgen, Eric},
  journal={Biometrics},
  year={2021},
  publisher={Wiley Online Library}
}

@article{egami2021identification,
  title={Identification and Estimation of Causal Peer Effects Using Double Negative Controls for Unmeasured Network Confounding},
  author={Egami, Naoki and Tchetgen Tchetgen, Eric J},
  journal={arXiv preprint arXiv:2109.01933},
  year={2021}
}

@article{kallus2021causal,
  title={Causal Inference Under Unmeasured Confounding With Negative Controls: A Minimax Learning Approach},
  author={Kallus, Nathan and Mao, Xiaojie and Uehara, Masatoshi},
  journal={arXiv preprint arXiv:2103.14029},
  year={2021}
}

@article{imbens2021controlling,
  title={Controlling for Unmeasured Confounding in Panel Data Using Minimal Bridge Functions: From Two-Way Fixed Effects to Factor Models},
  author={Imbens, Guido and Kallus, Nathan and Mao, Xiaojie},
  journal={arXiv preprint arXiv:2108.03849},
  year={2021}
}

@article{gabriel2020causal,
  title={Causal bounds for outcome-dependent sampling in observational studies},
  author={Gabriel, Erin E and Sachs, Michael C and Sj{\"o}lander, Arvid},
  journal={Journal of the American Statistical Association},
  pages={1--12},
  year={2020},
  publisher={Taylor \& Francis}
}

@article{cai2012identifying,
  title={On identifying total effects in the presence of latent variables and selection bias},
  author={Cai, Zhihong and Kuroki, Manabu},
  journal={arXiv preprint arXiv:1206.3239},
  year={2012}
}

@book{van2000asymptotic,
  title={Asymptotic statistics},
  author={{Van der Vaart}, Aad.},
  volume={3},
  year={2000},
  publisher={Cambridge university press}
}

@article{bond2016regression,
  title={Regression approaches in the test-negative study design for assessment of influenza vaccine effectiveness},
  author={Bond, HS and Sullivan, SG and Cowling, BJ},
  journal={Epidemiology \& Infection},
  volume={144},
  number={8},
  pages={1601--1611},
  year={2016},
  publisher={Cambridge University Press}
}

@misc{dean2021covid,
  title={Covid-19 vaccine effectiveness and the test-negative design},
  author={Dean, Natalie E and Hogan, Joseph W and Schnitzer, Mireille E},
  year={2021},
  publisher={Mass Medical Soc}
}

@article{patel2020postlicensure,
  title={Postlicensure evaluation of COVID-19 vaccines},
  author={Patel, Manish M and Jackson, Michael L and Ferdinands, Jill},
  journal={JAMA},
  volume={324},
  number={19},
  pages={1939--1940},
  year={2020},
  publisher={American Medical Association}
}

@article{hitchings2021effectiveness,
  title={Effectiveness of the ChAdOx1 vaccine in the elderly during SARS-CoV-2 Gamma variant transmission in Brazil},
  author={Hitchings, Matt and Ranzani, Otavio T and Dorion, Murilo and D'Agostini, Tatiana Lang and de Paula, Regiane Cardoso and de Paula, Olivia Ferreira Pereira and de Moura Villela, Edlaine Faria and Torres, Mario Sergio Scaramuzzini and de Oliveira, Silvano Barbosa and Schulz, Wade and others},
  journal={medRxiv},
  year={2021},
  publisher={Cold Spring Harbor Laboratory Press}
}

@article{thompson2021effectiveness,
  title={Effectiveness of COVID-19 vaccines in ambulatory and inpatient care settings},
  author={Thompson, Mark G and Stenehjem, Edward and Grannis, Shaun and Ball, Sarah W and Naleway, Allison L and Ong, Toan C and DeSilva, Malini B and Natarajan, Karthik and Bozio, Catherine H and Lewis, Ned and others},
  journal={New England Journal of Medicine},
  volume={385},
  number={15},
  pages={1355--1371},
  year={2021},
  publisher={Mass Medical Soc}
}

@inproceedings{chambers2018should,
  title={Should sex be considered an effect modifier in the evaluation of influenza vaccine effectiveness?},
  author={Chambers, Catharine and Skowronski, Danuta M and Rose, Caren and Serres, Gaston De and Winter, Anne-Luise and Dickinson, James A and Jassem, Agatha and Gubbay, Jonathan B and Fonseca, Kevin and Drews, Steven J and others},
  booktitle={Open forum infectious diseases},
  volume={5},
  number={9},
  pages={ofy211},
  year={2018},
  organization={Oxford University Press US}
}

@article{lipsitch2016observational,
  title={Observational studies and the difficult quest for causality: lessons from vaccine effectiveness and impact studies},
  author={Lipsitch, Marc and Jha, Ayan and Simonsen, Lone},
  journal={International journal of epidemiology},
  volume={45},
  number={6},
  pages={2060--2074},
  year={2016},
  publisher={Oxford University Press}
}

@article{lipsitch2010negative,
  title={Negative controls: a tool for detecting confounding and bias in observational studies},
  author={Lipsitch, Marc and Tchetgen Tchetgen, Eric and Cohen, Ted},
  journal={Epidemiology (Cambridge, Mass.)},
  volume={21},
  number={3},
  pages={383},
  year={2010},
  publisher={NIH Public Access}
}

@article{dukes2021proximal,
  title={Proximal mediation analysis},
  author={Dukes, Oliver and Shpitser, Ilya and Tchetgen Tchetgen, Eric J},
  journal={arXiv preprint arXiv:2109.11904},
  year={2021}
}

@article{ying2021proximal,
  title={Proximal Causal Inference for Complex Longitudinal Studies},
  author={Ying, Andrew and Miao, Wang and Shi, Xu and Tchetgen Tchetgen, Eric J},
  journal={arXiv preprint arXiv:2109.07030},
  year={2021}
}

@article{jackson2017influenza,
  title={Influenza vaccine effectiveness in the United States during the 2015--2016 season},
  author={Jackson, Michael L and Chung, Jessie R and Jackson, Lisa A and Phillips, C Hallie and Benoit, Joyce and Monto, Arnold S and Martin, Emily T and Belongia, Edward A and McLean, Huong Q and Gaglani, Manjusha and others},
  journal={New England Journal of Medicine},
  volume={377},
  number={6},
  pages={534--543},
  year={2017},
  publisher={Mass Medical Soc}
}

@article{flannery2019influenza,
  title={Influenza vaccine effectiveness in the United States during the 2016--2017 season},
  author={Flannery, Brendan and Chung, Jessie R and Monto, Arnold S and Martin, Emily T and Belongia, Edward A and McLean, Huong Q and Gaglani, Manjusha and Murthy, Kempapura and Zimmerman, Richard K and Nowalk, Mary Patricia and others},
  journal={Clinical Infectious Diseases},
  volume={68},
  number={11},
  pages={1798--1806},
  year={2019},
  publisher={Oxford University Press US}
}

@article{rolfes2019effects,
  title={Effects of influenza vaccination in the United States during the 2017--2018 influenza season},
  author={Rolfes, Melissa A and Flannery, Brendan and Chung, Jessie R and O’Halloran, Alissa and Garg, Shikha and Belongia, Edward A and Gaglani, Manjusha and Zimmerman, Richard K and Jackson, Michael L and Monto, Arnold S and others},
  journal={Clinical Infectious Diseases},
  volume={69},
  number={11},
  pages={1845--1853},
  year={2019},
  publisher={Oxford University Press US}
}

@article{chung2020effects,
  title={Effects of influenza vaccination in the United States during the 2018--2019 influenza season},
  author={Chung, Jessie R and Rolfes, Melissa A and Flannery, Brendan and Prasad, Pragati and O’Halloran, Alissa and Garg, Shikha and Fry, Alicia M and Singleton, James A and Patel, Manish and Reed, Carrie and others},
  journal={Clinical Infectious Diseases},
  volume={71},
  number={8},
  pages={e368--e376},
  year={2020},
  publisher={Oxford University Press US}
}

@article{tenforde2021influenza,
  title={Influenza vaccine effectiveness against hospitalization in the United States, 2019--2020},
  author={Tenforde, Mark W and Talbot, H Keipp and Trabue, Christopher H and Gaglani, Manjusha and McNeal, Tresa M and Monto, Arnold S and Martin, Emily T and Zimmerman, Richard K and Silveira, Fernanda P and Middleton, Donald B and others},
  journal={The Journal of Infectious Diseases},
  volume={224},
  number={5},
  pages={813--820},
  year={2021},
  publisher={Oxford University Press US}
}

@article{broome1980pneumococcal,
  title={Pneumococcal disease after pneumococcal vaccination: an alternative method to estimate the efficacy of pneumococcal vaccine},
  author={Broome, Claire V and Facklam, Richard R and Fraser, David W},
  journal={New England Journal of Medicine},
  volume={303},
  number={10},
  pages={549--552},
  year={1980},
  publisher={Mass Medical Soc}
}

@article{anders2018awed,
  title={The AWED trial (Applying Wolbachia to Eliminate Dengue) to assess the efficacy of Wolbachia-infected mosquito deployments to reduce dengue incidence in Yogyakarta, Indonesia: study protocol for a cluster randomised controlled trial},
  author={Anders, Katherine L and Indriani, Citra and Ahmad, Riris Andono and Tantowijoyo, Warsito and Arguni, Eggi and Andari, Bekti and Jewell, Nicholas P and Rances, Edwige and O’Neill, Scott L and Simmons, Cameron P and others},
  journal={Trials},
  volume={19},
  number={1},
  pages={1--16},
  year={2018},
  publisher={BioMed Central}
}

@article{utarini2021efficacy,
  title={Efficacy of Wolbachia-infected mosquito deployments for the control of dengue},
  author={Utarini, Adi and Indriani, Citra and Ahmad, Riris A and Tantowijoyo, Warsito and Arguni, Eggi and Ansari, M Ridwan and Supriyati, Endah and Wardana, D Satria and Meitika, Yeti and Ernesia, Inggrid and others},
  journal={New England Journal of Medicine},
  volume={384},
  number={23},
  pages={2177--2186},
  year={2021},
  publisher={Mass Medical Soc}
}

@article{schwartz2017rotavirus,
  title={Rotavirus vaccine effectiveness in low-income settings: An evaluation of the test-negative design},
  author={Schwartz, Lauren M and Halloran, M Elizabeth and Rowhani-Rahbar, Ali and Neuzil, Kathleen M and Victor, John C},
  journal={Vaccine},
  volume={35},
  number={1},
  pages={184--190},
  year={2017},
  publisher={Elsevier}
}

@article{boom2010effectiveness,
  title={Effectiveness of pentavalent rotavirus vaccine in a large urban population in the United States},
  author={Boom, Julie A and Tate, Jacqueline E and Sahni, Leila C and Rench, Marcia A and Hull, Jennifer J and Gentsch, Jon R and Patel, Manish M and Baker, Carol J and Parashar, Umesh D},
  journal={Pediatrics},
  volume={125},
  number={2},
  pages={e199--e207},
  year={2010},
  publisher={Am Acad Pediatrics}
}

@article{jackson2006evidence,
  title={Evidence of bias in estimates of influenza vaccine effectiveness in seniors},
  author={Jackson, Lisa A and Jackson, Michael L and Nelson, Jennifer C and Neuzil, Kathleen M and Weiss, Noel S},
  journal={International journal of epidemiology},
  volume={35},
  number={2},
  pages={337--344},
  year={2006},
  publisher={Oxford University Press}
}

@article{rosenbaum1987model,
  title={Model-based direct adjustment},
  author={Rosenbaum, Paul R},
  journal={Journal of the American Statistical Association},
  volume={82},
  number={398},
  pages={387--394},
  year={1987},
  publisher={Taylor \& Francis}
}

@article{newey2003instrumental,
  title={Instrumental variable estimation of nonparametric models},
  author={Newey, Whitney K and Powell, James L},
  journal={Econometrica},
  volume={71},
  number={5},
  pages={1565--1578},
  year={2003},
  publisher={Wiley Online Library}
}

@article{hu2008instrumental,
  title={Instrumental variable treatment of nonclassical measurement error models},
  author={Hu, Yingyao and Schennach, Susanne M},
  journal={Econometrica},
  volume={76},
  number={1},
  pages={195--216},
  year={2008},
  publisher={Wiley Online Library}
}

@article{cole2009consistency,
  title={The consistency statement in causal inference: a definition or an assumption?},
  author={Cole, Stephen R and Frangakis, Constantine E},
  journal={Epidemiology},
  volume={20},
  number={1},
  pages={3--5},
  year={2009},
  publisher={LWW}
}

@article{struchiner2007randomization,
  title={Randomization and baseline transmission in vaccine field trials},
  author={Struchiner, CJ and Halloran, ME},
  journal={Epidemiology \& Infection},
  volume={135},
  number={2},
  pages={181--194},
  year={2007},
  publisher={Cambridge University Press}
}

@article{hudgens2006causal,
  title={Causal vaccine effects on binary postinfection outcomes},
  author={Hudgens, Michael G and Halloran, M Elizabeth},
  journal={Journal of the American Statistical Association},
  volume={101},
  number={473},
  pages={51--64},
  year={2006},
  publisher={Taylor \& Francis}
}

@article{shi2021theory,
  title={Theory for identification and Inference with Synthetic Controls: A Proximal Causal Inference Framework},
  author={Shi, Xu and Miao, Wang and Hu, Mengtong and Tchetgen Tchetgen, Eric},
  journal={arXiv preprint arXiv:2108.13935},
  year={2021}
}

@incollection{lehmann2012completeness1,
  title={Completeness, similar regions, and unbiased estimation-Part I},
  author={Lehmann, Erich Leo and Scheff{\'e}, Henry},
  booktitle={Selected Works of EL Lehmann},
  pages={233--268},
  year={2012},
  publisher={Springer}
}

@incollection{lehmann2012completeness2,
  title={Completeness, similar regions, and unbiased estimation—part II},
  author={Lehmann, Erich Leo and Scheff{\'e}, Henry},
  booktitle={Selected Works of EL Lehmann},
  pages={269--286},
  year={2012},
  publisher={Springer}
}

@article{d2011completeness,
  title={On the completeness condition in nonparametric instrumental problems},
  author={D’Haultfoeuille, Xavier},
  journal={Econometric Theory},
  volume={27},
  number={3},
  pages={460--471},
  year={2011},
  publisher={Cambridge University Press}
}

@article{andrews2011examples,
  title={Examples of l2-complete and boundedly-complete distributions},
  author={Andrews, Donald WK},
  year={2011},
  publisher={Cowles Foundation Discussion Paper}
}

@article{shrank2011healthy,
  title={Healthy user and related biases in observational studies of preventive interventions: a primer for physicians},
  author={Shrank, William H and Patrick, Amanda R and Brookhart, M Alan},
  journal={Journal of general internal medicine},
  volume={26},
  number={5},
  pages={546--550},
  year={2011},
  publisher={Springer}
}

@article{black2018influenza,
  title={Influenza vaccination coverage among health care personnel—United States, 2017--18 influenza season},
  author={Black, Carla L and Yue, Xin and Ball, Sarah W and Fink, Rebecca V and de Perio, Marie A and Laney, A Scott and Williams, Walter W and Graitcer, Samuel B and Fiebelkorn, Amy Parker and Lu, Peng-Jun and others},
  journal={Morbidity and Mortality Weekly Report},
  volume={67},
  number={38},
  pages={1050},
  year={2018},
  publisher={Centers for Disease Control and Prevention}
}

@article{krammer2019human,
  title={The human antibody response to influenza A virus infection and vaccination},
  author={Krammer, Florian},
  journal={Nature Reviews Immunology},
  volume={19},
  number={6},
  pages={383--397},
  year={2019},
  publisher={Nature Publishing Group}
}

@article{foppa2016case,
  title={The case test-negative design for studies of the effectiveness of influenza vaccine in inpatient settings},
  author={Foppa, Ivo M and Ferdinands, Jill M and Chaves, Sandra S and Haber, Michael J and Reynolds, Sue B and Flannery, Brendan and Fry, Alicia M},
  journal={International journal of epidemiology},
  volume={45},
  number={6},
  pages={2052--2059},
  year={2016},
  publisher={Oxford University Press}
}

@article{feng2016influenza,
  title={Influenza vaccine effectiveness by test-negative design--Comparison of inpatient and outpatient settings},
  author={Feng, Shuo and Cowling, Benjamin J and Sullivan, Sheena G},
  journal={Vaccine},
  volume={34},
  number={14},
  pages={1672--1679},
  year={2016},
  publisher={Elsevier}
}

@inproceedings{bareinboim2012controlling,
  title={Controlling selection bias in causal inference},
  author={Bareinboim, Elias and Pearl, Judea},
  booktitle={Artificial Intelligence and Statistics},
  pages={100--108},
  year={2012},
  organization={PMLR}
}

@article{huitfeldt2019collapsibility,
  title={On the collapsibility of measures of effect in the counterfactual causal framework},
  author={Huitfeldt, Anders and Stensrud, Mats J and Suzuki, Etsuji},
  journal={Emerging themes in epidemiology},
  volume={16},
  number={1},
  pages={1--5},
  year={2019},
  publisher={Springer}
}

@article{leung2011herpes,
  title={Herpes zoster incidence among insured persons in the United States, 1993--2006: evaluation of impact of varicella vaccination},
  author={Leung, Jessica and Harpaz, Rafael and Molinari, Noelle-Angelique and Jumaan, Aisha and Zhou, Fangjun},
  journal={Clinical Infectious Diseases},
  volume={52},
  number={3},
  pages={332--340},
  year={2011},
  publisher={The University of Chicago Press}
}

@article{moline2021effectiveness,
  title={Effectiveness of COVID-19 vaccines in preventing hospitalization among adults aged≥ 65 years—COVID-NET, 13 states, February--April 2021},
  author={Moline, Heidi L and Whitaker, Michael and Deng, Li and Rhodes, Julia C and Milucky, Jennifer and Pham, Huong and Patel, Kadam and Anglin, Onika and Reingold, Arthur and Chai, Shua J and others},
  journal={Morbidity and Mortality Weekly Report},
  volume={70},
  number={32},
  pages={1088},
  year={2021},
  publisher={Centers for Disease Control and Prevention}
}

@article{olson2022effectiveness,
  title={Effectiveness of BNT162b2 vaccine against critical Covid-19 in adolescents},
  author={Olson, Samantha M and Newhams, Margaret M and Halasa, Natasha B and Price, Ashley M and Boom, Julie A and Sahni, Leila C and Pannaraj, Pia S and Irby, Katherine and Walker, Tracie C and Schwartz, Stephanie P and others},
  journal={New England Journal of Medicine},
  year={2022},
  publisher={Mass Medical Soc}
}

@article{dagan2021bnt162b2,
  title={BNT162b2 mRNA Covid-19 vaccine in a nationwide mass vaccination setting},
  author={Dagan, Noa and Barda, Noam and Kepten, Eldad and Miron, Oren and Perchik, Shay and Katz, Mark A and Hern{\'a}n, Miguel A and Lipsitch, Marc and Reis, Ben and Balicer, Ran D},
  journal={New England Journal of Medicine},
  year={2021},
  publisher={Mass Medical Soc}
}

@article{deaner2018proxy,
  title={Proxy controls and panel data},
  author={Deaner, Ben},
  journal={arXiv preprint arXiv:1810.00283},
  year={2018}
}

@article{deaner2021many,
  title={Many Proxy Controls},
  author={Deaner, Ben},
  journal={arXiv preprint arXiv:2110.03973},
  year={2021}
}

@article{fernandez2021effect,
  title={Effect modification of the association between comorbidities and severe course of COVID-19 disease by age of study participants: a systematic review and meta-analysis},
  author={Fern{\'a}ndez Villalobos, Nathalie Ver{\'o}nica and Ott, J{\"o}rdis Jennifer and Klett-Tammen, Carolina Judith and Bockey, Annabelle and Vanella, Patrizio and Krause, G{\'e}rard and Lange, Berit},
  journal={Systematic reviews},
  volume={10},
  number={1},
  pages={1--15},
  year={2021},
  publisher={Springer}
}

@article{baden2021efficacy,
  title={Efficacy and safety of the mRNA-1273 SARS-CoV-2 vaccine},
  author={Baden, Lindsey R and El Sahly, Hana M and Essink, Brandon and Kotloff, Karen and Frey, Sharon and Novak, Rick and Diemert, David and Spector, Stephen A and Rouphael, Nadine and Creech, C Buddy and others},
  journal={New England Journal of Medicine},
  volume={384},
  number={5},
  pages={403--416},
  year={2021},
  publisher={Mass Medical Soc}
}

@article{sadoff2021safety,
  title={Safety and efficacy of single-dose Ad26. COV2. S vaccine against Covid-19},
  author={Sadoff, Jerald and Gray, Glenda and Vandebosch, An and C{\'a}rdenas, Vicky and Shukarev, Georgi and Grinsztejn, Beatriz and Goepfert, Paul A and Truyers, Carla and Fennema, Hein and Spiessens, Bart and others},
  journal={New England Journal of Medicine},
  volume={384},
  number={23},
  pages={2187--2201},
  year={2021},
  publisher={Mass Medical Soc}
}

@article{polack2020safety,
  title={Safety and efficacy of the BNT162b2 mRNA Covid-19 vaccine},
  author={Polack, Fernando P and Thomas, Stephen J and Kitchin, Nicholas and Absalon, Judith and Gurtman, Alejandra and Lockhart, Stephen and Perez, John L and Marc, Gonzalo P{\'e}rez and Moreira, Edson D and Zerbini, Cristiano and others},
  journal={New England Journal of Medicine},
  year={2020},
  publisher={Mass Medical Soc}
}

@article{ghassami2022combining,
  title={Combining Experimental and Observational Data for Identification of Long-Term Causal Effects},
  author={Ghassami, AmirEmad and Shpitser, Ilya and  Tchetgen Tchetgen, Eric},
  journal={arXiv preprint arXiv:2201.10743},
  year={2022}
}

@article{ghassami2021proximal,
  title={Proximal Causal Inference with Hidden Mediators: Front-Door and Related Mediation Problems},
  author={Ghassami, AmirEmad and Shpitser, Ilya and Tchetgen Tchetgen, Eric },
  journal={arXiv preprint arXiv:2111.02927},
  year={2021}
}

@techreport{ghassami2021minimax,
  title={Minimax kernel machine learning for a class of doubly robust functionals},
  author={Ghassami, AmirEmad and Ying, Andrew and Shpitser, Ilya and Tchetgen Tchetgen, Eric and others},
  year={2021}
}

@inproceedings{mastouri2021proximal,
  title={Proximal causal learning with kernels: Two-stage estimation and moment restriction},
  author={Mastouri, Afsaneh and Zhu, Yuchen and Gultchin, Limor and Korba, Anna and Silva, Ricardo and Kusner, Matt and Gretton, Arthur and Muandet, Krikamol},
  booktitle={International Conference on Machine Learning},
  pages={7512--7523},
  year={2021},
  organization={PMLR}
}

@article{chen2003note,
  title={A note on the prospective analysis of outcome-dependent samples},
  author={Chen, Hua Yun},
  journal={Journal of the Royal Statistical Society: Series B (Statistical Methodology)},
  volume={65},
  number={2},
  pages={575--584},
  year={2003},
  publisher={Wiley Online Library}
}

@article{wang2022randomization,
  title={Randomization Inference for Cluster-Randomized Test-Negative Designs with Application to Dengue Studies: Unbiased estimation, Partial compliance, and Stepped-wedge design},
  author={Wang, Bingkai and Dufault, Suzanne M and Small, Dylan S and Jewell, Nicholas P},
  journal={arXiv preprint arXiv:2202.03379},
  year={2022}
}

@article{dufault2020analysis,
  title={Analysis of counts for cluster randomized trials: Negative controls and test-negative designs},
  author={Dufault, Suzanne M and Jewell, Nicholas P},
  journal={Statistics in medicine},
  volume={39},
  number={10},
  pages={1429--1439},
  year={2020},
  publisher={Wiley Online Library}
}

@article{hausman1978specification,
  title={Specification tests in econometrics},
  author={Hausman, Jerry A},
  journal={Econometrica: Journal of the econometric society},
  pages={1251--1271},
  year={1978},
  publisher={JSTOR}
}

@article{schnitzer2022estimands,
  title={Estimands and Estimation of COVID-19 Vaccine Effectiveness Under the Test-negative Design: Connections to Causal Inference.},
  author={Schnitzer, Mireille E},
  journal={Epidemiology (Cambridge, Mass.)},
  year={2022}
}

@article{hudgens2008toward,
  title={Toward causal inference with interference},
  author={Hudgens, Michael G and Halloran, M Elizabeth},
  journal={Journal of the American Statistical Association},
  volume={103},
  number={482},
  pages={832--842},
  year={2008},
  publisher={Taylor \& Francis}
}

@article{tchetgen2012causal,
  title={On causal inference in the presence of interference},
  author={{Tchetgen Tchetgen}, Eric J  and VanderWeele, Tyler J},
  journal={Statistical methods in medical research},
  volume={21},
  number={1},
  pages={55--75},
  year={2012},
  publisher={Sage Publications Sage UK: London, England}
}
\end{document}